\documentclass[12pt]{article}
\usepackage{amsmath}
\usepackage{graphicx}
\usepackage{enumerate}
\usepackage{url} 
\usepackage{ulem}

\usepackage{authblk}
\usepackage[utf8]{inputenc}
\usepackage{longtable}
\usepackage{booktabs}
\usepackage{multirow}
\usepackage{bm}
\usepackage{amsmath, amsfonts, amsthm, amssymb, amsbsy, graphicx, dsfont, mathtools, enumerate, enumitem}
\usepackage{graphicx}
\usepackage[english]{babel}
\usepackage{graphicx}
\usepackage{subfigure, caption}
\usepackage{natbib}
\usepackage{multibib}
\setcitestyle{authoryear,open={(},close={)}}
\newcites{response}{References}
\usepackage[CJKbookmarks=true,
            bookmarksnumbered=true,
			bookmarksopen=true,
			colorlinks=true,
			citecolor=blue, 
			linkcolor=blue,
			anchorcolor=red,
			urlcolor=blue]{hyperref}
\usepackage[usenames]{color}

\usepackage[letterpaper, left=1.2truein, right=1.2truein, top = 1.2truein, bottom = 1.2truein]{geometry}
\usepackage{prettyref,soul}
\usepackage[font=small,labelfont=it]{caption}
\RequirePackage[framemethod=TikZ]{mdframed}

\newtheorem{theorem}{Theorem}

\newtheorem{lemma}{Lemma}
\newtheorem{proposition}{Proposition}

\usepackage{apptools}
\AtAppendix{\counterwithin{lemma}{section}}
\AtAppendix{\counterwithin{theorem}{section}}
\AtAppendix{\counterwithin{corollary}{section}}
\AtAppendix{\counterwithin{proposition}{section}}

\theoremstyle{definition}

\theoremstyle{remark}
\newtheorem*{remark}{Remark}

\newcommand{\argmin}{\mathop{\rm argmin}}
\newcommand{\argmax}{\mathop{\rm argmax}}

\newcommand{\Prob}{\mathbb{P}}
\newcommand{\Expect}{\mathbb{E}}

\newcommand{\rank}{\mathop{\sf rank}}

\newcommand{\diag}{\mathop{\text{diag}}}
\newcommand{\supp}{{\rm supp}}

\newcommand{\D}{\ensuremath{\mathcal{D}}}

\newcommand{\M}{\mathcal{M}}
\newcommand{\Nm}{\mathcal{N}}

\newcommand{\Ws}{\ensuremath{W^{\mathrm{S}}}}
\newcommand{\Wbc}{\ensuremath{W^{\mathrm{BC}}}}
\newcommand{\Wst}{\ensuremath{W^{\mathrm{S} \tau}}}
\newcommand{\Wbct}{\ensuremath{W^{\mathrm{BC} \tau}}}
\newcommand{\s}{\ensuremath{\mathrm{S}}}
\newcommand{\T}{\ensuremath{^{\mathcal{T}}}}

\newcommand{\Zt}{\ensuremath{\tilde{Z}\T}}

\newcommand{\bc}{\ensuremath{\mathrm{BC}}}
\newcommand{\stau}{\ensuremath{\mathrm{S} \tau}}
\newcommand{\bct}{\ensuremath{\mathrm{BC} \tau}}

\newcommand{\Fs}{\ensuremath{\mathcal{F}}}
\newcommand{\Vs}{\ensuremath{V_s}}

\newcommand{\A}{\ensuremath{\mathrm{BC\star}}} 
\newcommand{\B}{\ensuremath{\mathrm{S\star}}} 
\newcommand{\upa}{\ensuremath{^{\A}}}
\newcommand{\upb}{\ensuremath{^{\B}}}

\newcommand{\C}{\star}
\newcommand{\upc}{\ensuremath{^{\C}}}

\def\D{\mathcal{D}}
\def\F{\mathcal{F}}
\def\E{\mathbb{E}}
\def\R{\mathbb{R}}
\def\S{S}
\def\diag{\mathrm{diag}}
\def\sign{\mathrm{sign}}
\def\vecs{s}

\begin{document}

\def\spacingset#1{\renewcommand{\baselinestretch}%
{#1}\small\normalsize} \spacingset{1}

\title{\bf Controlling the False Discovery Rate in Transformational Sparsity: Split Knockoffs}
\author{Yang Cao$^1$, Xinwei Sun$^2$\thanks{\url{sunxinwei@fudan.edu.cn}} and Yuan Yao$^1$\thanks{\url{yuany@ust.hk}}
    ~\\
    $^1$Hong Kong University of Science and Technology \\
    $^2$Fudan University
    }
    \date{}
\maketitle
\bigskip

\begin{abstract}
    Controlling the False Discovery Rate (FDR) in a variable selection procedure is critical for reproducible discoveries, and it has been extensively studied in sparse linear models. However, it remains largely open in scenarios where the sparsity constraint is not directly imposed on the parameters but on a linear transformation of the parameters to be estimated. Examples of such scenarios include total variations, wavelet transforms, fused LASSO, and trend filtering. In this paper, we propose a data-adaptive FDR control method, called the \emph{Split Knockoff} method, for this transformational sparsity setting. The proposed method exploits both variable and data splitting. The linear transformation constraint is relaxed to its Euclidean proximity in a lifted parameter space, which yields an orthogonal design that enables the orthogonal Split Knockoff construction. To overcome the challenge that exchangeability fails due to the heterogeneous noise brought by the transformation, new inverse supermartingale structures are developed via data splitting for provable FDR control without sacrificing power. Simulation experiments demonstrate that the proposed methodology achieves the desired FDR and power. We also provide an application to Alzheimer's Disease study, where atrophy brain regions and their abnormal connections can be discovered based on a structural Magnetic Resonance Imaging dataset (ADNI).
\end{abstract}

\noindent%
\vfill

\newpage

\tableofcontents

\newpage

\section{Introduction}
Variable selection or sparse model selection is a fundamental problem in statistical research. Equipped with the wide spread of modern data acquisition facilities, one can simultaneously measure a large number of covariates or features and it is desired to discover a relatively small amount of dominant factors governing the variations. In many scenarios, such a sparsity constraint does not always rely on a small number of measured covariates or features, but is about some transformations, often linear, of parameters. For instance, in signal processing such as images, sparsity of edges or jumps lies in wavelet transforms \citep{donoho1995adapting} or total variations \citep{ROF92,CDOS12}; in genomic studies of human cancer, sparsity in 1-D fused LASSO \citep{tibshirani2005sparsity} is associated with abnormality of copy numbers of genome orders in comparative genomic hybridization (CGH) data, a valuable way of understanding human cancer; in trend filtering \citep{kim2009ell_1}, sparsity lies in the change point detection of piece-wise linear time series; in statistical ranking, sparsity of graph gradients disclose the partial orders or candidate groups in ties \citep{SplitLBI,huang2020boosting}. 

In this paper, consider the following \emph{transformational sparsity} or  \emph{structural sparsity} problem in a linear regression where a linear transformation of parameters is sparse.
\begin{align}
  y = X\beta^*+\varepsilon,\
  \gamma^* = D\beta^*,\label{eq: model}
\end{align}
where $y\in \R^n$ is the response vector, $X\in \mathbb{R}^{n\times p}$ is the design matrix, $\beta^*\in \R^p$ is the unknown coefficient vector, $D\in \R^{m\times p}$ is the linear transformer, $\gamma^*\in \R^m$ is the sparse vector, 
and $\varepsilon\sim \Nm(0,\sigma^2 I_n)$ is Gaussian noise. 
Our purpose is to recover the support set of $\gamma^*$. For shorthand notations, we define the nonnull set $\S_1=\supp(\gamma^*)=\{i:\gamma^*_i\neq 0\}$, and the null set $\S_0=\{i:\gamma^*_i=0\}$. Note that if we take $m=p$ and $D=I_p$, this model is degenerated into the standard sparse linear regression. Hence model \eqref{eq: model} can be viewed as a generalization of the traditional sparse regression problem.

\paragraph*{Example.} Considering brain imaging data analysis for Alzheimer's Disease, $y$ represents the Alzheimer’s Disease Assessment Scale (ADAS) of patients, $X_{i,j}$ measures the gray matter volume of brain region $j$ in the cerebrum brain of subject $i$. Taking the identity matrix $D=I$ (where $m=p=90$), one searches for highly atrophy brain regions for AD; taking $D$ as the graph gradient operator on the brain region connectivity graph (where $m=463>p=90$), abnormal connections of brain regions due to disease progression are discovered. Sparsity associated with various linear transformations above discloses both important lesion regions that undergo severe damages in disease progression and highly differential connections that link stable regions to Hippocampus, one of the most important regions accounting for Alzheimer's Disease \citep{juottonen1999comparative}. Such discoveries based on the methodology in this paper are illustrated by Figure~\ref{fig: ad reg and con}, whose details will be discussed in Section~\ref{Sec: applications}.

\begin{figure}[!ht]
\centering
\includegraphics[width=\textwidth]{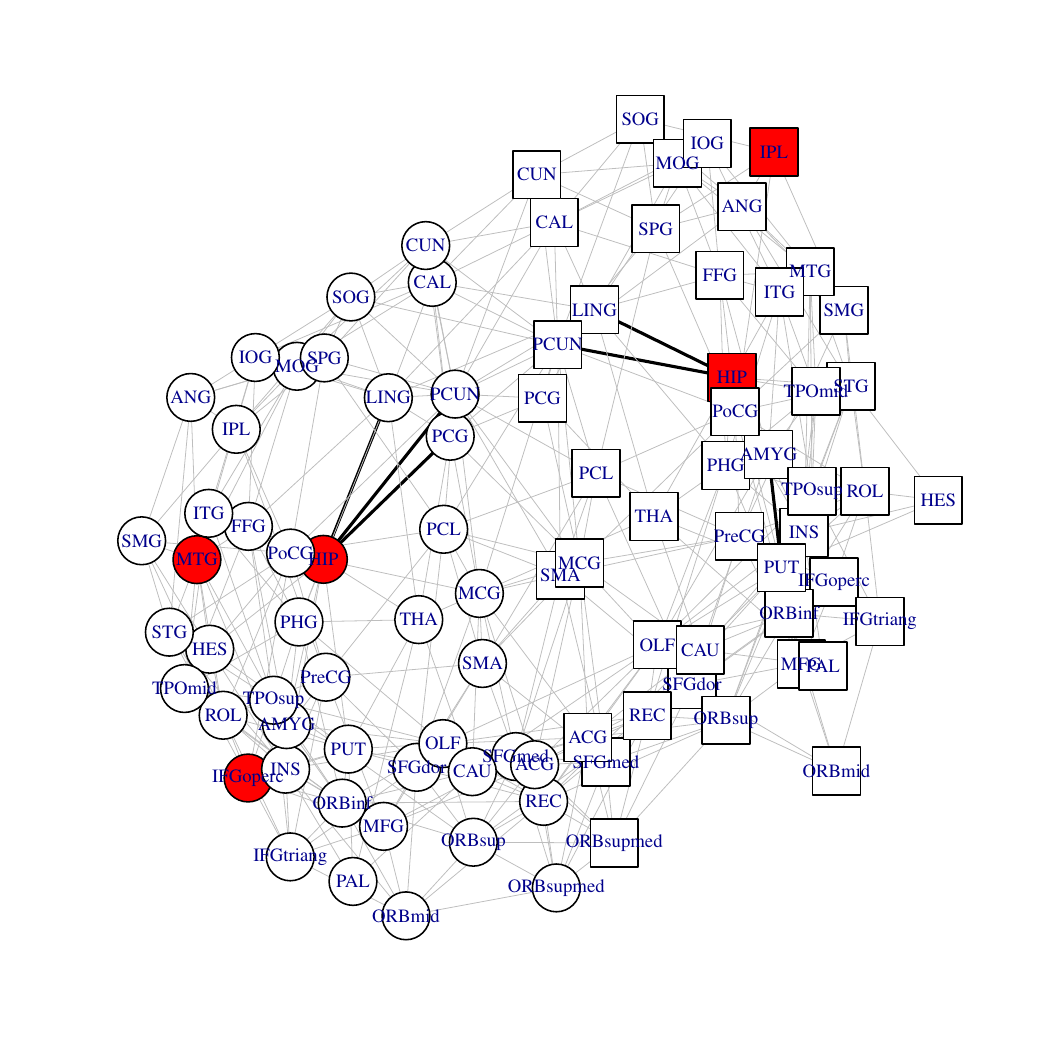}
\caption{Selected brain regions and connections in Alzheimer's Disease. Each vertex represents a Cerebrum brain region in Automatic Anatomical Labeling (AAL) atlas \citep{tzourio2002automated}, whose abbreviations and full names are provided in Table~\ref{tab:name of region}. Here, vertices with a circle shape represent the left brain regions, while the ones with a square shape represent the right brain regions. An edge connects two vertices if and only if the associated two brain regions are adjacent. The method in this paper selects red vertices as significantly degenerate lesion regions and bold edges for highly differential connections.} 
\label{fig: ad reg and con}
\end{figure}

 To evaluate the accuracy of an empirically discovered set $\hat{S}$ of non-null parameters, the false discovery rate (FDR) is the goal of this paper in favor of the reproducibility of discoveries. Formally speaking, FDR of a selection set $\hat{S}$ is defined as
    \begin{align}
    \label{eq: fdr}
        \mathrm{FDR} = & \E \left[\frac{|\{i:i\in \hat{S}\cap \S_0\}|}{|\hat{S}|\vee 1}\right].
    \end{align}
    
The problem of controlling the False Discovery Rate (FDR) has been widely studied since \cite{benjamini1995controlling}. In particular, \cite{barber2015controlling} recently proposed the knockoff method for sparse linear regression when $D=I_p$, with a theoretical guarantee of an upper bound of FDR. This work has been further applied to group sparse models, multi-task regression models \citep{dai2016knockoff}, Huber's robust regression for outlier detection \citep{xu2016false}, and high-dimensional scenarios \citep{barber2019knockoff}. In particular, \cite{candes2016panning} proposed Model-X knockoff for random design, and \citep{barber2020robust} showed that Model-X knockoff is robust to estimation error of random design parameters. Deep knockoff \citep{romano2019deep} has been developed for nonparametric random designs, and side information is considered in \cite{Ren20_side}. Derandomized Knockoffs have been proposed in \cite{Ren21_jasa}. However, for the general constraint $\gamma=D\beta$, it is not clear how to apply such knockoff methods to transformational sparsity, except for some special cases that can be reduced to sparse linear regression (see Section~\ref{sec:genlasso} for details).

To fill this gap, we propose a data-adaptive selection method for the transformational sparsity problem by developing a new method in the framework of knockoffs, called the \emph{Split Knockoff} method. In our approach, the linear submanifold constraint $\gamma=D\beta$ is relaxed into its Euclidean neighborhood with a proximity parameter $\nu$, which is known as the variable splitting technique in optimization. This leads to an orthogonal design as well as orthogonal knockoff copies. To overcome the challenge of exchangeability failure, the whole dataset is further split into two independent subsets, with $\beta$ being estimated on one and sparse $\gamma$ selected on the other with new designs of Split Knockoff statistics. Equipped with orthogonal Split Knockoff copies, the crucial new statistics design has signs as independent Bernoulli random variables, decoupled from magnitudes, that enables the inverse supermartingale inequalities to guarantee FDR control. It also enables us to handle some traditional types of knockoff statistics whose signs and magnitudes are dependent, via filtration refinement. The $\nu$-relaxation above can be used for power optimization as a trade-off between incoherence improvement and loss of weak signals.

The methodology is validated by simulation experiments and applied to the study of Alzheimer's Disease, where atrophy brain regions and their abnormal connections are successfully discovered.

\subsection{Organization of the paper}

\begin{itemize}

    \item In Section \ref{sec:splitknockoff}, we introduce the methodology of Split Knockoffs.
    \item In Section \ref{sec:mainresults}, the FDR of Split Knockoffs is shown under control, with an outline on how to achieve this by developing some new inverse supermartingale constructions.
    \item In Section \ref{sec: hd}, a high dimensional generalization of the methodology is discussed.
    \item In Section \ref{sec: discussion}, several important topics are discussed: (i) why Knockoffs with generalized LASSO fail in general settings of transformational sparsity; (ii) the main challenge in establishing theoretical FDR control of Split Knockoffs, the failure of exchangeability; (iii) how the $\nu$-relaxation affects the selection power in terms of the model selection consistency of Split LASSO regularization paths.
    \item In Section \ref{Sec: simulation results}, simulation experiments are conducted with a comparison of the performance of Split Knockoffs in terms of FDR and power.
    \item In Section \ref{Sec: applications}, Alzheimer's Disease is studied based on brain imaging data, where Split Knockoffs discovers the abnormal lesion regions and their connections in brains.
    \item In Section \ref{sec: conclusion}, conclusions and future directions are discussed.
\end{itemize}
Details for proofs, as well as additional results on the application of Alzheimer's disease are provided in supplementary material sections.

\section{The Split Knockoff Method}

\label{sec:splitknockoff}

Let's start with the generalized LASSO \citep{tibshirani2011solution}, the most popular method to handle the transformational sparsity. Recall that the generalized LASSO solves the following optimization problem with $\lambda>0$, 
\begin{equation} \label{eq:gen_lasso}
    \min_\beta\ \frac{1}{2n}\|y-X\beta\|_2^2+\lambda\|\gamma\|_1,\ \ \ \ \mbox{subject to $\gamma = D\beta$.} 
\end{equation}
However, the major hurdle to prohibit a knockoff design for adaptive variable selection lies in the linear constraint $\gamma = D\beta$ under general $D$, see Section \ref{sec:genlasso} for detailed discussions. 

To overcome this challenge, a natural treatment here is to relax the linear constraint $\gamma = D\beta$ to its Euclidean neighbourhood (proximity) in the lifted parameter space $(\beta,\gamma)$, rendering an unconstrained optimization problem,
\begin{equation} \label{eq:split_lasso}
    \min_{\beta,\gamma} \frac{1}{2n}\|y - X\beta\|_2^2 + \frac{1}{2\nu}\|D\beta - \gamma\|_2^2 + \lambda\|\gamma\|_1,  \ \ \ \lambda>0, 
\end{equation}
where $\nu>0$ is a parameter to control the Euclidean gap between $D\beta$ and $\gamma$. In other words, we shall allow model parameters to vary in the neighborhood or proximity of the linear subspace $D\beta = \gamma$, where the larger is $\nu$, the larger is the relaxation proximity. Such a $\nu$-relaxation renders an \emph{orthogonal design} (identity matrix) for $\gamma$, leading to \emph{orthogonal Split Knockoff} features that are crucial for FDR control, as well as improving the 
selection power to identify strong nonnull features in the presence of noise, as we shall see below. Such a treatment is known as the variable splitting in optimization, hence \eqref{eq:split_lasso} is called \emph{Split LASSO} in this paper.

Another critical treatment in our methodology is that, instead of using the whole dataset, we randomly split the data into two independent subsets, with one to estimate non-sparse intercept parameter $\beta$ and the other to construct knockoffs for the selection of sparse parameter $\gamma$. Such a treatment is crucial to enable independent signs of Split Knockoff statistics, as Bernoulli random variables, recovering supermartingale structures for provable FDR control. 

The detailed procedure goes as follows.

\subsection{Data Splitting and Intercept Estimation on the First Dataset}

\label{sec: intercept est}

\paragraph*{Data Splitting.} The dataset $\D=(X, y)$ is randomly split into two subsets as $\D_1=(X_1, y_1)$ and $\D_2=(X_2, y_2)$ with $n_1$ and $n_2$ samples respectively, where $n_1+n_2=n$, and $n_2\ge m+p$.\footnote{When this constraint is not satisfied, one can use the first dataset for variable screening to reduce the dimension as presented in Section \ref{sec: hd}.}

\paragraph*{Estimation of Intercept ($\widehat{\beta}(\lambda)$) with $\D_1$.}

In the following, we present two typical choices of $\widehat{\beta}(\lambda)$ estimated with $\D_1$. Both choices are later shown to have theoretical FDR control (detailed in Theorem \ref{theorem: fdr} in Section \ref{sec:mainresults}) and demonstrate the desired selection power empirically (see Sections \ref{Sec: simulation results} and \ref{Sec: applications} for more details).

\begin{enumerate}
    \item Take $\widehat{\beta}(\lambda)$ from the Split LASSO regularization path \eqref{eq:split_lasso}, based on the first dataset $\D_1=(X_1, y_1)$. 
    \item Alternatively, to maximize power, one can take $\widehat{\beta}(\lambda)=\widehat{\beta}_{\hat{\lambda},\hat{\nu}}$ as an optimal estimator with minimal cross-validation loss (with respect to the parameters $\lambda$ and $\nu$) on the Split LASSO path.
\end{enumerate}

In fact, the intercept $\widehat{\beta}(\lambda)$ determined by $\D_1$ can be any continuous function with respect to $\lambda$ that satisfies\footnote{This condition is proposed to ensure (by Proposition \ref{prop: zero probability event}) that $Z$ defined in Section \ref{sec: feature and knockoff sig} won't be infinite.} $\lim_{\lambda\to\infty}\frac{\widehat{\beta}(\lambda)}{\lambda} = 0$ to achieve theoretical FDR control, as stated in Theorem \ref{theorem: fdr}. However, we recommend learning $\widehat{\beta}(\lambda)$ from $\D_1$ to demonstrate the desired selection power in addition to FDR control. Examples of $\widehat{\beta}(\lambda)$ that satisfy the above conditions include constant estimators and the solution paths provided above.

\subsection{Construction of Knockoff Matrix with the Second Dataset}
\label{sec: construction}

Now we construct fake knockoff features based on the second dataset $\D_2=(X_2, y_2)$. 
For this purpose, the transformational sparsity model \eqref{eq: model} on the second dataset $\D_2=(X_2, y_2)$ can be rewritten as the following linear regression in the lifted parameter $(\beta^*,\gamma^*)$ space with heterogeneous noise: 
\begin{equation} \label{model}
  \tilde{y}=A_\beta\beta^*+A_\gamma\gamma^*+\tilde\varepsilon,
\end{equation}
where we denote $\varepsilon_2$ to be the Gaussian 
noise in $y_2$ and
\begin{align}
\tilde{y}=
\left(
\begin{array}{c}
\frac{y_2}{\sqrt{n_2}} \\
0_m
\end{array}
\right),\ 
A_\beta=
\left(
\begin{array}{c}
\frac{X_2}{\sqrt{n_2}} \\
\frac{D}{\sqrt{\nu}} 
\end{array}
\right),\ 
A_\gamma=
\left(
\begin{array}{c}
0_{n_2\times m} \\
-\frac{I_m}{\sqrt{\nu}} 
\end{array}
\right),\ 
\tilde{\varepsilon}=
\left(
\begin{array}{c}
\frac{\varepsilon_2}{\sqrt{n_2}} \\
0_m 
\end{array}
\right).\label{eq: features}
\end{align}

\paragraph*{Knockoff Construction on $\D_2$.} The split knockoff copy matrix $\tilde{A}_\gamma$ satisfies
\begin{align}
    \tilde{A}_\gamma^T\tilde{A}_\gamma = A_\gamma^TA_\gamma,\ 
    A_\beta^T\tilde{A}_\gamma = A_\beta^TA_\gamma,\ A_\gamma^T\tilde{A}_\gamma =  A_\gamma^TA_\gamma-\diag(\vecs),\label{eq: copy}
\end{align}
where $\vecs\in\mathbb R^m$ is some non-negative vector. Since the original features $A_\gamma$ is an \emph{orthogonal design}, the Split Knockoff matrix, $\tilde{A}_\gamma$, is thus an \emph{orthogonal matrix}, that imitates the inner product or angles of features in $A_\gamma$, but is different to the original features as much as possible. In particular, the existence of the nonsparse intercept features further requires that $\tilde{A}_\gamma$ imitates the inner product or angles between $A_\gamma$ and $A_\beta$. Explicit solutions of Equation \eqref{eq: copy} when $n_2\ge m+ p$ are discussed in Section \ref{sec: details for copy}. 
For convenience, $\tilde{A}_\gamma$ is partitioned according to that of $A_\gamma$, $\tilde{A}_\gamma^T=(\tilde{A}_{\gamma,1}^T;\tilde{A}_{\gamma,2}^T)^T$, i.e. $\tilde{A}_{\gamma,1}\in\R^{n_2\times m}$ is the submatrix consisting of the first $n_2$ rows of $\tilde{A}_{\gamma}$ and $\tilde{A}_{\gamma,2}\in\R^{m\times m}$ is the remaining submatrix.

\subsection{Feature and Knockoff Significance}

\label{sec: feature and knockoff sig}
Given $(\tilde{y},A_\beta, A_\gamma,\tilde{A}_\gamma)$ and $\widehat{\beta}(\lambda)$, we are now ready to compute feature and knockoff significance statistics.\footnote{Here, we compute the regularization paths on $A_\gamma$ and $\tilde{A}_\gamma$ separately rather than jointly, which can lower the computational cost at lower dimensionality.} 

\paragraph*{Feature Significance ($Z$) on $\D_2$.}

\begin{enumerate}
\item Compute the Split LASSO regularization path for $\gamma$,
    \begin{align} 
        \gamma(\lambda)&:=\arg\min_{\gamma} \frac12\|\tilde y-A_\beta\widehat{\beta}(\lambda)-A_\gamma\gamma\|_2^2+\lambda\|\gamma\|_1, \ \ \ \lambda>0.\label{eq:gamma} 
\end{align}
\item Define the feature significance as the supremum of $\lambda>0$ on the regularization path $\gamma_i(\lambda)$ such that $\gamma_i(\lambda)$ is nonzero for all $i\in\{1,2,\cdots, m\}$:
\begin{align}
    Z_i=  \sup\left\{\lambda>0:\gamma_i(\lambda)\neq0\right\}, \label{def: Z}
\end{align}
or $Z_i = 0$ if the set $\left\{\lambda>0:\gamma_i(\lambda)\neq0\right\}$ is empty.
\end{enumerate}

\paragraph*{Knockoff Significance ($\tilde{Z}$) on $\D_2$.}

\begin{enumerate}
\item Compute the Split LASSO regularization path for $\tilde{\gamma}$,
   \begin{align}
    \tilde{\gamma}(\lambda) :=\arg \min_{\tilde{\gamma}} \frac12\|\tilde y-A_\beta\widehat{\beta}(\lambda)-\tilde{A}_\gamma\tilde{\gamma}\|_2^2+\lambda\|\tilde{\gamma}\|_1, \ \ \ \lambda>0.\label{eq:t_gamma}
    \end{align}
\item Define the knockoff significance as the supremum of $\lambda>0$ on the regularization path $\tilde{\gamma}_i(\lambda)$ such that $\tilde{\gamma}_i(\lambda)$ is nonzero, for all $i\in\{1,2,\cdots, m\}$:
\begin{align}
    \tilde{Z}_i=  \sup\left\{\lambda>0:\tilde\gamma_i(\lambda)\neq0\right\}. \label{def: tZ}
\end{align}
or $\tilde{Z}_i = 0$ if the set $\left\{\lambda>0:\tilde\gamma_i(\lambda)\neq0\right\}$ is empty.
\end{enumerate}

In summary, Equation \eqref{eq:gamma} and \eqref{eq:t_gamma} define regularization paths $\gamma(\lambda)$ and $\tilde{\gamma}(\lambda)$ that typically transition from zero to nonzero as $\lambda$ decreases from $+\infty$ to 0. The supremums of $\lambda$ for which $\gamma(\lambda)\neq 0$ or $\tilde{\gamma}(\lambda)\neq 0$ indicate the importance of $\gamma$ and $\tilde{\gamma}$, respectively. The larger these supremums, the more important the corresponding features are. These supremums are recorded as the feature significance $Z$ or the knockoff significance $\tilde{Z}$. 

In addition to recording the supremums of $\lambda$ for which $\gamma(\lambda)\neq 0$ or $\tilde{\gamma}(\lambda)\neq 0$ as significance statistics, the signs of $\gamma(\lambda)$ and $\tilde{\gamma}(\lambda)$ at such supremums can also be recorded. This allows for a \emph{truncation of sign mismatch} on $\tilde{Z}$ to define a new knockoff significance statistic.

Formally speaking, for all $i$, define the reference signs $r_i$ and $\tilde{r}_i$ as the signs of $\gamma_i(\lambda)$ and $\tilde{\gamma}_i(\lambda)$ ``upon'' becoming nonzero (i.e. at $Z$ and $\tilde{Z}$) when $\lambda$ decreases from $+\infty$ to 0. That is, 
\begin{align}
    &r_i:=\left\{
    \begin{array}{ccl}
        1   &   & \mbox{if }Z_i> 0\mbox{ and }\limsup_{\lambda\to Z_i^-}\sign(\gamma_i(\lambda)) = 1,\\
        0   &   &  \mbox{if }Z_i= 0,\\
        -1  &   & \mbox{if }Z_i> 0\mbox{ and }\liminf_{\lambda\to Z_i^-}\sign(\gamma_i(\lambda)) = -1,\label{def: equivalent def r}\\
    \end{array} \right.\\
    &\tilde r_i:=\left\{\begin{array}{ccl}
        1   &   & \mbox{if }\tilde{Z}_i> 0 \mbox{ and }\limsup_{\lambda\to \tilde Z_i^-}\sign(\tilde\gamma_i(\lambda)) = 1,\\
        0   &   &  \mbox{if }\tilde Z_i= 0,\\
        -1  &   & \mbox{if }\tilde Z_i> 0\mbox{ and }\liminf_{\lambda\to \tilde Z_i^-}\sign(\tilde\gamma_i(\lambda)) = -1.\\
    \end{array}\right.\label{def: equivalent def tilde r}
\end{align}
Proposition \ref{prop: well definedness} in Section \ref{sec:mainresults} provides an equivalent definition of $r$ and $\tilde{r}$ through the Karush-Kuhn-Tucker (KKT) conditions of Equation \eqref{eq:gamma} and Equation \eqref{eq:t_gamma}.

Using $r$ and $\tilde{r}$, the \emph{sign mismatch truncation} of $\tilde{Z}$ is defined by
\begin{equation} \label{eq:truncation}
    \tau(\tilde{Z}):=\tilde{Z}\odot {\bf 1}\{r = \tilde{r}\}
\end{equation}
where the symbol $\odot$ stands for Hadamard product.
Such a truncation sets zeroes those knockoff significance values if their signs differ to their associated feature significance values. 
As we shall see later in this paper, the truncated knockoff significance may reduce the conservativeness in selections compared with the untruncated one.

\subsection{Three Types of Split Knockoff Statistics}
\label{sec: w stat}
We introduce a family of Split Knockoff $W$ statistics. 

\paragraph*{Split Knockoff Statistics.}
\begin{enumerate}
    \item $\Ws:=Z\odot \sign(Z -\tilde{Z} )\footnote{It is shown in Proposition \ref{prop: zero probability event} that $\{\exists i: Z_i=\tilde{Z}_i>0\}$ is a zero probability event. For simplicity of notations, the event $\{\exists i: Z_i=\tilde{Z}_i>0\}$ will be omitted throughout this paper.},$ where S refers to ``Split''.
    \item $\Wst:=Z\odot \sign(Z -\tau(\tilde{Z}) ),$ where S$\tau$ refers to applying truncation \eqref{eq:truncation} on $\Ws$.
    \item $\Wbc:= (Z \vee \tilde{Z}) \odot\sign(Z -\tilde{Z} ),$ where BC refers to the original definition adopted by Barber-Cand\`{e}s in \cite{barber2015controlling}.
\end{enumerate}

All three versions will be handled in a unified framework in this paper. 
For shorthand notation, we use $W\upc$ to represent any one in the family, where $\C\in\{\s, \stau, \bc\}$. In all cases, as we wish to select $i$ when $W_i\upc$ is large and positive. Let $q$ be our target FDR, two data-dependent threshold rules on a pre-set nominal level $q$ are defined as 
\begin{equation*}
  \mbox{(Split Knockoff)}\ \ \ T_q\upc=\min\left\{\lambda\in\mathcal{W}\upc:\frac{|\{i:W_i\upc\le-\lambda\}|}{1\vee|\{i:W_i\upc\ge \lambda\}|}\le q\right\}, 
\end{equation*}
\begin{equation*}
  \mbox{(Split Knockoff+)}\ \ \ T_q\upc=\min\left\{\lambda\in\mathcal{W}\upc:\frac{1+|\{i:W_i\upc\le-\lambda\}|}{1\vee|\{i:W_i\upc\ge \lambda\}|}\le q\right\},
\end{equation*}
or $T_q\upc=+\infty$ if this set is empty, where $\mathcal{W}\upc = \{|W_j\upc|: j = 1, 2, \cdots, m\}\backslash\{0\}$.
In all cases, the selector is defined as
\begin{equation}
  \hat{S}\upc=\{i:W_i\upc\ge T_q\upc\}.\nonumber
\end{equation}

The following proposition summarizes an important inclusive relationship among the three selectors.
\begin{proposition}[Inclusion Property of Selectors]
    \label{prop: relations}
    For the selectors $\hat{S}^{\bc}$, $\hat{S}^{\s}$ and $\hat{S}^{\stau}$, there holds
    \begin{align*}
        \hat{S}^{\bc}\subseteq \hat{S}^{\s}\subseteq \hat{S}^{\stau}.
    \end{align*}
\end{proposition}
In other words, $\hat{S}^{\stau}$ achieves the highest selection power among these selectors while $\hat{S}^{\bc}$ is the most conservative. This relationship is also validated by simulation experiments in Section \ref{Sec: simulation results}. The proof of Proposition \ref{prop: relations} is given in Section \ref{sec: proof prop1}. Next section will disclose that they all achieve the desired FDR control.

\section{FDR Control of Split Knockoffs} 
\label{sec:mainresults}

In this section, we first show that the false discovery rates are under control for the Split Knockoff method, then present the key ideas on how to reach this analysis. In particular, we consider the basic setting with $n_2\ge m +p$ here, leaving the high dimensional extension to Section \ref{sec: hd}.

Specifically, the following theorem of FDR control is established for all three selectors. For Split Knockoff, we control a ``modified'' FDR (mFDR) as \cite{barber2015controlling} that adds $q^{-1}$ in the denominator, which should have little effect if a large number of features are selected; for Split Knockoff+, we get the exact FDR control.
  
    \begin{theorem}[FDR Control of Split Knockoffs]
    \label{theorem: fdr}
    For all $0<q\le 1$, $\C\in\{\s,\stau, \bc\}$, and all $\nu>0$, there holds
    \begin{itemize}
        \item[(a)] (mFDR of Split Knockoff)
    \begin{equation}
           \E\left[\frac{\left|\left\{i:i\in \hat{S}\upc\cap S_0\right\}\right|}{\left|\hat{S}\upc\right|+q^{-1}}\right]\le q.\nonumber
    \end{equation}
    \item[(b)] (FDR of Split Knockoff+)
    \begin{equation}
           \E\left[\frac{\left|\left\{i:i\in \hat{S}\upc\cap S_0\right\}\right|}{\left|\hat{S}\upc\right|\vee 1}\right]\le q.\nonumber
    \end{equation}
    \end{itemize}
  \end{theorem}
Since the FDR is uniformly under control for all $\nu>0$, 
 the hyperparameter $\nu$ from Split LASSO can be used to optimize the power of Split Knockoffs. The influence of $\nu$ on power is of two folds (see Section \ref{sec: sign consistency short}). On one hand, increasing $\nu$ may increase the power of discovering strong nonnull features whose magnitudes are large. It is because that enlarging $\nu$ will improve the incoherence condition of the Split LASSO such that strong nonnull features will appear earlier on the path than the weak ones and nulls. On the other hand, increasing $\nu$ may lose the power of discovering weak nonnull features whose magnitudes are below $\nu$ and the noise scale. Therefore a trade-off of these two aspects will lead to a good power by optimizing $\nu$, which can be achieved empirically by cross-validation as what we describe in intercept estimation $\widehat{\beta}$. 
 Performance of such a choice will be confirmed in simulation experiments in Section \ref{Sec: simulation results}.

The main challenge in establishing the FDR control of Split Knockoffs lies in the failure of exchangeability (see Section \ref{sec: fail exchange} for more discussions). Such a failure leads to dependency between the magnitude and sign of the original type of $W$ statistics ($\Wbc$), which further fails the inverse supermartingale argument in \cite{barber2015controlling, barber2019knockoff} for provable FDR control. To overcome this challenge, our key technical development is based on the orthogonal design and the rendered orthogonal Split Knockoff copies, which together with data splitting leads to independent Bernoulli processes for signs 
of $\Ws$ and $\Wst$, decoupled or independent to their magnitudes. This enables us some new \emph{inverse supermartingale structures} such that the FDR control can be proved, 
in a similar way to \cite{barber2019knockoff}. Furthermore, $\Ws$ enjoys a particularly nice structure that its associated inverse supermartingale has a more refined filtration than that of $\Wbc$, which enables us to reach an upper bound for $\Wbc$ via $\Ws$.
Below we will present the key ideas in detail for the FDR control. 

First of all, following the standard steps in knockoffs as in \cite{barber2015controlling}, we can transfer the problem of bounding FDR by $q$ in Theorem \ref{theorem: fdr} into the problem of bounding $\E\left[\M_{T_q\upc}(W\upc)\right]$ by one (see Section \ref{sec: proof thm} for details), where for any $T> 0$, $\M_{T}(W\upc)$ is defined as
\begin{align}\label{def: mtw}
    \M_{T}(W\upc)=\frac{\sum_{i\in S_0}1\{W_i\upc\ge T\}}{1+\sum_{i\in S_0}1\{W_i\upc\le -T\}}.
\end{align}
The following proposition reveals the relationships among $\M_{T}(\Wst)$, $\M_{T}(\Ws)$, and $\M_{T}(\Wbc)$, in particular.
\begin{proposition}[Inequalities of $W$ Statistics]
    \label{prop: inequality of w statistics}
    The following holds for all $T>0$. 
    \begin{enumerate}
        \item For all $i$, $\{\Wst_i\le-T\}\subseteq\{\Ws_i\le-T\}\subseteq\{\Wbc_i\le-T\}$.
        \item For all $i$, $\{\Wbc_i\ge T\} = \{\Ws_i\ge T\}\subseteq\{\Wst_i\ge T\}$.
        \item Combining the above two points together, there holds $\M_{T}(\Wst)\ge \M_{T}(\Ws)\ge \M_{T}(\Wbc)$.
    \end{enumerate}
\end{proposition}

Proposition \ref{prop: inequality of w statistics} suggests that the statistics $\Wst$ are the least conservative in terms of FDR, since $\M_{T}(\Wst) = \max\{\M_{T}(\Wst),\M_{T}(\Ws),\M_{T}(\Wbc)\}$, where $\M_{T}(W\upc)$ is introduced to upper bound the FDR for $\C\in\{\s,\stau, \bc\}$ respectively. The proof of Proposition \ref{prop: inequality of w statistics} is provided in Section \ref{sec: proof inequality of w statistics}.

Next, we will handle the FDR control of Split Knockoff statistics $\Ws$ and the Barber-Cand\`{e}s type statistics $\Wbc$ in two cases, respectively, while the analysis on the sign-mismatch truncated statistics $\Wst$ can be done in a similar way as $\Ws$. The section is organized as follows:

\begin{enumerate}
    \item An inverse supermartingale structure on $\M_{T}(\Ws)$ is introduced, which enables the optional stopping theorem to bound $\E\left[\M_{T_q^{\s}}(\Ws)\right]$, in Section \ref{sec:s analysis}.
    \item Barber-Cand\`{e}s type of knockoff statistics $\Wbc$ renders a coarser filtration than $\Ws$, which enables $\E\left[\M_{T_q^{\bc}}(\Ws)\right]$ providing an upper bound of $\E\left[\M_{T_q^{\bc}}(\Wbc)\right]$, in Section \ref{sec: BC W analysis}.
\end{enumerate}

\subsection{Case I: Split Knockoff Statistics}
\label{sec:s analysis}

The key property of Split Knockoff statistics $\Ws$ is that the magnitudes $|\Ws|$ are independent of the signs ${\text{sign}(\Ws)}$, and the signs ${\text{sign}(\Ws)}$ are independent Bernoulli random variables. This property allows us to construct an inverse supermartingale, which will be presented below.. 

We begin with the Karush-Kuhn-Tucker (KKT) conditions of optimization problems described in Equations \eqref{eq:gamma} and \eqref{eq:t_gamma}. These equations imply that the optimal solutions $\gamma(\lambda),\ \tilde{\gamma}(\lambda)$ with respect to $\lambda>0$ must satisfy the following KKT conditions:
\begin{subequations}
    \label{eq: kkts}
    \begin{align}
        \lambda\rho(\lambda) + \frac{\gamma(\lambda)}{\nu}&= \frac{D\widehat{\beta}(\lambda)}{\nu},\label{eq: feature kkt}\\
        \lambda\tilde{\rho}(\lambda) + \frac{\tilde{\gamma}(\lambda)}{\nu} &= \frac{D\widehat{\beta}(\lambda)}{\nu} +\underbrace{\left\{  - \diag(\vecs)\gamma^* + \frac{\tilde{A}^T_{\gamma,1}}{\sqrt{n_2}} \varepsilon_2 \right\}}_{=:\zeta},\label{eq: knockoff kkt} \\
        \rho(\lambda) &\in \partial \|\gamma(\lambda)\|_1,\\ \quad
        \tilde{\rho}(\lambda) & \in \partial \|\tilde{\gamma}(\lambda)\|_1, 
    \end{align}
\end{subequations}
where $\varepsilon_2$ is the Gaussian noise in $y_2$, and $\partial \|\cdot\|_1$ is the set of subgradients of $\ell_1$-norm $\|\cdot\|_1$.  For shorthand notation, define $\zeta := - \diag(\vecs)\gamma^* + \frac{\tilde{A}^T_{\gamma,1}}{\sqrt{n_2}} \varepsilon_2$.

Proposition \ref{prop: well definedness} below provides equivalent definitions of $r$ and $\tilde{r}$ through the functions $\rho(\lambda)$ and $\tilde{\rho}(\lambda)$ by exploiting the KKT conditions \eqref{eq: kkts}. This suggests that the definitions of $r$ and $\tilde{r}$ given by Equations \eqref{def: equivalent def r} and \eqref{def: equivalent def tilde r} are well-defined. The proof of Proposition \ref{prop: well definedness} can be found in Section \ref{sec: proof well defined}. 

\begin{proposition}
\label{prop: well definedness}
        The following definitions of $r$ and $\tilde{r}$ are equivalent with Equation \eqref{def: equivalent def r} and \eqref{def: equivalent def tilde r}:
    \begin{align}
    r_i:=\left\{
    \begin{array}{ccl}
        \sign(\rho_i(Z_i))   &   & \mbox{if }Z_i> 0,\\
            0   &   &  \mbox{if }Z_i= 0,
        \end{array} \right.\ \tilde r_i:=\left\{\begin{array}{ccl}
            \sign(\tilde\rho_i(\tilde{Z}_i))   &   & \mbox{if }\tilde{Z}_i> 0,\\
            0   &   &  \mbox{if }\tilde{Z}_i= 0.
        \end{array}\right.\label{def: r and tilde r}
    \end{align}
\end{proposition}

Now we are ready to show that the magnitudes $|\Ws|$ are independent to the signs $\{\sign(\Ws)\}$. Note that $\widehat{\beta}(\lambda)$ is determined from the first dataset $\D_1=(X_1, y_1)$, whence independent to the second dataset $\D_2=(X_2,y_2)$ and particularly $\zeta$. Therefore, it suffices to consider the random variable $\zeta$ conditional on a predetermined $\widehat{\beta}(\lambda)$. From Equation \eqref{eq: feature kkt}, $\gamma(\lambda)$ is determined by $\widehat{\beta}(\lambda)$ and thus independent to $\zeta$, so is $|\Ws|=Z$ as a function of $\gamma(\lambda)$. On the other hand, from Equation \eqref{eq: knockoff kkt}, conditional on a determined $\widehat{\beta}(\lambda)$, $\tilde\gamma(\lambda)$ as well as $\tilde{Z}$, are determined by $\zeta$ from $\D_2=(X_2, y_2)$. 
Therefore, $|\Ws|$, determined by $\widehat{\beta}(\lambda)$ on $\D_1=(X_1,y_1)$, is independent to $\sign(\Ws)=\sign(Z-\tilde{Z})$, depending on $\zeta$ with $\D_2=(X_2,y_2)$. 

To see that the signs $\{\sign(\Ws)\}$ are independent Bernoulli random variables, $\zeta$ follows an independently joint Gaussian distribution 
as a result of the \emph{orthogonal} Split Knockoff matrix satisfying \eqref{eq: copy} (see Equation \eqref{eq: ortho of agamma1} in Section \ref{sec: proof lemma} for details), {\it i.e.} 
\begin{align}
    \zeta\sim \mathcal{N}\left(-\diag(\vecs)\gamma^*, \frac{1}{n_2}\diag(\vecs)(2I_m-\diag(\vecs)\nu)\sigma^2\right). \label{eq: zeta dis}
\end{align}
Then we have the following lemma.

\begin{lemma}
    \label{lemma: independent Ws front}
    Given any determined $\widehat{\beta}(\lambda)$, $1\{\Ws_i<0\}$ are some independent Bernoulli random variables. Furthermore, for $i\in S_0\cap\{i:|\Ws_i| = Z_i> 0\}$, there holds
    \begin{align*}
        \Prob[\Ws_i<0]\ge 
        \frac{1}{2}.
    \end{align*}
\end{lemma}

This lemma states that given $\widehat{\beta}(\lambda)$ which determines $|\Ws|=Z$, $\sign(\Ws)$ consists of independent random variables with $\Prob[\Ws_i<0]\ge 1/2$ on the nulls satisfying $|\Ws_i| = Z_i> 0$. With such a property, we can construct an inverse supermartingale whose optional stopping theorem gives an upper bound on $\E\left[\M_{T_q^{\s}}(\Ws)\right]$. 

For simplicity, we rearrange the index of $\Ws$, such that $|\Ws_{(1)}|\ge|\Ws_{(2)}|\ge\cdots\ge|\Ws_{(m^*)}|>0$ for $\{(1), (2), \cdots, (m^*)\}= S_0\cap\{i:|\Ws_i| = Z_i> 0\}$. Further denote $B_{(i)} = 1\{\Ws_{(i)}<0\}$,
then there holds
\begin{align}
    \frac{\sum_{i\in S_0}1\{\Ws_i\ge T_q^{\s}\}}{1+\sum_{i\in S_0}1\{\Ws_i\le -T_q^{\s}\}} & =  \frac{1+\sum_{i\in S_0}1\{|\Ws_i|\ge T_q^{\s}\}}{1+\sum_{i\in S_0}1\{|\Ws_i|\ge T_q^{\s}, \Ws_i<0\}}-1,\nonumber\\
    & = \frac{1+J}{1+B_{(1)}+B_{(2)}+\cdots+B_{(J)}}-1,\label{eq: transferred bound}
\end{align}
where $J\le m^*$ is defined to be the index satisfying 
\begin{align*}
    |\Ws_{(1)}|\ge|\Ws_{(2)}|\ge\cdots\ge|\Ws_{(J)}|\ge T_q^{\s}>|\Ws_{(J+1)}|\ge\cdots\ge|\Ws_{(m^*)}|,
\end{align*}
in other words, $J = \argmax_{k\le m^*}\{|\Ws_{(k)}|\ge T_q^{\s}\}$. The following lemma summarizes the \emph{inverse supermartingale} inequality, 
which gives an upper bound for the Equation \eqref{eq: transferred bound}.

\begin{lemma}
    \label{lemma: key lemma front}
    Given any determined $\widehat{\beta}(\lambda)$, 
    let $\{G_i\}_{i=1}^m$ be some proper Borel sets such that $B_i=1\{\zeta_i\in G_i\}$ with $\Prob[B_i=1]=\rho_i$. Let $\rho>0$ satisfy $\rho\le \min_{i\in S_0}\{\rho_i\}$. Let $J$ be a stopping time with respect to the filtration $\{\mathcal{F}_j\}_{j=1}^m$ in inverse time defined as
    \begin{align*}
        \mathcal{F}_j = \sigma\left(\left\{\sum_{i=1}^jB_{(i)}, \zeta_{(j+1)}, \cdots, \zeta_{(m)}\right\}\right)
    \end{align*}
    Then
    \begin{align*}
        \Expect\left[\frac{1+J}{1+B_{(1)}+B_{(2)}+\cdots+B_{(J)}}\right]\le \rho^{-1}.
    \end{align*}
\end{lemma}
One can verify that $J = \argmax_{k\le m^*}\{|\Ws_{(k)}|\ge T_q^{\s}\}$ 
is indeed a stopping time with respect to the filtration $\{\mathcal{F}_j\}_{j=1}^{m}$ in inverse time. Applying Lemma \ref{lemma: key lemma front} to Equation \eqref{eq: transferred bound} with the property that $\Prob[B_i = 1]\ge \rho= 1/2$ for $i\in S_0$ as shown in Lemma \ref{lemma: independent Ws front}, we reach our desired result. The same procedure can be applied to the case of $\Wst$.

\subsection{Case II: Barber-Cand\`{e}s Type Statistics}
\label{sec: BC W analysis}

The key lies in that $\Ws$ provides us a bridge to the analysis of $\Wbc$, enabling an upper bound on $\M_{T_q^{\bc}}(\Wbc)$ via an inverse martingale inequality for $\M_{T_q^{\bc}}(\Ws)$ due to the following facts.

\begin{itemize}
    \item[(a)] $\M_{T}(\Ws)$ upper bounds $\M_{T}(\Wbc)$, that $\M_{T}(\Wbc)\leq \M_{T}(\Ws)$ for all $T> 0$ as presented in Proposition \ref{prop: inequality of w statistics};
    \label{item a}
    \item[(b)] The inverse supermartingale associated with $\M_{T}(\Ws)$ has a more refined filtration than that of $\M_{T}(\Wbc)$, such that the stopping time $T_q^{\bc}$ is also a stopping time of the former.
\end{itemize}

Precisely, we have the following critical proposition. 
\begin{proposition}[Inclusion Property of Filtration]
\label{prop: s to bc}
    Let $\Fs^{\bc}(T)$ ($\Fs^{\s}(T)$) be the filtration induced by $\Wbc$ ($\Ws$) respectively, i.e.
    \begin{align*}
        \Fs^{\bc}(T) & = \sigma\left(\#\{i: \Wbc_i\ge T\}, \#\{i: \Wbc_i\le -T\},\{\zeta_i: |\Wbc_i|< T\}\right),\\
        \Fs^{\s}(T) & = \sigma\left(\#\{i: \Ws_i\ge T\}, \#\{i: \Ws_i\le -T\},\{\zeta_i: |\Ws_i|< T\}\right).
    \end{align*}
    Then there holds for any $T>0$ that
    \begin{align*}
        \Fs^{\bc}(T)\subseteq \Fs^{\s}(T).
    \end{align*}
\end{proposition}

For any stopping time $T_q^{\bc}$ in reverse time adapted to the filtration $\Fs^{\bc}(T)$, Proposition \ref{prop: s to bc} points out that $\Fs^{\s}(T)$ is a superset refinement of $\Fs^{\bc}(T)$, whence $T_q^{\bc}$ is also a stopping time in reverse time adapted to the filtration $\Fs^{\s}(T)$. Applying optional stopping theorem on $\M_{T_q^{\bc}}(\Ws)$, there holds
\begin{align*}
    \Expect\left[\M_{T_q^{\bc}}(\Ws)\right]\le 1.
\end{align*}
Combining with the property that $\M_{T}(\Wbc)\leq \M_{T}(\Ws)$ for any $T> 0$, there further holds
\begin{align*}
    \Expect\left[\M_{T_q^{\bc}}(\Wbc)\right]\le\Expect\left[\M_{T_q^{\bc}}(\Ws)\right]\le 1,
\end{align*}
which is our desired result.  

The proof of Proposition \ref{prop: s to bc} will be provided in Section \ref{sec: proof s to bc}. The complete proof of Theorem \ref{theorem: fdr} can be found in Section \ref{sec: proof thm}, with the proofs of supporting lemmas in Section \ref{sec: proof lemma}.

\section{A Generalization in High Dimensional Settings}

\label{sec: hd}

In this section, we extend the Split Knockoffs into high dimensional settings, where the sample size is limited and the condition of $n_2\ge m+p$ cannot be satisfied. In such settings, we will conduct feature screening in the first dataset $\D_1$ to deduct the number of features, and then perform Split Knockoffs on the screened subset of features in the second dataset $\D_2$. 

In particular, we will use the first dataset $\D_1 = (X_1, y_1)$ to conduct initial screening and give estimated support sets $\hat S_\beta$, $\hat S_\gamma$ for $\beta$, $\gamma$ respectively, such that $n_2\ge |\hat S_\beta|+|\hat S_\gamma|$ is satisfied. The feature screening is a well-studied topic, and a list of the screening methods can be found in but not limited to \citep{wasserman2009high, wu2010screen}. After the screening step, the rest steps of Split Knockoffs will be conducted on the estimated support sets $\hat S_\beta$, $\hat S_\gamma$.

For the rest steps, we first generate the intercept $\widehat{\beta}(\lambda)$ as a bounded continuous function ($\R_+\to \R^{|\hat S_\beta|}$) from $\D_1$. Then it remains to generate the feature and knockoff importance statistics $Z$ and $\tilde{Z}$ in order to perform Split Knockoffs.

Let $X_{\hat S_\beta}$ be the submatrix of $X_2$, consisting of the columns indexed by $\hat S_\beta$. Let $D_{\hat S_\beta, \hat S_\gamma}$ be the submatrix of $D$, consisting of the columns indexed by $\hat S_\beta$ and rows indexed by $\hat S_\gamma$. Further, let
\begin{align}
\tilde{y}=
\left(
\begin{array}{c}
\frac{y_2}{\sqrt{n_2}} \\
0_{|\hat S_\gamma|}
\end{array}
\right),\ 
A_\beta=
\left(
\begin{array}{c}
\frac{X_{\hat S_\beta}}{\sqrt{n_2}} \\
\frac{D_{\hat S_\beta, \hat S_\gamma}}{\sqrt{\nu}} 
\end{array}
\right),\ 
A_\gamma=
\left(
\begin{array}{c}
0_{n_2\times |\hat S_\gamma|} \\
-\frac{I_{|\hat S_\gamma|}}{\sqrt{\nu}} 
\end{array}
\right),\ 
\tilde{\varepsilon}=
\left(
\begin{array}{c}
\frac{\varepsilon_2}{\sqrt{n_2}} \\
0_{|\hat S_\gamma|} 
\end{array}
\right).\label{eq: features hd}
\end{align}
For shorthand notations, we abuse the notations here and use the same notations as in Equation \eqref{eq: features}, while keep in mind that $\tilde{y}$, $A_\beta$, $A_\gamma$ and $\tilde{\varepsilon}$ defined in Equation \eqref{eq: features hd} are dependent on the estimated support sets $\hat S_\beta$, $\hat S_\gamma$. With Equation \eqref{eq: features hd}, we can construct the split knockoff copy matrix $\tilde{A}_\gamma$ satisfying Equation \eqref{eq: copy} in high dimensional settings. Then the feature and knockoff statistics can be generated in the same way as in Equation \eqref{def: Z} and Equation \eqref{def: tZ}.

For the FDR control about the procedure above, the following Theorem \ref{thm: fdr hd} states that when the estimated support set $\hat S_\beta$
includes the true support set of $\beta$, such a procedure above will not cause any loss in the FDR control. The proof of Theorem \ref{thm: fdr hd} is given in Section \ref{sec: proof hd thm}.

\begin{theorem}
    \label{thm: fdr hd}
    Let $S_\beta$ be the true support sets for $\beta$. Let $\Upsilon$ be the event that $S_\beta\subseteq \hat S_\beta$, then there holds for all $0<q\le 1$, $\C\in\{\s,\stau, \bc\}$, and all $\nu>0$,
    \begin{itemize}
        \item[(a)] (mFDR of Split Knockoff)
    \begin{equation}
           \E\left[\left.\frac{\left|\left\{i:i\in \hat{S}\upc\cap S_0\right\}\right|}{\left|\hat{S}\upc\right|+q^{-1}}\right|\Upsilon\right]\le q.\nonumber
    \end{equation}
    \item[(b)] (FDR of Split Knockoff+)
    \begin{equation}
           \E\left[\left.\frac{\left|\left\{i:i\in \hat{S}\upc\cap S_0\right\}\right|}{\left|\hat{S}\upc\right|\vee 1}\right|\Upsilon\right]\le q.\nonumber
    \end{equation}
    \end{itemize}
\end{theorem}

The event that $\Upsilon = \{S_\beta\subseteq \hat S_\beta\}$ is often known as the sure screening event \citep{fan2008sure} in the literature. It is worth to mention that an immediate corollary follows from Theorem \ref{thm: fdr hd}, that $\mathrm{mFDR}\le q+\Prob\left[\Upsilon^C\right]$ for Split Knockoff and $\mathrm{FDR}\le q+\Prob\left[\Upsilon^C\right]$ for Split Knockoff+. Theorem \ref{thm: fdr hd} achieves comparable results with Theorem 2 in \cite{barber2019knockoff}. The simulation experiments that validate the effectiveness of this extension is presented in Section \ref{sec: hd_simulation}.

\section{Discussion}

\label{sec: discussion}

 In this section, we first discuss a straightforward way of conducting Knockoffs on generalized LASSO under a special case of transformational sparsity when $D$ is of full row-rank, and why such a procedure does not work in general. Next, we discuss how the exchangeability property fails for Split Knockoffs and the challenges for provable FDR control. Finally, we discuss how the model selection consistency of Split LASSO may affect the selection power of Split Knockoffs.

\subsection{Generalized LASSO and Knockoffs} \label{sec:genlasso}
 Within the large literature dealing with the transformational sparsity problem, generalized LASSO \citep{tibshirani2011solution} is the most popular one. In a special case when $D$ is of full row-rank, i.e. a surjective linear map, generalized LASSO on transformational sparsity can be converted to LASSO on direct sparsity where one can design knockoffs; yet in general, it remains open how to do so. 

Recall that the generalized LASSO solves the optimization problem below with $\lambda>0$, 
\begin{equation}
    \min_\beta\ \frac{1}{2n}\|y-X\beta\|_2^2+\lambda\|D\beta\|_1. \label{genlasso_problem}
\end{equation}
In the special case that $D$ is surjective, i.e. $\rank D=m\le p$,  this problem can be equivalently represented by the following LASSO procedure with respect to $\lambda>0$,
\begin{align}
        \min_\gamma \ \frac{1}{2n}\|y-X D^\dagger\gamma-X D_0\beta_0\|^2_2+\lambda\|\gamma\|_1,\nonumber
\end{align}
where $D^\dagger\in \R^{p\times m}$ is the pseudo inverse for $D$, $D_0\in \R^{p\times (p-\rank{D})}$ is a matrix whose columns spans  the null space $\ker(D)$, and $\beta_0\in \R^{p-\rank{D}}$ 
is the representation coefficient for the null space $\ker(D)$ without sparsity. Note that we can write $\beta^* = D^\dagger\gamma^*+D_0\beta_0^*$ for some $\beta_0^*\in \R^{p-\rank{D}}$, then Equation \eqref{eq: model} becomes
\begin{align}
    y=X D^\dagger\gamma^* + X D_0\beta_0^* + \varepsilon.\nonumber
\end{align}
In this regard, one can construct standard knockoffs in this special case that $\rank D=m\le p$. To see this,  
Taking $U\in \R^{n\times(n-p+\rank{D})}$ as an orthogonal complement for the column space of $X D_0\in \R^{n\times (p-\rank{D})}$, we have
\begin{align}
    U^T y = U^T X D^\dagger\gamma^* + U^T\varepsilon. \label{tran}
\end{align}
Now one can treat $U^T X D^\dagger$ as a new design matrix, $U^T y$ as the response vector and $U^T\varepsilon\sim\mathcal{N}(0, I_{n-p+\rank{D}})$ as the Gaussian noise. Then the transformational sparsity problem is transferred to a sparse linear regression problem, where one can apply the standard knockoff method for the FDR control. 

However, a shortcoming of this approach is that it may suffer poor selection power when $D$ is nontrivial and/or $X$ is highly co-related, such that the incoherence requirements for sparse recovery or model selection consistency \citep{DonHuo01,Tropp04,zhao2006model,Wainwright09,ORXYY16,huang2020boosting} fail for the design matrix $U^T X D^\dagger$.

Furthermore, this conversion does not work in the general setting $m>p$. In the case $m>p$, 
$\gamma$ lies in the column space of $D$, a proper subspace of $\R^m$, and it remains open how to design a knockoff method under such a constraint before our work.

\subsection{Failure of Exchangeability in Split Knockoffs}

\label{sec: fail exchange}

The main challenge in establishing the FDR control of Split Knockoffs lies in the failure of exchangeability. In this section, we provide the details on how the exchangeability of standard Knockoffs fails in Split Knockoffs.

Specifically, \cite{barber2015controlling} (Lemma 2-3), established pairwise exchangeability for both features and responses in standard knockoffs, the latter of which however fails in Split Knockoffs. 

\begin{proposition}[Failure of Pairwise Exchangeability]
    \label{prop: exchange}
    For any $S\subseteq \S_0$, there holds
    \begin{subequations}
    \begin{align}
        [A_\gamma, \tilde{A}_\gamma]_{\mathrm{swap}\{S\}}^T[A_\gamma, \tilde{A}_\gamma]_{\mathrm{swap}\{S\}} = &  [A_\gamma, \tilde{A}_\gamma]^T[A_\gamma, \tilde{A}_\gamma]\label{good},\\
        [A_\gamma, \tilde{A}_\gamma]_{\mathrm{swap}\{S\}}^T\tilde y\overset{d}{\neq} &  [A_\gamma, \tilde{A}_\gamma]^T\tilde y,\label{devil}
    \end{align}
\end{subequations}
where $[A_\gamma, \tilde{A}_\gamma]_{\mathrm{swap}\{S\}}$ denotes a swap of the $j$-th column of $A_\gamma$ and $\tilde{A}_\gamma$ in $[A_\gamma, \tilde{A}_\gamma]$ for all $j\in S$.
\end{proposition}

In other words, for Split Knockoffs, while pairwise exchangeability for the features holds in Equation \eqref{good}, that for the responses fails in Equation \eqref{devil}. 
To see the details for Proposition \ref{prop: exchange}, it can be calculated that
\begin{align}
    A_\gamma^T\tilde y= 0,\ \ \ \ \tilde{A}_\gamma^T\tilde y = \frac{\tilde{A}^T_{\gamma, 1}}{\sqrt{n_2}}X_2\beta^*+\frac{\tilde{A}^T_{\gamma, 1}}{\sqrt{n}}\varepsilon_2,
    \label{I_original}
\end{align}
where $\tilde{A}^T_{\gamma, 1}$ together with $\tilde{A}^T_{\gamma, 2}$ are defined in the end of Section \ref{sec: construction}. Moreover, by Equation \eqref{eq: copy}, there holds
    \begin{align}
        \tilde{A}^T_{\gamma,1}\frac{X_2}{\sqrt{n_2}}+\tilde{A}^T_{\gamma,2}\frac{D}{\sqrt{\nu}}= & -\frac{D}{\nu},\ \ \ \ 
        -\frac{\tilde{A}^T_{\gamma,2}}{\sqrt{\nu}}= \frac{I}{\nu}-\diag(\vecs).
    \end{align}
Thus it can be solved that $\tilde{A}^T_{\gamma,1}\frac{X_2}{\sqrt{n_2}} = \diag(\vecs)\cdot D$. Plugging the solution into Equation \eqref{I_original}, there holds
\begin{align}
    [A_\gamma, \tilde{A}_\gamma]^T\tilde y=
    \begin{bmatrix}
        0\\
        -\diag(\vecs)\gamma^*+\frac{\tilde{A}^T_{\gamma, 1}}{\sqrt{n}}\varepsilon
    \end{bmatrix},
\end{align}
where swapping any $i\in \S_0$ in the first and second block  
will lead to different distributions.

The failure of exchangeability may impose a theoretical challenge for knockoff-based methods, as the beautiful symmetry results \citep{barber2015controlling} can not be applicable. 
For random designs, \cite{barber2020robust} exploits KL divergence to measure the ``distance'' to exchangeability, and then gives approximate but not exact FDR control based on the ``distance''. It is also worth mentioning that this method no longer relies on the martingale arguments and uses direct analysis.  
However, it is not clear how to apply such methods to fixed designs in our scenario. Instead, our strategy in this paper is to exploit the orthogonal design in Split LASSO and the data splitting, which leaves us independent signs of Split Knockoff statistics, recovering the inverse supermartingale structure similar to \cite{barber2019knockoff} without using the exchangeability.

\subsection{Selection Power and Model Selection Consistency of Split LASSO}

\label{sec: sign consistency short}

In this section, we discuss the model selection (sign) consistency of Split LASSO with respect to $\nu$, to reveal how $\nu$ might affect the selection power of Split Knockoffs.
As we shall see below, enlarging $\nu$ might increase the power of discovering strong nonnull features whose magnitudes are large, by improving the incoherence condition; on the other hand, doing so may lose the power of discovering weak nonnull features of small magnitudes.

This point is made precise by Proposition \ref{thm: sign consistency front} on the model selection (sign) consistency of Split LASSO. Define $H_\nu := I_m - \frac{D[\Sigma_X+L_D]^{-1}D^T}{\nu}$, where $\Sigma_X = \frac{X^TX}{n}$ and $L_D = \frac{D^TD}{m}$. Further denote $H_\nu^{11}$ as the submatrix of index set $\S_1$, $H_\nu^{00}$ for index $\S_0$, and $H_\nu^{10}$, $H_\nu^{01}$ to be the covariance matrices between $\S_1$ and $\S_0$. 
It is now ready to state Proposition \ref{thm: sign consistency front}, leaving its proof in Section \ref{sec:proof_path_consistency}.

\begin{proposition}[Model Selection Consistency of Split LASSO \eqref{eq:split_lasso}]
\label{thm: sign consistency front}
    Assume that the design matrix $X$ and $D$ satisfy 
    \begin{itemize}
    \item\textbf{Restricted-Strongly-Convex:} there exists $C_\mathrm{min}>0$, such that the smallest eigenvalue of $H_\nu^{11}$ is larger than $C_\mathrm{min}$;
    \item\textbf{$\nu$-Incoherence Condition:} there exists a parameter $\chi_\nu\in (0, 1]$, such that $\|H_\nu^{01} [H_\nu^{11}]^{-1}\|_\infty \le 1-\chi_\nu.$
\end{itemize}
    Let the columns of $X$ be normalized as $\max_{i\in [1:p]}\frac{\|x_i\|_2}{\sqrt{n}}\le 1$. Then there exists some constant $C>0$, such that for the sequence of $\{\lambda_n\}$ satisfying $\lambda_n > \frac{C}{\chi_\nu}\sqrt{\frac{\sigma^2\log m}{n}}$, the following holds with probability larger than $1-4e^{-C n\lambda_n^2}$.
    \begin{enumerate}
        \item (No-false-positive) Split LASSO \eqref{eq:split_lasso} has a unique solution $(\hat\beta, \hat\gamma)\in \R^p \times\R^m$ without false positives w.r.t. $\gamma$. 
        \item (Sign-consistency) In addition, $\hat\gamma$ recovers the sign of $\gamma^*$, if there holds
        \begin{equation}\label{eq: min snr}
            \min_{i\in \S_1}{\gamma^*_i}> \lambda_n \nu \left[\frac{\sigma}{2C_\mathrm{min}}+\| [H_\nu^{11}]^{-1}\|_\infty\right].
        \end{equation}
    \end{enumerate}
\end{proposition}

For the $\nu$-Incoherence Condition of Proposition \ref{thm: sign consistency front}, one can observe that as $\nu \to \infty$, $H_\nu = I_m - \frac{1}{\nu}D[\Sigma_X+L_D]^{-1}D^T \succeq I_m - \frac{1}{\nu}D[\Sigma_X]^{-1}D^T \to I_m$, therefore $H_\nu^{01}\to 0_{|\S_0|\times|\S_1|}$, while $H_\nu^{11}\to I_{|S_1|}$. In this situation, the term $\|H_\nu^{01} [H_\nu^{11}]^{-1}\|_\infty$ drops to zero, and the incoherence condition is satisfied with arbitrarily large $\chi_\nu\to 1$. 

On the other hand, however, the increase of $\nu$ makes it harder to meet the condition on the minimal signal-noise-ratio \eqref{eq: min snr}. This may cause a potential loss in selecting weak nonnull features.

Hence a sufficiently large $\nu$ will ensure the incoherence condition for model selection consistency such that strong nonnull features will be selected earlier on the Split LASSO path than the nulls, at the cost of possibly losing weak nonnull features. A good power must rely on a proper choice of $\nu$ for the trade-off. In simulation Section \ref{Sec: simulation results}, we indeed observe that the selection power of Split Knockoffs undergoes a first increase then decrease trend as $\nu$ grows. In practice, one may apply the cross validation over $(\nu, \lambda)$ on the Split LASSO path with subset data $\D_1$, to maximize the power for the optimal intercept estimator $\widehat{\beta}(\lambda)=\widehat\beta_{\hat\nu, \hat\lambda}$. Equipped with the FDR control for all $\nu>0$ in Theorem \ref{theorem: fdr}, it maximizes the empirical power with a desired FDR. Such an empirical strategy is validated by simulation experiments in Section \ref{Sec: simulation results} and renders satisfied results in the study of Alzheimer's Disease in Section \ref{Sec: applications}.

\section{Simulation Experiment}

\label{Sec: simulation results}

In this section, we show by several simulation experiments that our proposed Split Knockoff method performs well with transformational sparsity. Particularly, under the cross-validation optimal choice of $\widehat{\beta}(\lambda) = \widehat\beta_{\hat\nu, \hat\lambda}$ and $\nu=\hat\nu$, Split Knockoffs achieve both desired FDR control and high selection power in all the three choices of $W$ statistics and respective selectors. 
The simulation experiments presented in this section are under the basic setting where $n_2\ge m+p$, while the experiments under the high dimensional setting will be given in Section \ref{sec: hd_simulation}.

\subsection{Experimental Setting}\label{sec: simulation_settings}

In model \eqref{eq: model}, we generate $X\in \mathbb R^{n\times p}$ ($n=500$ and $p=100$) i.i.d. from $\Nm(0_p, \Sigma)$, where $\Sigma_{i,i}=1$ and $\Sigma_{i,j}=c^{|i-j|}$ for $i\neq j$, with feature correlation $c=0.5$. Define $\beta^*\in\mathbb R^p$ by
\begin{equation*}
    \beta_i^*:=\left\{
    \begin{array}{ccl}
        1   &   & i \le 20,\ i \equiv 0, -1 (\mathrm{mod}\ 3),\\
        0   &   & \mathrm{otherwise}.
    \end{array} \right.
\end{equation*}
Then $n$ linear measurements are generated by
$$y = X \beta^* + \varepsilon,$$
where $\varepsilon\in \mathbb R^n$ is generated i.i.d. from $\Nm(0, 1)$.

For transformational sparsity, we need to specify the linear transformer $D$ such that $\gamma^*=D\beta^*$, where $\gamma^*$ is sparse. Our choice of $\beta^*$ has two types of transformational sparsity that lead to the following three choices of $D$.  
\begin{itemize}
    \item $\beta^*$ is sparse with many zero elements such that we can take $D_1=I_p$ as the identity matrix where $m = p$. 
    \item $\beta^*$ is a uni-dimensional piecewise constant function such that we can take $D_2$  as the 1-D graph difference operator on a line, i.e. $D_2\in \mathbb R^{(p-1)\times p}$, $D_2(i, i)=1$, $D_2(i, i+1)=-1$ for $i=1,\cdots, p-1$, and $D_2(i, j) = 0$ for other pairs of $(i, j)$, where in this case $m = p-1<p$.
    \item Combining both cases, $\beta^*$ is a sparse piecewise constant function such that we can take $D_3=\left[\begin{array}{c}
         D_1 \\
         D_2
    \end{array}\right]\in \R^{(2p-1)\times p}$, where in this case $m=2p-1>p$.
\end{itemize}

In simulation experiments, we use \url{glmnet} package \citep{friedman2010regularization, simon2011regularization} to compute regularization paths for Split LASSO, etc. 
For the data splitting, we randomly split the dataset $\D = (X, y)$ into two parts $\D_1 = (X_1, y_1)$ and $\D_2 = (X_2, y_2)$ with $n_1$ and $n_2$ samples respectively where $n_1 = 200$ and $n_2 = 300$. The performance of the selection power is presented together with the performance of the FDR control, where the selection power is defined as $\mathrm{Power} = \frac{|S_1\cap \hat{S}|}{|S_1|}$.

For the first choice of $\widehat{\beta}(\lambda)$ in Section \ref{sec: intercept est}, we take $\widehat{\beta}(\lambda)$ as $\widehat{\beta}_\nu(\lambda)$, the solution path with respect to $\widehat{\beta}(\lambda)$ in the $\nu$-Split LASSO regularization path with dataset $\D_1=(X_1, y_1)$. For the regularization paths calculated in Split Knockoffs, we take $\log\lambda$ from a grid between 0 and -6 with a step size $h_\lambda = 0.01$.

For the second choice of $\widehat{\beta}(\lambda)$ in Section \ref{sec: intercept est}, we take $\widehat{\beta}(\lambda)$ as a fixed cross validation optimal estimator $\widehat{\beta}_{\hat\nu, \hat{\lambda}}$ 
in the Split LASSO path with dataset $\D_1=(X_1, y_1)$, screening on $\log\nu$ from a grid between 0 and 2 with a step size $0.4$, and $\log\lambda$ from a grid between 0 and -8 with a step size $0.4$. For the regularization paths of feature and knockoff significance in Equation \eqref{eq:gamma} and Equation \eqref{eq:t_gamma}, we take $\log\lambda$ from a grid between 0 and -6 with a step size $h_\lambda = 0.01$.

\subsection{Performance of Split Knockoffs}

\label{sec: plots for sk}

\begin{figure}[!ht]
\centering
\subfigure[$\Ws$ in $D_1$]{
\begin{minipage}[t]{0.3\textwidth}
\centering
\includegraphics[width=\textwidth]{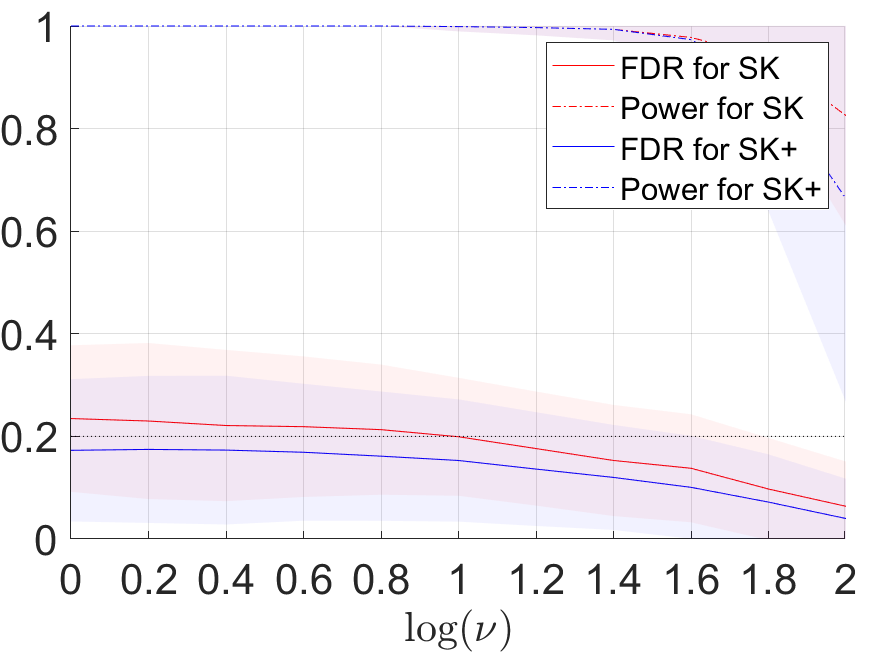}
\end{minipage}%
}%
\subfigure[$\Ws$ in $D_2$]{
\begin{minipage}[t]{0.3\textwidth}
\centering
\includegraphics[width=\textwidth]{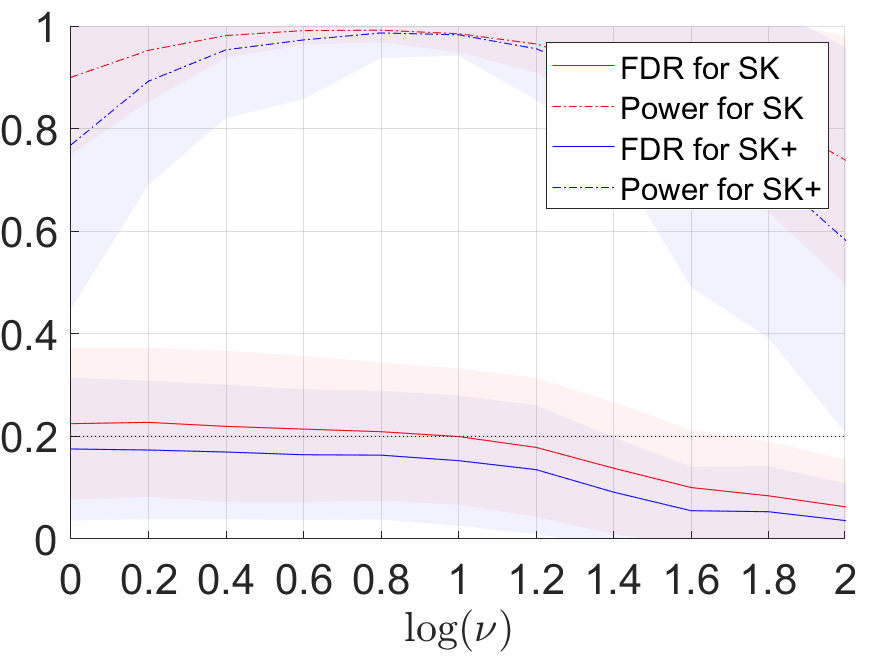}
\end{minipage}%
}%
\subfigure[$\Ws$ in $D_3$]{
\begin{minipage}[t]{0.3\textwidth}
\centering
\includegraphics[width=\textwidth]{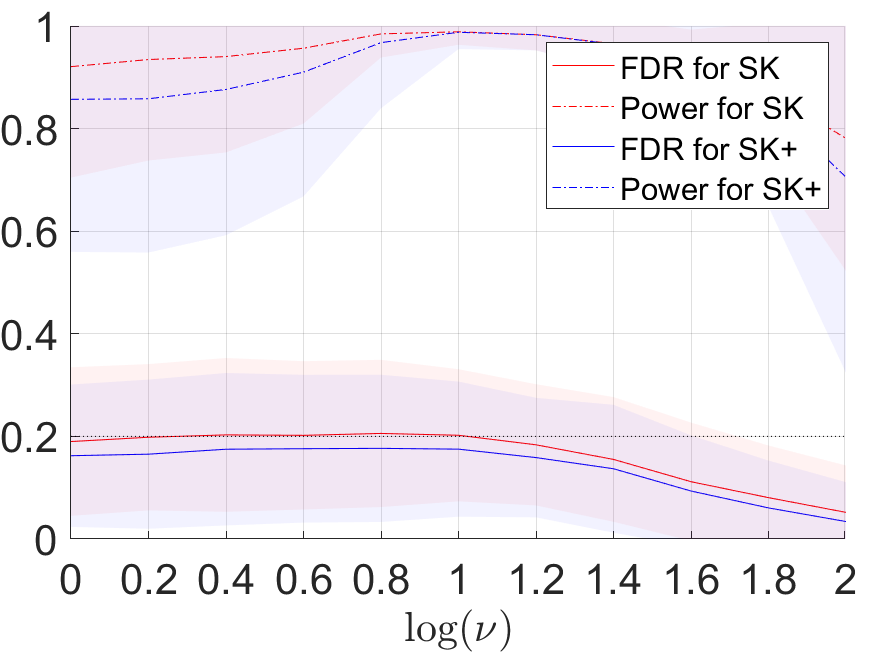}
\end{minipage}%
}%

\centering
\subfigure[$\Wst$ in $D_1$]{
\begin{minipage}[t]{0.3\textwidth}
\centering
\includegraphics[width=\textwidth]{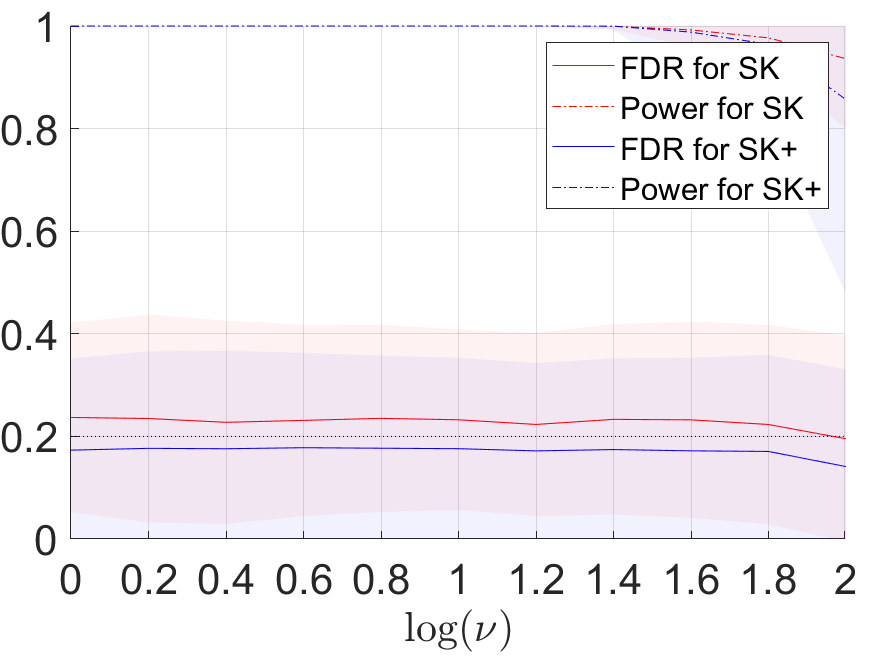}
\end{minipage}%
}%
\subfigure[$\Wst$ in $D_2$]{
\begin{minipage}[t]{0.3\textwidth}
\centering
\includegraphics[width=\textwidth]{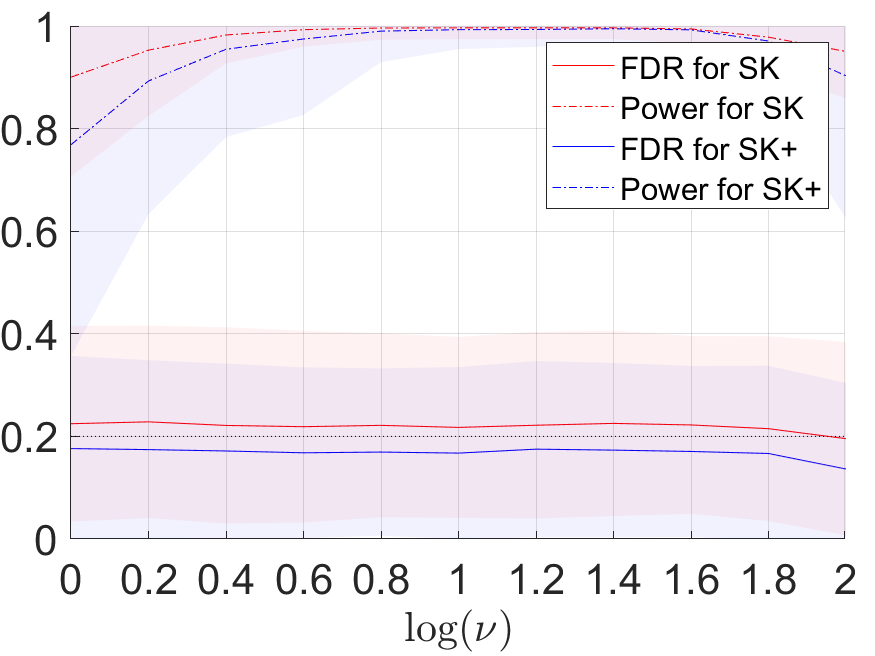}
\end{minipage}%
}%
\subfigure[$\Wst$ in $D_3$]{
\begin{minipage}[t]{0.3\textwidth}
\centering
\includegraphics[width=\textwidth]{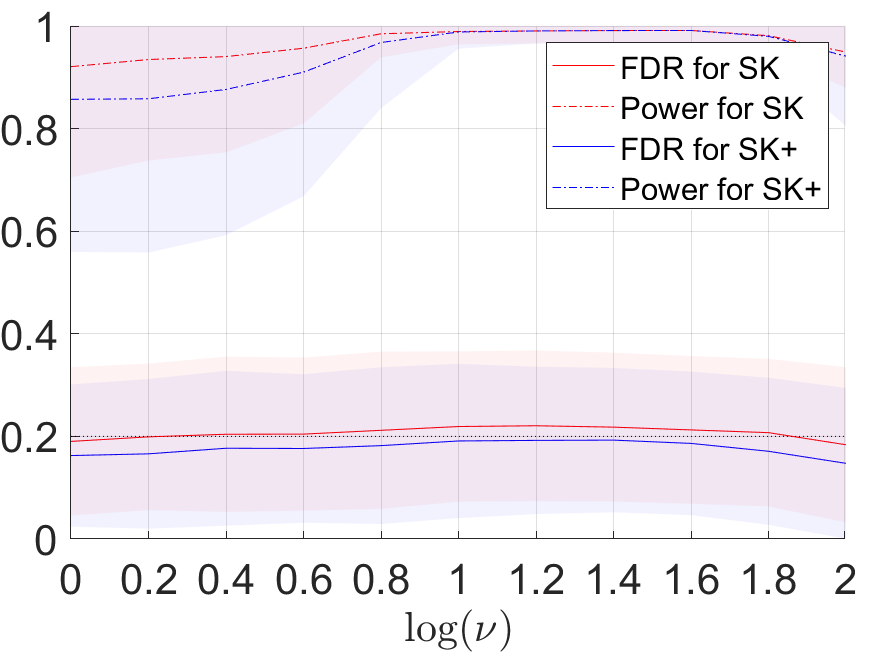}
\end{minipage}%
}%

\centering
\subfigure[$\Wbc$ in $D_1$]{
\begin{minipage}[t]{0.3\textwidth}
\centering
\includegraphics[width=\textwidth]{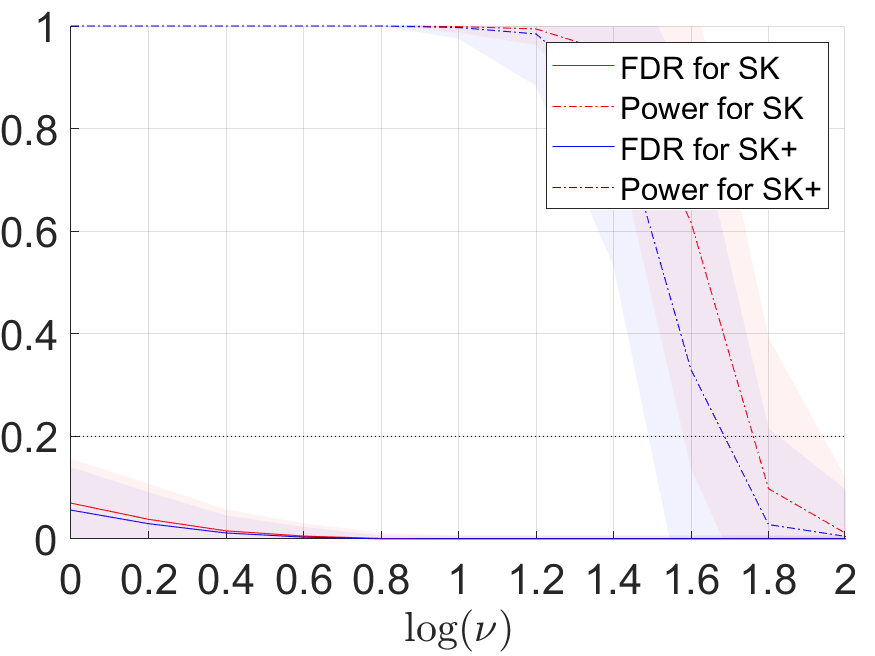}
\end{minipage}%
}%
\subfigure[$\Wbc$ in $D_2$]{
\begin{minipage}[t]{0.3\textwidth}
\centering
\includegraphics[width=\textwidth]{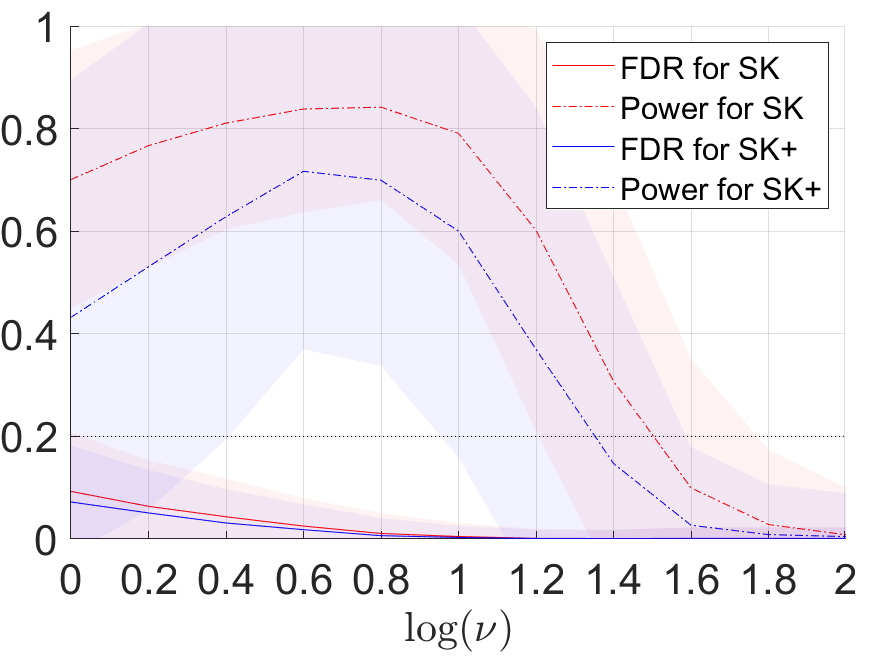}
\end{minipage}%
}%
\subfigure[$\Wbc$ in $D_3$]{
\begin{minipage}[t]{0.3\textwidth}
\centering
\includegraphics[width=\textwidth]{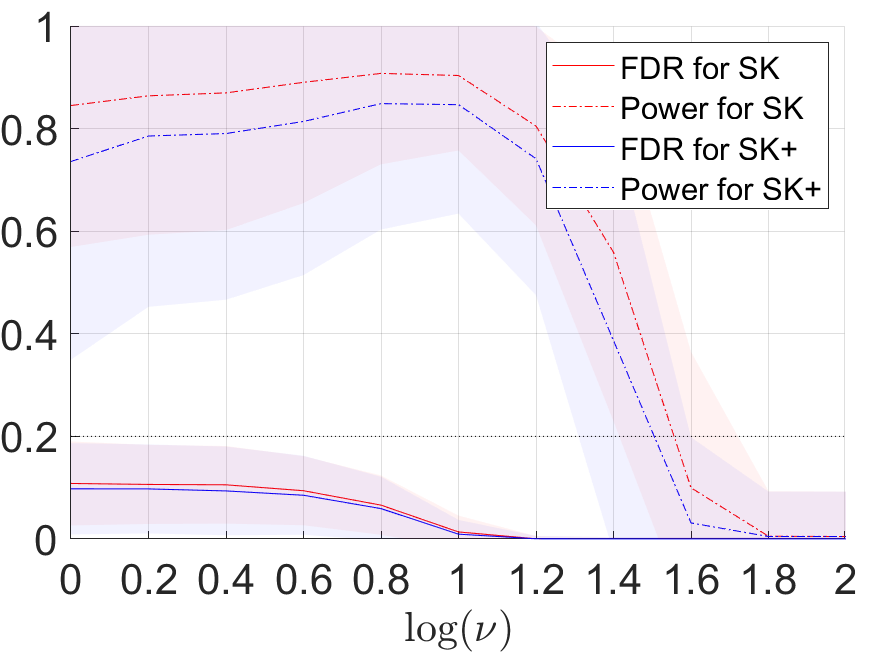}
\end{minipage}%
}%

\caption{The performance of Split Knockoffs: FDR and Power for $q=0.2$. $\widehat{\beta}(\lambda)$ is taken as $\widehat{\beta}_\nu(\lambda)$, the solution path of $\widehat{\beta}(\lambda)$ in the $\nu$-Split LASSO path. The curves in the figures represent the average performance of Split Knockoffs in FDR and Power in 200 simulation instances, while the shaded areas represent the 80\% confidence intervals truncated to the range $[0, 1]$.}
\label{fig: simulation with beta(lambda)}
\end{figure}

In Figure \ref{fig: simulation with beta(lambda)}, we plot the performance of Split Knockoffs on $\log(\nu)$ between 0 and 2 with a step size 0.2, where $\widehat{\beta}(\lambda)$ is taken as the solution path $\widehat{\beta}_\nu(\lambda)$ in the $\nu$-Split LASSO path. In Figure \ref{fig: simulation with beta hat}, we plot the performance of Split Knockoffs on $\log(\nu)$ between 0 and 2 with a step size 0.2, where $\widehat{\beta}(\lambda)$ is taken as a fixed cross validation optimal estimator $\widehat{\beta}_{\hat\nu, \hat{\lambda}}$ on the Split LASSO paths. 

In these cases, the FDR of Split Knockoffs are all under control; while the performance in the selection power differs from one to another. The cross-validation optimal estimator choice of $\widehat{\beta}(\lambda)=\widehat\beta_{\hat\nu, \hat\lambda}$ shown in Figure \ref{fig: simulation with beta hat}, clearly improves the selection power of Split Knockoffs compared with the $\nu$-Split LASSO solution path choice of $\widehat{\beta}(\lambda) = \widehat{\beta}_\nu(\lambda)$ shown in Figure \ref{fig: simulation with beta(lambda)}.

\begin{figure}[!ht]
\centering
\subfigure[$\Ws$ in $D_1$]{
\begin{minipage}[t]{0.3\textwidth}
\centering
\includegraphics[width=\textwidth]{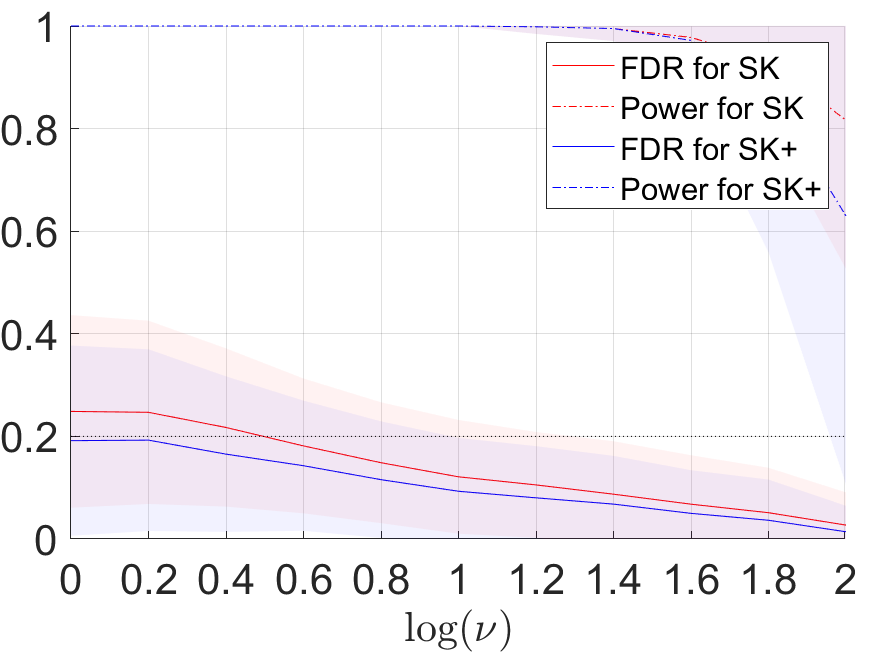}
\end{minipage}%
}%
\subfigure[$\Ws$ in $D_2$]{
\begin{minipage}[t]{0.3\textwidth}
\centering
\includegraphics[width=\textwidth]{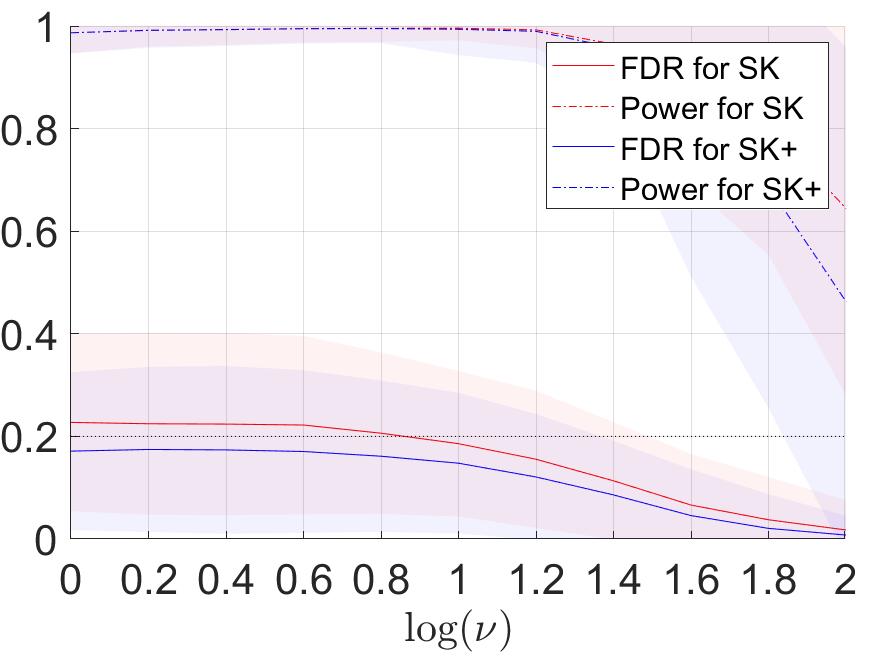}
\end{minipage}%
}%
\subfigure[$\Ws$ in $D_3$]{
\begin{minipage}[t]{0.3\textwidth}
\centering
\includegraphics[width=\textwidth]{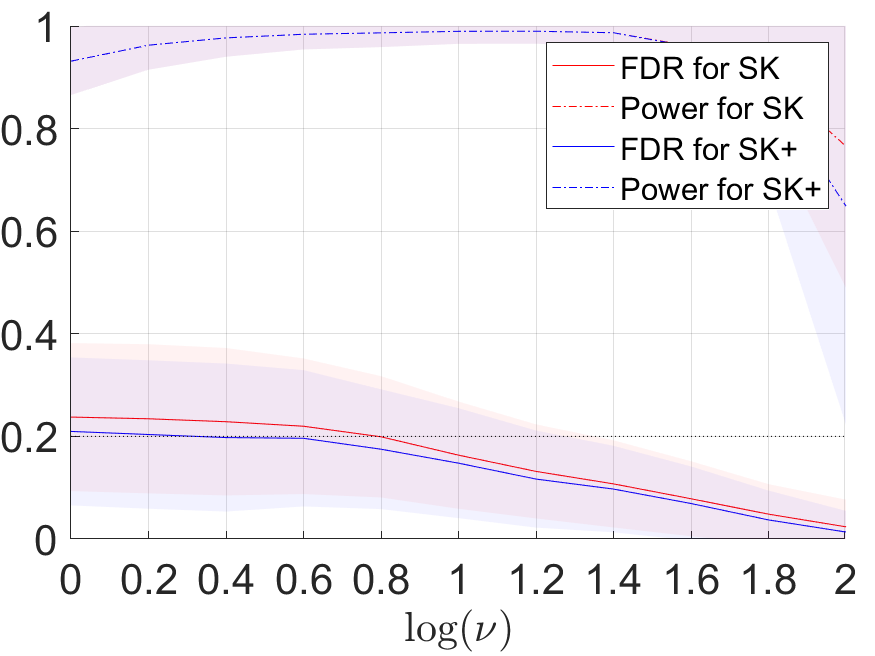}
\end{minipage}%
}%

\centering
\subfigure[$\Wst$ in $D_1$]{
\begin{minipage}[t]{0.3\textwidth}
\centering
\includegraphics[width=\textwidth]{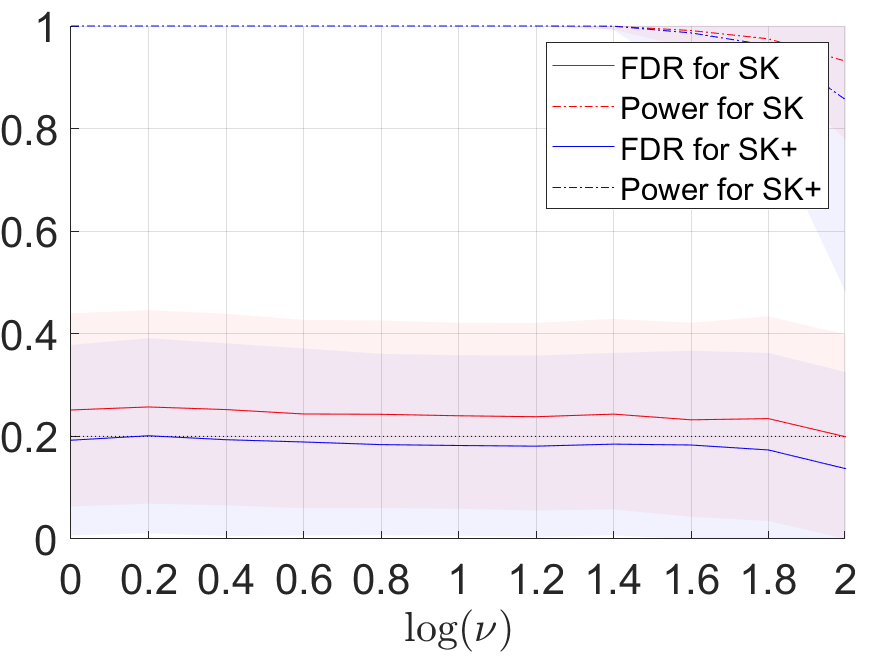}
\label{fig: wst1}
\end{minipage}%
}%
\subfigure[$\Wst$ in $D_2$]{
\begin{minipage}[t]{0.3\textwidth}
\centering
\includegraphics[width=\textwidth]{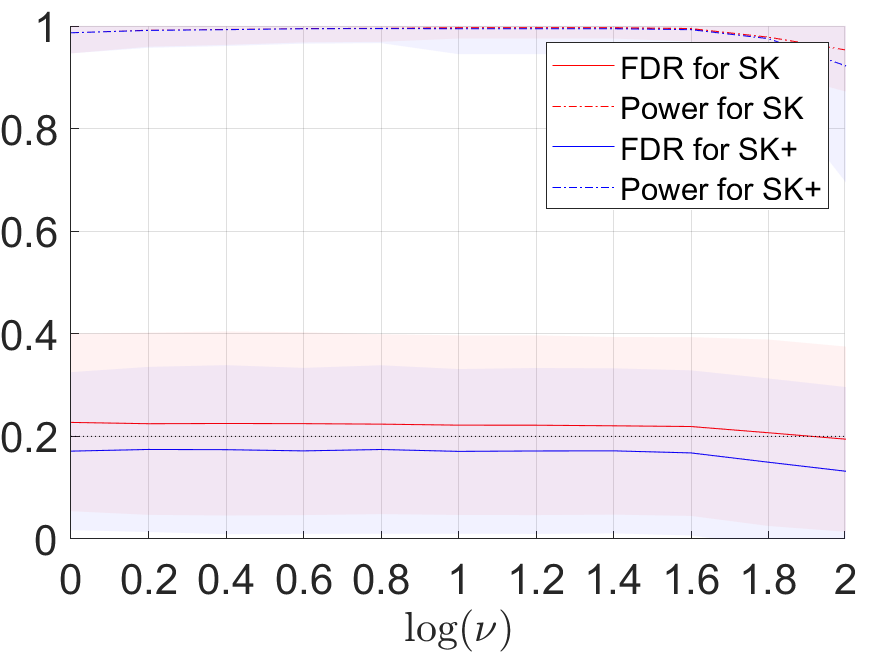}
\label{fig: wst2}
\end{minipage}%
}%
\subfigure[$\Wst$ in $D_3$]{
\begin{minipage}[t]{0.3\textwidth}
\centering
\includegraphics[width=\textwidth]{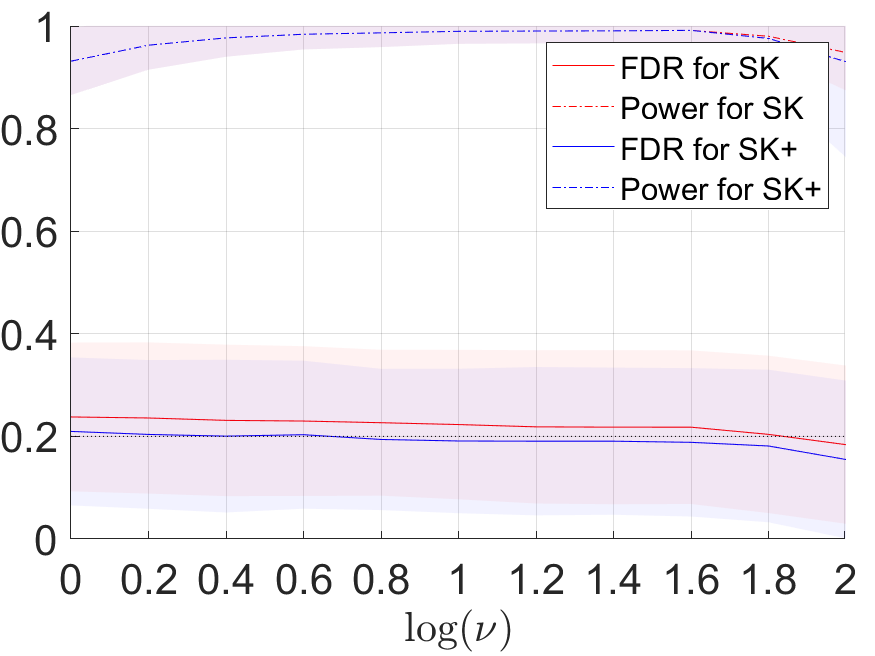}
\label{fig: wst3}
\end{minipage}%
}%

\centering
\subfigure[$\Wbc$ in $D_1$]{
\begin{minipage}[t]{0.3\textwidth}
\centering
\includegraphics[width=\textwidth]{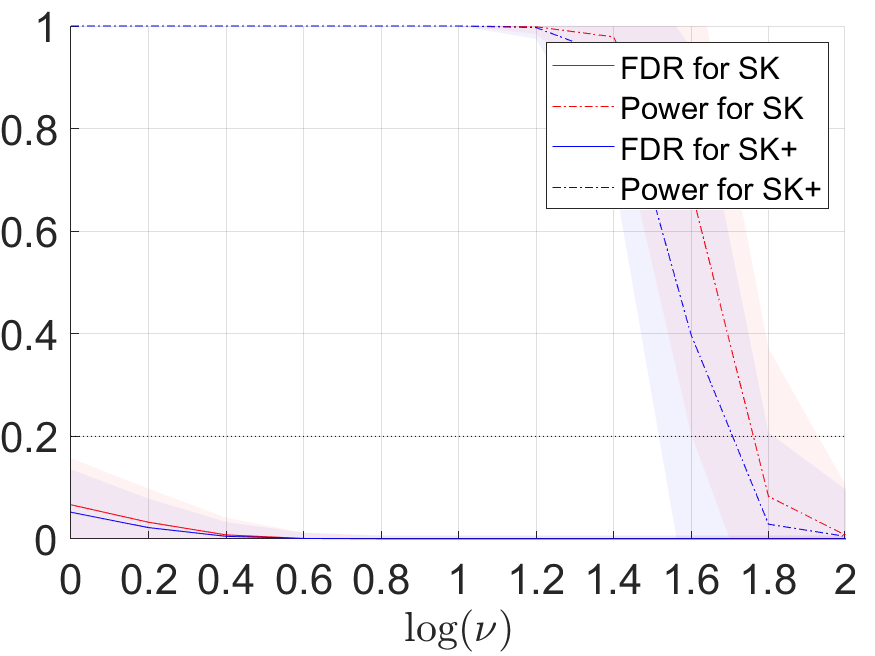}
\end{minipage}%
}%
\subfigure[$\Wbc$ in $D_2$]{
\begin{minipage}[t]{0.3\textwidth}
\centering
\includegraphics[width=\textwidth]{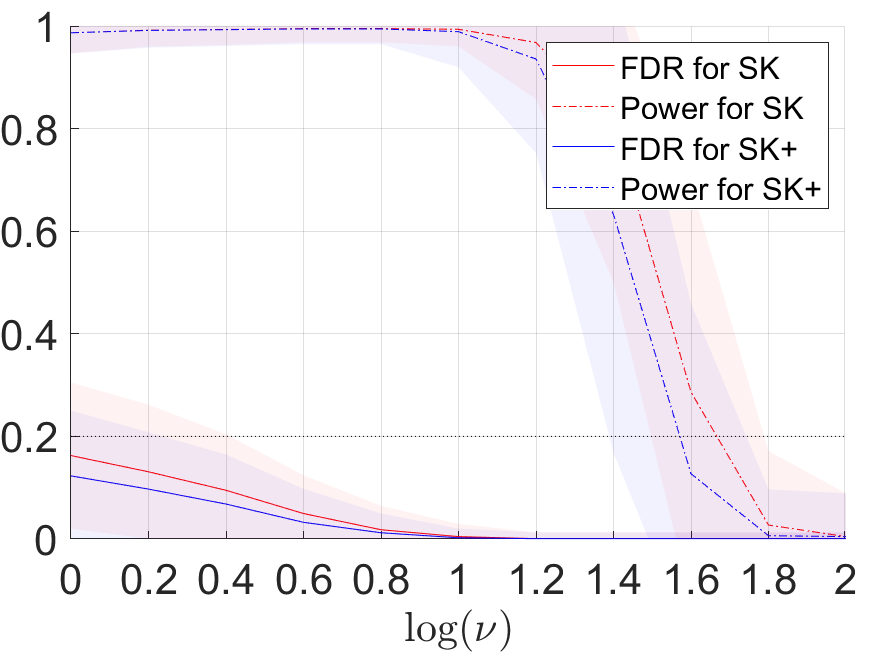}
\end{minipage}%
}%
\subfigure[$\Wbc$ in $D_3$]{
\begin{minipage}[t]{0.3\textwidth}
\centering
\includegraphics[width=\textwidth]{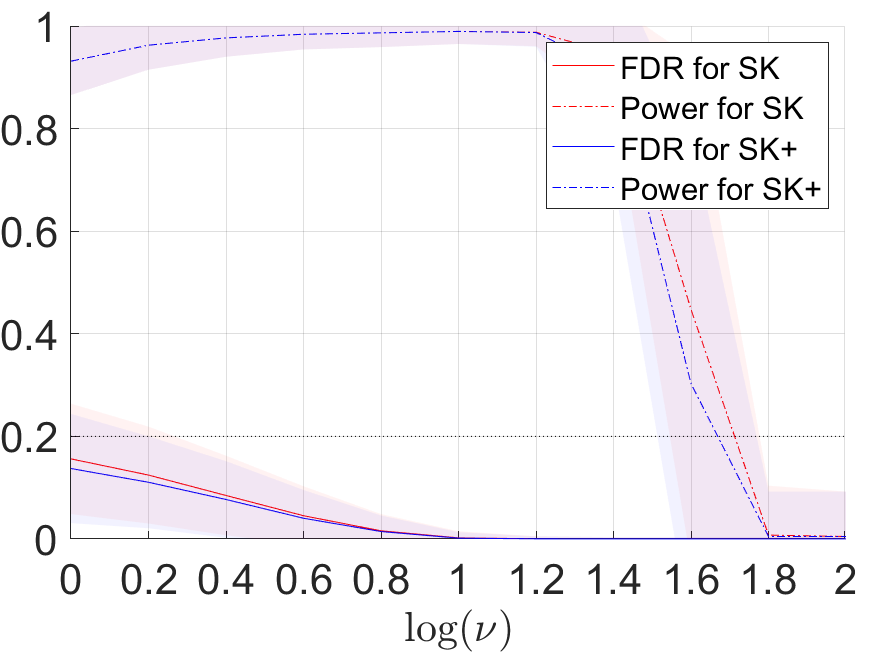}
\end{minipage}%
}%

\caption{The performance of Split Knockoffs: FDR and Power for $q=0.2$. $\widehat{\beta}(\lambda)$ is taken as a fixed cross validation optimal estimator $\widehat{\beta}_{\hat\nu, \hat{\lambda}}$. The curves in the figures represent the average performance of Split Knockoffs in FDR and Power in 200 simulation instances, while the shaded areas represent the 80\% confidence intervals truncated to the range $[0, 1]$.}
\label{fig: simulation with beta hat}
\end{figure}

In the family of $W$-statistics defined in Section \ref{sec: w stat}, $\Wst$ shows a less conservative FDR  and slightly better selection power compared with the others in Figure \ref{fig: simulation with beta(lambda)} in the tested regime of $\nu$. 
Meanwhile, $\Wbc$ shows the most conservative performance compared with $\Ws$ and $\Wst$, which trades the selection power for a better empirical FDR control. Such a phenomenon is explained by Proposition \ref{prop: relations} and \ref{prop: inequality of w statistics}.

In all cases, particularly for $D_2$ and $D_3$, the selection power of Split Knockoffs shows a first increasing then decreasing trend with respect to $\nu$. This is discussed by Proposition \ref{thm: sign consistency front} in Section \ref{sec: sign consistency short}, that increasing $\nu$ will help discover strong nonnull features at the cost of possibly losing weak ones. Thus a good power relies on a proper choice of $\nu$ for the trade-off. In practice, cross validation renders satisfactory results, which is validated in the next section.

\subsection{Comparisons between Split Knockoffs and Knockoffs}

In this section, we conduct more simulation experiments on the cross validation optimal estimator choice of $\widehat{\beta}(\lambda)=\widehat\beta_{\hat\nu, \hat\lambda}$, which gives higher selection power compared with the $\nu$-Split LASSO solution path choice of $\widehat{\beta}(\lambda) = \widehat{\beta}_\nu(\lambda)$, as shown in Section \ref{sec: plots for sk}. 
Here, for the choice of $\nu$ in the calculation of the feature significance and knockoff significance by Equation \eqref{eq:gamma} and Equation \eqref{eq:t_gamma}, we take the cross validation optimal choice of $\nu = \hat\nu$ to maximize the selection power of Split Knockoffs.

We show in Table \ref{table: comparison} the performance of Split Knockoffs with all the three versions of $W$ statistics, under the above choice of $\widehat{\beta}(\lambda)$ and $\nu$ in the simulation settings described in Section \ref{sec: simulation_settings}. We also provide comparisons between Split Knockoffs and standard Knockoffs when both are applicable (the case of $D_1$ and $D_2$). In particular, for the 1-D fused LASSO case that $D_2\in \mathbb R^{(p-1)\times p}$ is the graph difference operator on a line, we make Knockoffs applicable by introducing the induced LASSO problem from the generalized LASSO problem (see Section \ref{sec:genlasso} for details). We note that the Knockoffs are implemented in the whole dataset without data splitting.

\begin{table}[!ht]
    \caption{FDR and Power for Knockoffs and Split Knockoffs in simulation experiments ($q = 0.2$). The intercept $\widehat{\beta}(\lambda)$ for Split Knockoffs is taken as a fixed cross-validation optimal estimator $\widehat{\beta}_{\hat\nu, \hat{\lambda}}$, with the $\nu$ for calculating the feature and knockoff significance taken as $\hat\nu$. In this table, we present the average performance of Knockoffs and Split Knockoffs in FDR and Power, together with the standard deviations in 200 simulation instances. For shorthand notations, we use "SK(+)" to refer to "Split Knockoff(+)".}
    \centering
    \begin{tabular}{c|cccc}
    \hline
        Performance & Knockoff & SK with $\Ws$ & SK with $\Wst$ & SK with $\Wbc$\\
        \hline
        FDR in $D_1$ & 0.2233 & 0.2481 &  0.2517 &  0.0667 \\
        ~& $\pm$0.1584 & $\pm$0.1465 & $\pm$0.1465 & $\pm$0.0706 \\
        Power in $D_1$ & 1.0000 & 1.0000 & 1.0000 & 1.0000 \\
        ~& $\pm$0.0000 & $\pm$0.0000 & $\pm$0.0000 & $\pm$0.0000 \\
        \hline
        FDR in $D_2$ & 0.2813 & 0.2206 &  0.2206 & 0.1492  \\
        ~& $\pm$0.1771 & $\pm$0.1341 & $\pm$0.1341 & $\pm$0.1130 \\
        Power in $D_2$ & 0.5571 & 0.9886 &  0.9886  & 0.9886\\
        ~& $\pm$0.3231 & $\pm$0.0299 & $\pm$0.0299 & $\pm$0.0299 \\
        \hline
        FDR in $D_3$ & N/A & 0.2374 & 0.2386 &  0.1529 \\
        ~& N/A & $\pm$0.1124 & $\pm$0.1133 & $\pm$0.0859 \\
        Power in $D_3$ & N/A & 0.9352  & 0.9352 &  0.9352 \\
        ~& N/A & $\pm$0.0509 & $\pm$0.0509 & $\pm$0.0509 \\
    \hline
    \hline
        Performance & Knockoff+ & SK+ with $\Ws$ & SK+ with $\Wst$ & SK+ with $\Wbc$ \\
        \hline
        FDR in $D_1$ & 0.1787 &   0.1914 &  0.1929  & 0.0521  \\
        ~& $\pm$0.1487 & $\pm$0.1444 & $\pm$0.1447 & $\pm$0.0649 \\
        Power in $D_1$ & 1.0000& 1.0000 & 1.0000 & 1.0000 \\
        ~& $\pm$0.0000 & $\pm$0.0000 & $\pm$0.0000 & $\pm$0.0000\\
        \hline
        FDR in $D_2$ & 0.1649 & 0.1709 &  0.1709 & 0.1085 \\
        ~& $\pm$0.1985 & $\pm$0.1217  & $\pm$0.1217  & $\pm$0.0984 \\
        Power in $D_2$ & 0.2914 & 0.9886 &  0.9886  & 0.9886 \\
        ~& $\pm$0.3669 & $\pm$0.0299 & $\pm$0.0299 & $\pm$0.0299 \\
        \hline
        FDR in $D_3$ & N/A & 0.2100  & 0.2110 & 0.1347  \\
        ~& N/A & $\pm$0.1124 & $\pm$0.1128 & $\pm$0.0845 \\
        Power in $D_3$ & N/A& 0.9352  & 0.9352 &  0.9352 \\
        ~& N/A & $\pm$0.0509 & $\pm$0.0509 & $\pm$0.0509 \\
    \hline
      \end{tabular}
      \label{table: comparison}
\end{table}

As shown in Table \ref{table: comparison}, Split Knockoffs(+) with the cross-validation optimal choice $\widehat{\beta}(\lambda)=\widehat\beta_{\hat\nu, \hat\lambda}$ achieve desired FDR in the family of $W$ under all cases of transformational sparsity. In particular, Split Knockoffs with $\Wbc$ is the most conservative with a lower FDR and the same power. In all cases, (Split) Knockoff+ has a better control in the FDR compared with (Split) Knockoff, at the cost of a potential loss in selection power. 

Compared with standard Knockoffs, Split Knockoffs exhibits higher power in the case of $D_2$. In this case, the correlated $X$ and $D_2$ destroy the incoherence condition of induced the LASSO problem (see Equation \eqref{tran} in Section \ref{sec:genlasso} for discussions) such that standard Knockoffs suffer from losing the model selection consistency in the case of $D_2$. It hurts the selection power for standard Knockoffs. In contrast, for Split Knockoffs, the improved $\nu$-incoherence condition of Split LASSO as discussed in Section \ref{sec: sign consistency short} leads to better model selection consistency on regularization paths which helps improve the selection power.

\section{Application: Alzheimer's Disease}
\label{Sec: applications}

In this experiment, we apply the Split Knockoff method to study lesion regions of brains and their connections in  Alzheimer's Disease (AD), which is the major cause of dementia and has attracted increasing attention in recent years. 

\subsection{Dataset}

The data is obtained from ADNI (\url{http://adni.loni.ucla.edu}) dataset, acquired by structural Magnetic Resonance Imaging (MRI) scan. In total, the dataset contains $n=752$ samples, with 126 AD, 433 Mild Cognitive Impairment (MCI), and 193 Normal Controls (NC). For each image, we implement the Dartel VBM \citep{ashburner2007fast} for pre-processing, followed by the toolbox \emph{Statistical Parametric Mapping} (SPM) for segmentation of gray matter (GM), white matter (WM), and cerebral spinal fluid (CSF). Then we use Automatic Anatomical Labeling (AAL) atlas \citep{tzourio2002automated} to partition the whole brain into $p = 90$ Cerebrum brain anatomical regions, with the volume of each region (summation of all GMs in the region) provided.

We use $X \in \mathbb{R}^{n \times p}$ to denote the design matrix, with each element $X_{i,j}$ representing the column-wise normalized volume of the region $j$ in the subject $i$'s brain. The response variable vector $y \in \mathbb{R}^{n}$ denotes the Alzheimer's Disease Assessment Scale (ADAS), which was originally designed to assess the severity of cognitive dysfunction \citep{rosenwg1984scale} and was later found to be able to clinically distinguish the diagnosed AD from normal controls \citep{zec1992alzheimer}. We test two types of transformational sparsity:
\begin{enumerate}
    \item[(a)] $D=I_{p}$, the identity matrix, for selecting the regions that account for the Alzheimer's disease, where in this case $m=p=90$;
    \item[(b)] $D$ is the graph difference operator matrix on the brain region connectivity graph, for selecting the connections of regions with large activation changes accounting for the disease, where in this case $m=463>p=90$;
\end{enumerate}
In the experiments below, we choose the Split Knockoff statistics $\Wst$ to demonstrate the results as it exhibits the best selection power in simulations, where $\widehat{\beta}(\lambda)$ is taken as a fixed cross validation optimal estimator $\widehat{\beta}_{\hat\nu, \hat{\lambda}}$, screening over $\log\nu$ on a grid between 0 and 2 with a step size $0.4$, and $\log\lambda$ on a grid between 0 and -8 with a step size $ 0.4$. For the regularization paths in Equation \eqref{eq:gamma} and \eqref{eq:t_gamma}, we take $\log\lambda$ from an arithmetic sequence between 0 and -6 with step size $h_\lambda = 0.01$. The dataset $\D = (X, y)$ is randomly split into two parts $\D_1 = (X_1, y_1)$ and $\D_2 = (X_2, y_2)$ with $n_1$ and $n_2$ samples respectively, where $n_1 = 150$ and $n_2 = n-n_1 = 602$. We show the experimental results on one random data split instance in this section, where for multiple instances of data splits the frequency plots on the most frequently selected regions and connections are provided as supplementary materials in Section \ref{sec: freq}. 
\subsection{Region Selection}
In this experiment, consider the selection of lesion regions, with $D = I_p$. The target FDR level is set at $q = 0.2$.

\begin{table}[!ht]
    \caption{Selected Regions by Split Knockoff on Alzheimer's Disease ($q = 0.2$). We choose the $W$ statistics for Split Knockoff as $\Wst$, with $\widehat{\beta}(\lambda)$ taken as a fixed cross validation optimal estimator $\widehat{\beta}_{\hat\nu, \hat{\lambda}}$, where the minimal cross validation loss lies on the choice $\log \hat{\nu}=0$. }
    \centering
    \begin{tabular}{c|cc}
    \hline
        Region & \multicolumn{2}{c}{Split Knockoff with $\log\nu$}\\
        ~  &[0, 1.8]& [1.9, 2] \\
        \hline
        Inferior frontal gyrus, opercular part (L) &  $\surd$&  $\surd$\\
        Hippocampus (L) &  $\surd$ & $\surd$ \\
        Hippocampus (R) & $\surd$  & $\surd$\\
        Inferior parietal gyrus (R) &  $\surd$ &  \\
        Middle temporal gyrus (L) & $\surd$ & \\
        \hline
      \end{tabular}
      \label{tab: region}
\end{table}

Our region selection results are summarized in Table \ref{tab: region} where $\log(\nu)$ is between 0 and 2 with a step size 0.1. At $\log\hat\nu=0$ where $\widehat{\beta}_{\hat\nu, \hat{\lambda}}$ is optimal in terms of cross-validation loss, our algorithm selects five regions that are all related to AD and have been reported to suffer from degeneration during disease progression \citep{vemuri2010role, schuff2009mri, karas2007precuneus, greene2010subregions}. Specifically, it was reported in \cite{vemuri2010role} that the hippocampus in both sides are responsible for learning and memory; while the middle temporal gyrus participates in language and memory processing. Moreover, \cite{tyler2005temporal} and \cite{schremm2018cortical} found that the opercular part of the Inferior frontal gyrus may be associated with tone and inflectional morpheme processing. Finally, the inferior parietal gyrus which is associated with motor and sensory, was found to suffer from volume reduction after the stage from normal control (NC) to mild cognitive impairment \citep{greene2010subregions}, {\it i.e.}, an intermediate stage between NC to Alzheimer's Disease.

\subsection{Connection Selection}

\label{sec: connection selection}

 In this experiment, we consider the connection or edge detection, in which each connection is associated with two adjacent brain regions. Formally speaking, we set $D$ as the graph gradient (difference) operator on the graph $G=(V,E)$ where $V$ denotes the vertex set of brain regions and $E$ denotes the (oriented) edge set of region pairs in neighbor, such that $D(\beta)(i,j) = \beta_i - \beta_j$ for $(i,j)\in E$. Thus $D$ is the graph gradient operator that measures the differences in the impacts of Alzheimer's Disease between adjacent regions in brains. We shall expect a properly selected edge will connect regions of high contrast variations in the degree of atrophy during the disease progression. For example, a selected edge may connect an atrophy region in the brain significantly impaired by Alzheimer's Disease to another region which is less influenced by the disease.

\begin{table}[!ht]
    \caption{Selected connections by Split Knockoffs on Alzheimer's Disease ($q = 0.2$). We choose the $W$ statistics for Split Knockoff as $\Wst$, with $\widehat{\beta}(\lambda)$ taken as a fixed cross validation optimal estimator $\widehat{\beta}_{\hat\nu, \hat{\lambda}}$, where the minimal cross-validation loss lies on the choice $\log\hat{\nu}=0.4$.}
    \centering
    \resizebox{\textwidth}{!}{
    \begin{tabular}{cc|cccc}
        \hline
        \multicolumn{2}{c|}{Connection} & \multicolumn{4}{c}{Split Knockoff with $\log\nu$}\\
        Region 1 & Region 2  & [0, 0.1] & [0.2, 1.6] & 1.7 & [1.8, 2]\\
        \hline
        Hippocampus (L) & Posterior cingulate gyrus (L)
        & $\surd$ & $\surd$ & $\surd$ &  $\surd$ \\
        Hippocampus (L) & Lingual gyrus (L)
        & $\surd$ & $\surd$ & $\surd$ & $\surd$ \\
        Hippocampus (L) & Fusiform gyrus (L)
        & $\surd$ &  & $\surd$ &   \\
        Hippocampus (L) & Precuneus (L)
        & $\surd$ & $\surd$ & $\surd$ & $\surd$ \\
        Hippocampus (R) & Insula (R) 
        & $\surd$ &  &  &  \\
        Hippocampus (R) & Lingual gyrus (R)  
        & $\surd$ & $\surd$ & $\surd$ & $\surd$ \\
        Hippocampus (R) & Precuneus (R)
        & $\surd$ & $\surd$ & $\surd$ &  $\surd$ \\
        Hippocampus (R) & Superior temporal gyrus (R) 
        & $\surd$ &  &  &   \\
        Amygdala (R) & Putamen (R) 
        & $\surd$ & $\surd$ & $\surd$ &  $\surd$ \\
        Amygdala (R) & Caudate nucleus (R)  
        & $\surd$ &  &  &   \\
        Middle frontal gyrus, orbital part (L) &  Insula (L) 
        & $\surd$ &  &  &  \\
        \hline
    \end{tabular}
    }
    \label{tab: connection}
\end{table}

For connection selection, our results are summarized in Table \ref{tab: connection} where $\log(\nu)$ is between 0 and 2 with a step size 0.1. In total, there are eleven connections selected, six of which are selected at the cross-validation optimal $\hat{\nu}$. Among all the eleven selected connections, eight of them are associated with the Hippocampus on the left side or the right side, where the Hippocampus is one of the most early affected regions during disease progression \citep{juottonen1999comparative}. Similarly, the selected connections that involve the Amygdala echo the previous findings that the Amygdala is also affected early \citep{knafo2012amygdala}, which can explain the neuropsychiatric symptoms that are commonly observed in mild stages of AD. These studies provide us with references to why these pairs of adjacent regions have different degrees of atrophy. Finally, the connection between Insula (L) and Middle frontal gyrus (L) may be a false discovery, as both regions were reported to be atrophied and this connection is not included in the cross-validation optimal selection. Although it was found that the atrophy happens in the whole region of the Insula while only the sub-region of the middle frontal gyrus \citep{busatto2008voxel}, it is still ambiguous to claim that the involved regions are significantly different in terms of degrees of degeneration. In summary, Figure \ref{fig: ad reg and con} shows the selected regions and connections at cross-validated $\hat\nu$. In the graph, each vertex represents a cerebrum brain region in Automatic Anatomical Labeling (AAL) atlas \citep{tzourio2002automated}, with abbreviations of each region marked in the vertex. A comparison table between the full region names and their abbreviations is now given in Table \ref{tab:name of region}.

\section{Conclusion}

\label{sec: conclusion}

In this paper, we propose Split Knockoffs as a data-adaptive FDR control method for the transformational sparsity recovery in linear regression where a linear transformation of parameters is sparse. By relaxing the linear subspace constraint to its neighborhood in a lifted parameter space, our method has FDR under control and further gains power in improving the model selection consistency conditions on Split LASSO regularization paths. The main theoretical contribution of this paper is that we construct some new supermartingale structures to achieve a theoretical FDR control for Split Knockoffs where exchangeability is broken by the heterogeneous noise bought by the transformation. A high-dimensional generalization is also discussed. In a real-world application to Alzheimer's Disease study with MRI data, Split Knockoffs discover important atrophy lesion regions in the brain and neighboring region connections of high contrast in atrophy variations during disease progression. Future directions include generalizations of our methodology to directional FDR and random designs.

\bibliography{reference}


\newpage
\appendix

\section{Construction of Split Knockoff Copies}

\label{sec: details for copy}

In this section, we provide more details on the construction of Split Knockoff copy matrices. An explicit form of the construction is shown in the following proposition.

\begin{proposition}[Split Knockoff Matrix]
\label{prop: existence of copy}
    If $n_2\ge m+p$, for any vector $\vecs$ in Equation \eqref{eq: copy} satisfying
\begin{align}
    \diag(\vecs)  \succeq 0,\ 2C_\nu-\diag(\vecs) \succeq 0,\label{eq:vecs}
\end{align}
where $C_\nu:=\Sigma_{\gamma,\gamma}-\Sigma_{\gamma,\beta}\Sigma_{\beta,\beta}^{-1}\Sigma_{\beta,\gamma}$, $\Sigma_{\beta,\beta}:=A_\beta^TA_\beta$, $\Sigma_{\beta,\gamma}=\Sigma_{\gamma,\beta}^T:=A_\beta^TA_\gamma$, and $\Sigma_{\gamma,\gamma}:=A_\gamma^TA_\gamma$, there is a valid Split Knockoff matrix for $n_2\ge m+p$, 
\begin{align}
    \tilde{A}_\gamma=A_\gamma(I_{m}-C_\nu^{-1}\diag(\vecs))+A_\beta \Sigma_{\beta,\beta}^{-1}\Sigma_{\beta,\gamma}C_{\nu}^{-1}\diag(\vecs)+\tilde{U}K,
\end{align}
where $\tilde{U}\in \R^{(n_2+m)\times m}$ is the orthogonal complement of $[A_\beta, A_\gamma]$ and $K\in \R^{m\times m}$ satisfies $K^TK=2\diag(\vecs)-\diag(\vecs)C_\nu^{-1}\diag(\vecs)$. 
\end{proposition}
\begin{remark}
The requirement $n_2\ge m+p$ is from the property that $\tilde{U}\in \R^{(n_2+m)\times m}$ is an orthogonal complement of $[A_\beta, A_\gamma]\in \R^{(n_2+m)\times (m+p)}$.
\end{remark}
In the following, we show that any $\vecs$ satisfying Equation \eqref{eq:vecs} will ensure the existence of $\tilde{A}_\gamma$. 
Note that a necessary and sufficient condition for the existence of $\tilde{A}_\gamma$ satisfying Equation \eqref{eq: copy} is
\begin{align*}
    G: = 
    \begin{bmatrix}
        \Sigma_{\beta,\beta} & \Sigma_{\beta,\gamma} & \Sigma_{\beta,\gamma}\\
        \Sigma_{\beta,\gamma} & \Sigma_{\gamma,\gamma} & \Sigma_{\gamma,\gamma}-\diag(\vecs)\\
        \Sigma_{\beta,\gamma} & \Sigma_{\gamma,\gamma}-\diag(\vecs) & \Sigma_{\gamma,\gamma}
    \end{bmatrix}
    \succeq 0.
\end{align*}
This holds if and only if the Schur complement of $\Sigma_{\beta,\beta}$ is positive semi-definite, i.e.
\begin{align}
    \begin{bmatrix}
        C_\nu & C_\nu-\diag(\vecs)\\
        C_\nu-\diag(\vecs) & C_\nu
    \end{bmatrix}
    \succeq 0,\nonumber
\end{align}
which holds if and only if $C_\nu$ and its Schur complement are positive semi-definite, i.e.
\begin{subequations}
\begin{align}
    &C_\nu  \succeq 0,\nonumber\\
    &C_\nu-(C_\nu-\diag(\vecs))C_\nu^{-1}(C_\nu-\diag(\vecs))=2\diag(\vecs)-\diag(\vecs)C_\nu^{-1}\diag(\vecs)
    \succeq 0.\nonumber
\end{align}
\end{subequations}
The equations above are equivalent with Equation \eqref{eq:vecs}. Then one can verify that the construction in Proposition \ref{prop: existence of copy} satisfies Equation \eqref{eq: copy} for any vector $\vecs$ satisfying Equation \eqref{eq:vecs}.

There are various choices for $\vecs=(s_i)$ satisfying \eqref{eq:vecs} for Split Knockoffs. Below we give two typical examples. 

\begin{itemize}
    \item[(a)] (SDP for discrepancy maximization) One can maximize the discrepancy between Knockoffs and its corresponding features by solving the following SDP
\begin{eqnarray*} 
\mbox{maximize} & & \sum_i s_i , \\
\mbox{subject to} & & 0\leq s_i \leq \frac{1}{\nu} \mbox{ and } \frac{\diag(s)}{2} \preceq  C_\nu.
\end{eqnarray*}
\item[(b)] (Equi-correlation) Take $s_i = 2\lambda_{\mathrm{min}}(C_\nu)\land \frac{1}{\nu}$ for all $i\in \{1, 2, \cdots, m\}$, where $\lambda_{\mathrm{min}}(C_\nu)$ represents the minimal eigenvalue of $C_\nu$.
\end{itemize}

\section{Proofs}
\label{sec: proof}

In this section, we will prove the propositions and theories in the main text in the following order. We will first prove Proposition \ref{prop: relations}, the inclusion property of selectors, then Proposition \ref{prop: s to bc}, a preliminary proposition for Theorem \ref{theorem: fdr}. After that, we will provide the proof of our main theorem on the FDR control --- Theorem \ref{theorem: fdr} --- and some additional supporting lemmas. Finally, we will present the proof of Theorem \ref{thm: fdr hd}, the FDR control theorem for Split Knockoffs in high dimensional settings, followed by the proof of Proposition \ref{thm: sign consistency front}, the model selection consistency of Split LASSO.

\subsection{Proof of Proposition \ref{prop: relations}}

\label{sec: proof prop1}

\begin{proof}

We will prove Proposition \ref{prop: relations} by showing the following two arguments one by one.

\begin{enumerate}
    \item $\hat{S}^{\s}\subseteq \hat{S}^{\stau}$;
    \item $\hat{S}^{\bc}\subseteq \hat{S}^{\s}$.
\end{enumerate}

\textbf{1) $\hat{S}^{\s}\subseteq \hat{S}^{\stau}$.} First, we show in the following that $\Ws_i\le \Wst_i$ for all $i$. By definition, $\tilde{Z}_i\ge \Zt_i$, therefore

\begin{align*}
    \Ws_i = Z_i\cdot\sign(Z_i - \tilde{Z}_i)\le  Z_i\cdot\sign(Z_i - \Zt_i) = \Wst_i.
\end{align*}

   Our next goal is to show that for any determined $\D_1$, $\D_2$, there holds $T_q^{\stau}\le T_q^{\s}$. Once such an argument is established, by the property $\Ws_i\le \Wst_i$ for all $i$, there holds
    \begin{align*}
        \hat{S}^{\s} = \{i: \Ws_i\ge T_q^{\s}\}\subseteq \{i: \Wst_i\ge T_q^{\s}\}\subseteq \{i: \Wst_i\ge T_q^{\stau}\} = \hat{S}^{\stau}.
    \end{align*}
    Thus it will be sufficient to show $T_q^{\stau}\le T_q^{\s}$ to prove $\hat{S}^{\s}\subseteq \hat{S}^{\stau}$.

By the property that $\Ws\le \Wst$ for all $i$, for all $\lambda>0$, there holds
\begin{align*}
    \{i: \Wst_i\le -\lambda\}\subseteq \{i: \Ws_i\le -\lambda\},\ \{i: \Ws_i\ge \lambda\}\subseteq \{i: \Wst_i\ge \lambda\},
\end{align*}
which further suggest that
\begin{align*}
    &\frac{|\{i:\Wst_i\le-\lambda\}|}{1\vee|\{i:\Wst_i\ge \lambda\}|}\le \frac{|\{i:\Ws_i\le-\lambda\}|}{1\vee|\{i:\Ws_i\ge \lambda\}|},  \frac{1+|\{i:\Wst_i\le-\lambda\}|}{1\vee|\{i:\Wst_i\ge \lambda\}|}\le \frac{1+|\{i:\Ws_i\le-\lambda\}|}{1\vee|\{i:\Ws_i\ge \lambda\}|}.
\end{align*} 
Thus by the definition of threshold $T_q^{\s}$ and $T_q^{\stau}$, there holds $T_q^{\stau}\le T_q^{\s}$.

\textbf{2) $\hat{S}^{\bc}\subseteq \hat{S}^{\s}$.} The proof follows a similar stream as above. First, we show that $\Wbc_i\le \Ws_i$ for all $i$, as
\begin{itemize}
    \item if $Z_i>\tilde{Z}_i$, $\Wbc_i = Z_i = \Ws_i$;
    \item if $Z_i<\tilde{Z}_i$, $\Wbc_i = -\tilde{Z}_i<-Z_i = \Ws_i$.
\end{itemize}

We will then show in the following that for any determined $\D_1$, $\D_2$, there holds $T_q^{\s}\le T_q^{\bc}$. Once the argument is established, by the property $\Wbc_i\le \Ws_i$ for all $i$, there holds
    \begin{align*}
        \hat{S}^{\bc} = \{i: \Wbc_i\ge T_q^{\bc}\}\subseteq \{i: \Ws_i\ge T_q^{\bc}\}\subseteq \{i: \Ws_i\ge T_q^{\s}\} = \hat{S}^{\s}.
    \end{align*}
    Thus it will be sufficient to show $T_q^{\s}\le T_q^{\bc}$ to prove $\hat{S}^{\bc}\subseteq \hat{S}^{\s}$.

By the property that $\Wbc\le \Ws$ for all $i$, for all $\lambda>0$, there holds
\begin{align*}
    \{i: \Ws_i\le -\lambda\}\subseteq \{i: \Wbc_i\le -\lambda\},\ \{i: \Wbc_i\ge \lambda\}\subseteq \{i: \Ws_i\ge \lambda\},
\end{align*}
which further suggest that
\begin{align*}
    &\frac{|\{i:\Ws_i\le-\lambda\}|}{1\vee|\{i:\Ws_i\ge \lambda\}|}\le \frac{|\{i:\Wbc_i\le-\lambda\}|}{1\vee|\{i:\Wbc_i\ge \lambda\}|},  \frac{1+|\{i:\Ws_i\le-\lambda\}|}{1\vee|\{i:\Ws_i\ge \lambda\}|}\le \frac{1+|\{i:\Wbc_i\le-\lambda\}|}{1\vee|\{i:\Wbc_i\ge \lambda\}|}.
\end{align*} 
Thus by the definition of threshold $T_q^{\bc}$ and $T_q^{\s}$, there holds $T_q^{\s}\le T_q^{\bc}$.
\end{proof}

\subsection{Proof of Proposition \ref{prop: inequality of w statistics}}
\label{sec: proof inequality of w statistics}

\begin{proof}
    We prove the properties in Proposition \ref{prop: inequality of w statistics} one by one.

    \textbf{1) For all $i$, $\{\Wst_i\le-T\}\subseteq\{\Ws_i\le-T\}\subseteq\{\Wbc_i\le-T\}$.}

    We first show that there holds $\{\Wst_i\le -T\}\subseteq\{\Ws_i\le -T\}$ for all $i$. For all $i$, note that $\{\Ws_i\le -T\} = \{Z_i\ge T, \tilde{Z}_i>Z_i\}$, while $\{\Wst_i\le -T\} = \{Z_i\ge T, \tau(\tilde{Z})_i>Z_i\}$. Since by definition $\tau(\tilde{Z})_i = \tilde{Z}_i\cdot 1\{r_i=\tilde{r}_i\} \le \tilde{Z}_i$, there holds $\{\Wst_i\le -T\}\subseteq\{\Ws_i\le -T\}$.

    Then we show that there holds $\{\Ws_i\le -T\}\subseteq\{\Wbc_i\le -T\}$ for all $i$. For all $i$, note that $\{\Wbc_i\le -T\} = \{Z_i\vee \tilde{Z}_i\ge T, \tilde{Z}_i>Z_i\}$, while $\{\Ws_i\le -T\} = \{Z_i\ge T, \tilde{Z}_i>Z_i\}$. Since $Z_i\vee \tilde{Z}_i\ge Z_i$, there holds $\{\Ws_i\le -T\}\subseteq\{\Wbc_i\le -T\}$.
    
    \textbf{2) For all $i$, $\{\Wbc_i\ge T\} = \{\Ws_i\ge T\}\subseteq\{\Wst_i\ge T\}$.}

    We first show that  $\{\Ws_i\ge T\}\subseteq\{\Wst_i\ge T\}$ for all $i$. For all $i$, note that $\{\Ws_i\ge T\} = \{Z_i\ge T, \tilde{Z}_i<Z_i\}$, while $\{\Wst_i\ge T\} = \{Z_i\ge T, \tau(\tilde{Z})_i<Z_i\}$. Since by definition $\tau(\tilde{Z})_i = \tilde{Z}_i\cdot 1\{r_i=\tilde{r}_i\} \le \tilde{Z}_i$, there holds $\{\Ws_i\ge T\}\subseteq\{\Wst_i\ge T\}$.

    Then we show that $\{\Wbc_i\ge T\} = \{\Ws_i\ge T\}$ for all $i$. For all $i$, note that $\{\Wbc_i\ge T\} = \{Z_i\vee \tilde{Z}_i\ge T, \tilde{Z}_i<Z_i\}$, where $Z_i\vee \tilde{Z}_i = Z_i$ for $Z_i>\tilde{Z}_i$. Therefore, $\{\Wbc_i\ge T\} = \{Z_i\ge T, \tilde{Z}_i<Z_i\} = \{\Ws_i\ge T\}$.

    \textbf{3) There holds $\M_{T}(\Wst)\ge \M_{T}(\Ws)\ge \M_{T}(\Wbc)$.}
    
    Since $\{\Wst_i\le-T\}\subseteq\{\Ws_i\le-T\}\subseteq\{\Wbc\le-T\}$ for all $i$, there holds
    \begin{align}
        \sum_{i\in S_0}1\{\Wbc_i\le -T\}\ge\sum_{i\in S_0}1\{\Ws_i\le -T\}\ge \sum_{i\in S_0}1\{\Wst_i\le -T\}.\label{ineq: s->st -}
    \end{align}
    Moreover, since $\{\Wbc_i\ge T\} = \{\Ws_i\ge T\}\subseteq\{\Wst\ge T\}$ for all $i$, there holds
    \begin{align}
        \sum_{i\in S_0}1\{\Wst_i\ge T\}\ge \sum_{i\in S_0}1\{\Ws_i\ge T\} = \sum_{i\in S_0}1\{\Wbc_i\ge T\}.\label{ineq: s->st +}
    \end{align}
    Combining Equation \eqref{ineq: s->st -} with Equation \eqref{ineq: s->st +}, there holds $\M_{T}(\Wst)\ge \M_{T}(\Ws)\ge \M_{T}(\Wbc)$.
\end{proof}

\subsection{Proof of Proposition \ref{prop: well definedness}}

\label{sec: proof well defined}

\begin{proof}
    We first show that for $Z_i>0$, $\{\limsup_{\lambda\to Z_i^-}\sign(\gamma_i(\lambda)) = 1\}$ is the complement of $\{ \liminf_{\lambda\to Z_i^-}\sign(\gamma_i(\lambda)) = -1\}$, which is equivalent with the following two points
    \begin{enumerate}
        \item $\{\limsup_{\lambda\to Z_i^-}\sign(\gamma_i(\lambda)) \neq 1\}\cap\{\liminf_{\lambda\to Z_i^-}\sign(\gamma_i(\lambda)) \neq -1\} = \emptyset$,
        \item $\{\limsup_{\lambda\to Z_i^-}\sign(\gamma_i(\lambda)) = 1\}\cap\{ \liminf_{\lambda\to Z_i^-}\sign(\gamma_i(\lambda)) = -1\} = \emptyset$.
    \end{enumerate}
    We prove the two points in the following.
    
    \textbf{1)} 
    We show by contradiction that for $Z_i>0$, the statements $\limsup_{\lambda\to Z_i^-}\sign(\gamma_i(\lambda)) \neq 1$ and $\liminf_{\lambda\to Z_i^-}\sign(\gamma_i(\lambda)) \neq -1$ cannot both be true.

    Suppose that both statements above are true, then since $\sign(x)\in\{-1, 0, 1\}$ for $x\in\R$, there holds:
    \begin{itemize}
        \item $\limsup_{\lambda\to Z_i^-}\sign(\gamma_i(\lambda)) \neq 1$ suggests that there exists $\delta_1>0$, such that for $Z_i-\delta_1<\lambda<Z_i$, there holds $\gamma_i(\lambda)\le 0$;
        \item $\liminf_{\lambda\to Z_i^-}\sign(\gamma_i(\lambda)) \neq -1$ suggests that there exists $\delta_2>0$, such that for $Z_i-\delta_2<\lambda<Z_i$, there holds $\gamma_i(\lambda)\ge 0$.
    \end{itemize}
    Therefore, for $Z_i-\min\{\delta_1, \delta_2\}<\lambda<Z_i$, there holds $\gamma_i(\lambda)= 0$. Since $\widehat{\beta}(\lambda)$ is continuous, $\gamma(\lambda)$ satisfying Equation \eqref{eq: kkts} is continuous. Therefore, $\gamma_i(Z_i)=\lim_{\lambda\to Z_i^+}\gamma_i(\lambda) = 0$. This further suggests that $\gamma_i(\lambda)= 0$ for $\lambda> Z_i-\min\{\delta_1, \delta_2\}$, which contradicts with the definition of $Z$ in Equation \eqref{def: Z}.

    \textbf{2)} 
    We show by contradiction that for $Z_i>0$, the statements $\limsup_{\lambda\to Z_i^-}\sign(\gamma_i(\lambda)) = 1$ and $\liminf_{\lambda\to Z_i^-}\sign(\gamma_i(\lambda)) = -1$ cannot both be true.

    Suppose that both statements above are true, then since $\rho(\lambda) \in \partial \|\gamma(\lambda)\|_1$, there holds:
    \begin{itemize}
        \item $\limsup_{\lambda\to Z_i^-}\sign(\gamma_i(\lambda)) = 1$ suggests that $\limsup_{\lambda\to Z_i^-}\rho_i(\lambda) = 1$;
        \item $\liminf_{\lambda\to Z_i^-}\sign(\gamma_i(\lambda)) = -1$ suggests that $\liminf_{\lambda\to Z_i^-}\rho_i(\lambda) = -1$.
    \end{itemize}
    However, since $\widehat{\beta}(\lambda)$ is continuous, $\rho(\lambda)$ satisfying Equation \eqref{eq: kkts} is continuous. Therefore, $1 = \limsup_{\lambda\to Z_i^-}\rho_i(\lambda) = \liminf_{\lambda\to Z_i^-}\rho_i(\lambda) = -1$, which leads to contradiction. 

    Then we show that Equation \eqref{def: equivalent def r} is equivalent with the definition of $r$ given in Equation \eqref{def: r and tilde r}. In particular, we show that for $Z_i>0$, there holds
    \begin{enumerate}
        \item $\{\limsup_{\lambda\to Z_i^-}\sign(\gamma_i(\lambda)) = 1\}\subseteq\{\sign(\rho_i(Z_i))=1\}$, 
        \item $\{\liminf_{\lambda\to Z_i^-}\sign(\gamma_i(\lambda)) = -1\}\subseteq\{\sign(\rho_i(Z_i))=-1\}$.
    \end{enumerate}
    The above statements are sufficient because it is already shown above that for $Z_i>0$, $\{\limsup_{\lambda\to Z_i^-}\sign(\gamma_i(\lambda)) = 1\}$ is the complement of $\{ \liminf_{\lambda\to Z_i^-}\sign(\gamma_i(\lambda)) = -1\}$. We prove the two points in the following.

    \textbf{1)} Suppose that $\limsup_{\lambda\to Z_i^-}\sign(\gamma_i(\lambda)) = 1$. Then since $\rho(\lambda) \in \partial \|\gamma(\lambda)\|_1$, there holds $\limsup_{\lambda\to Z_i^-}\rho_i(\lambda) = 1$. By the continuity of $\rho(\lambda)$, there holds $\rho_i(Z_i) = \limsup_{\lambda\to Z_i^-}\rho_i(\lambda) = 1$. Therefore $\sign(\rho_i(Z_i)) = 1$.

    \textbf{2)} Suppose that $\liminf_{\lambda\to Z_i^-}\sign(\gamma_i(\lambda)) = -1$. Then since $\rho(\lambda) \in \partial \|\gamma(\lambda)\|_1$, there holds $\liminf_{\lambda\to Z_i^-}\rho_i(\lambda) = -1$. By the continuity of $\rho(\lambda)$, there holds $\rho_i(Z_i) = \liminf_{\lambda\to Z_i^-}\rho_i(\lambda) = -1$. Therefore $\sign(\rho_i(Z_i)) = -1$.

    Equation \eqref{def: equivalent def tilde r} can be shown to be well-defined and equivalent with the definition of $\tilde r$ given in Equation \eqref{def: r and tilde r} in the exact same way as above. This ends the proof.
\end{proof}

\subsection{Proof of Proposition \ref{prop: s to bc}}

\label{sec: proof s to bc}

\begin{proof}

We first present the following two useful properties on $\Ws$ and $\Wbc$. Recall that by definitions of $\Wbc$ and $\Ws$, there holds for all $i\in \{1, 2,\cdots, m\}$ that
\begin{enumerate}
    \item[(i)] if $Z_i>\tilde{Z}_i$, then $|\Wbc_i|= Z_i = |\Ws_i|$;
    \item[(ii)] if $Z_i<\tilde{Z}_i$, then $|\Wbc_i|= \tilde{Z}_i> Z_i = |\Ws_i|$.
\end{enumerate}

Now we proceed to prove the argument that $\Fs^{\bc}(T)\subseteq \Fs^{\s}(T)$. Below we show that, the filtration associated with $\M_{T}(\Wbc)$ is a refinement of the filtration associated with $\M_{T}(\Ws)$, hence a stopping time adapted to the latter is also adapted to the former.

To see this, recall that by definitions of $\Wbc$ and $\Ws$, 
from (i) we see $\{i:\Wbc_i>0\}=\{i:\Ws_i>0\}$, on which $|\Wbc_i|= Z_i = |\Ws_i|$, hence 
$$\#\{i: \Wbc_i\ge T\}=\#\{i: \Ws_i\ge T\};$$ 
from (i) and (ii) we see for all $i$, $|\Ws_i|\leq |\Wbc_i|$, hence $$\{\zeta_i: |\Wbc_i|< T\} \subseteq \{\zeta_i: |\Ws_i|< T\}.$$ 
Therefore, $\{\zeta_i: |\Wbc_i|< T\}$ can be represented by $\{\zeta_i: |\Ws_i|< T\}$. Now it remains to consider $\#\{i: \Wbc_i\le -T\}$. In fact,
\begin{align*}
    \{i: \Wbc_i\le -T\}&=\{i: \Wbc_i\le -T, |\Ws_i|\ge T\}\cup \{i: \Wbc_i\le -T, |\Ws_i|< T\},\\
    & = \{i: \Wbc_i\le -T, \Ws_i\le -T\}\cup \{i: \Wbc_i\le -T, |\Ws_i|< T\},\\
    & = \{i: \Ws_i\le -T\}\cup \{i: \Wbc_i\le -T, |\Ws_i|< T\},
\end{align*}
which implies that $\#\{i: \Wbc_i\le -T\}$ can be sufficiently determined by $\#\{i: \Ws_i\le -T\}$ and $\{\zeta_i: |\Ws_i|< T\}$, both of which are already included in $\Fs^{\s}(T)$. This finally shows that $\Fs^{\bc}(T)\subseteq \Fs^{\s}(T)$. 
\end{proof}

\subsection{Proof of Theorem \ref{theorem: fdr}}

\label{sec: proof thm}

In this section, we will give the complete proof of Theorem \ref{theorem: fdr}. Our treatment includes the following four types of $W$-statistics:
\begin{enumerate}
    \item $\Ws:=Z\odot \sign(Z -\tilde{Z} ),$ where S refers to ``Split''.
    \item $\Wst:=Z\odot \sign(Z -\tau(\tilde{Z}) ),$ where S$\tau$ refers to applying truncation \eqref{eq:truncation} on $\Ws$.
    \item $\Wbc:= (Z \vee \tilde{Z}) \odot\sign(Z -\tilde{Z} ),$ where BC refers to the original definition adopted by Barber-Cand\`{e}s in \cite{barber2015controlling}.
    \item $\Wbct:=(Z \vee \tau(\tilde{Z}))\odot \sign(Z -\tau(\tilde{Z}) ),$ where BC$\tau$ refers to applying truncation \eqref{eq:truncation} on $\Wbc$.
\end{enumerate}
The last one $\Wbct$ is added here for completeness. For shorthand notation, we use $W\upc$ to represent any one of the four cases, where $\C\in\{\s, \stau, \bc, \bct\}$. In addition, $(\B, \A)$ is used to denote one of the pairs in $\{(\s, \bc), (\stau, \bct)\}$ with truncation $\tau$ adopted or not.

We will first show that by a standard procedure in Knockoffs as in \cite{barber2015controlling}, the problem of bounding the FDR by $q$ can be transferred into the problem of bounding $\E\left[\M_{T_q\upc}(W\upc)\right]$ by one, where $\M_{T}(W\upc)$ for any $T>0$ is defined in Equation \eqref{def: mtw}.
Following that, we divide the proof into two parts:
\begin{enumerate}
    \item In Section \ref{sec: proof thm part 1}, we will prove that $\E\left[\M_{T_q\upb}(W\upb)\right]\le 1$ by introducing an inverse supermartingale structure associated with $\M_{T}(W\upb)$.
    \item In Section \ref{sec: proof thm part 2}, we will show that $T_q\upa$ is a stopping time with respect to a filtration associated with $\M_{T}(W\upb)$, which enables us to show that $\E\left[\M_{T_q\upa}(W\upa)\right]\le\E\left[\M_{T_q\upa}(W\upb)\right]\le 1$.
\end{enumerate}

We begin the proof from the following common procedure of Knockoffs \citep{barber2015controlling}, that the upper bound of FDR in the case of $W\upc$ is transferred into the upper bound of $\E\left[\M_{T_q\upc}(W\upc)\right]$. The procedure goes as the following for Split Knockoff(+).

\begin{itemize}
    \item[(a)] (Split Knockoff) The mFDR can be bounded by the following product:
\begin{align}
    \E\left[\frac{\sum_{i\in S_0}1\{W_i\upc\ge T_q\upc\}}{\sum_{i}1\{W_i\upc\ge T_q\upc\}+q^{-1}}\right] \le \E\left[\frac{1+\sum_{i}1\{W_i\upc\le -T_q\upc\}}{\sum_{i}1\{W_i\upc\ge T_q\upc\}+q^{-1}}\frac{\sum_{i\in S_0}1\{W_i\upc\ge T_q\upc\}}{1+\sum_{i\in S_0}1\{W_i\upc\le -T_q\upc\}}\right]\label{knock3}.
\end{align}
By the definition of the Split Knockoff threshold, there holds
\begin{equation}
    \frac{\sum_{i}1\{W_i\upc\le -T_q\upc\}}{1\vee\sum_{i}1\{W_i\upc\ge T_q\upc\}}\le q\le 1,\nonumber
\end{equation}
which implies
\begin{equation}
    \sum_{i}1\{W_i\upc\le -T_q\upc\}\le q\sum_{i}1\{W_i\upc\ge T_q\upc\}.\nonumber
\end{equation}
Consequently, there holds
\begin{align}
    \frac{1+\sum_{i}1\{W_i\upc\le -T_q\upc\}}{\sum_{i}1\{W_i\upc\ge T_q\upc\}+q^{-1}}\le& \frac{1+q[\sum_{i}1\{W_i\upc\ge T_q\upc\}]}{\sum_{i}1\{W_i\upc\ge T_q\upc\}+q^{-1}}=q.\nonumber
\end{align}
Combined with Equation \eqref{knock3}, there holds
\begin{align}
    \E\left[\frac{\sum_{i\in S_0}1\{W_i\upc\ge T_q\upc\}}{\sum_{i}1\{W_i\upc\ge T_q\upc\}+q^{-1}}\right] &\le q\E\left[\frac{\sum_{i\in S_0}1\{W_i\upc\ge T_q\upc\}}{1+\sum_{i\in S_0}1\{W_i\upc\le -T_q\upc\}}\right] = q\E\left[\mathcal{M}_{T_q\upc}(W\upc)\right].\label{eq: knockoff bound}
\end{align}
\item[(b)] (Split Knockoff+) The following lines established the result,
\begin{align}
    \E\left[\frac{\sum_{i\in S_0}1\{W_i\upc\ge T_q\upc\}}{1\vee\sum_{i}1\{W_i\upc\ge T_q\upc\}}\right] & \le \E\left[\frac{1+\sum_{i}1\{W_i\upc\le -T_q\upc\}}{1\vee\sum_{i}1\{W_i\upc\ge T_q\upc\}}\frac{\sum_{i\in S_0}1\{W_i\upc\ge T_q\upc\}}{1+\sum_{i\in S_0}1\{W_i\upc\le -T_q\upc\}}\right],\nonumber\\
    &\le q\E\left[\frac{\sum_{i\in S_0}1\{W_i\upc\ge T_q\upc\}}{1+\sum_{i\in S_0}1\{W_i\upc\le -T_q\upc\}}\right] = q\E\left[\mathcal{M}_{T_q\upc}(W\upc)\right].\label{eq: knockoff+ bound}
\end{align}
\end{itemize}

Then we transfer the problem of bounding (m)FDR by $q$ in Theorem \ref{theorem: fdr} into the problem of bounding $\E\left[\M_{T_q\upc}(W\upc)\right]$ by one using Equation \eqref{eq: knockoff bound} and Equation \eqref{eq: knockoff+ bound}. Then we will prove the following two inequalities respectively in Section \ref{sec: proof thm part 1} and Section \ref{sec: proof thm part 2}.
\begin{enumerate}
    \item We will first prove that $\E\left[\M_{T_q\upb}(W\upb)\right]\le 1$ in Section \ref{sec: proof thm part 1}.
    \item We will then show that $\E\left[\M_{T_q\upa}(W\upa)\right]\le\E\left[\M_{T_q\upa}(W\upb)\right]\le 1$ in Section \ref{sec: proof thm part 2}.
\end{enumerate}

\subsubsection{Proof of Theorem \ref{theorem: fdr}: Case I}

\label{sec: proof thm part 1}

As a reminder for the notations, $\B$ represents an arbitrary element from $(\s, \stau)$. 
In this section,
We will target to formulate a supermartingale structure associated with $\M_{T}(W\upb)$. In order to do such a thing, we will need to show that $W\upb$ is a statistics whose sign $\{\sign(W\upb)\}$ and length ($|W\upb|$) are independent from each other.

To show such a independence property, we will need to take a deeper look at the KKT conditions in Equation \eqref{eq: kkts}. 
From the definition, there holds
\begin{enumerate}
    \item $\widehat{\beta}(\lambda)$ from $\D_1=(X_1, y_1)$ and $\zeta$ from $\D_2 = (X_2, y_2)$ are independent from each other;
    \item $\gamma(\lambda)$ is determined by $\widehat{\beta}(\lambda)$ from $\D_1=(X_1, y_1)$, which is the same for $Z$ and $r:=\sign(\gamma(Z-))$ as functions of $\gamma(\lambda)$;
    \item conditional on $\widehat{\beta}(\lambda)$, $\tilde{\gamma}$ is determined by $\zeta$ from $\D_2=(X_2, y_2)$, which is the same for $\tilde{Z}$, $\tilde{r}:=\sign(\tilde{\gamma}(\tilde{Z-}))$, and $\tau(\tilde{Z})$ as functions of $\tilde{\gamma}$.
\end{enumerate}
Therefore, conditional on $\widehat{\beta}(\lambda)$ from $\D_1=(X_1, y_1)$ which determines $|W\upb|=Z$, the difference between the feature significance \eqref{eq: feature kkt} and knockoff significance \eqref{eq: knockoff kkt} lies on the random variable $\zeta$ from $\D_2=(X_2, y_2)$.
Thus the length $|W\upb|$ and sign $\{\sign(W\upb)\}$ of $W\upb$ are independent from each other. Additional calculations on $\zeta$ shows that $\zeta$ are consist of independent Gaussian random variables, 
i.e. 
for the Split Knockoff matrix satisfying \eqref{eq: copy}, the distribution 
of $\zeta+\diag(\vecs)\gamma^*$ satisfies Equation \eqref{eq: zeta dis}. 
Then the following lemma can be given for $W\upb\in\{W^{\s}, W^{\stau}\}$, in addition to $\Ws$ only in Lemma \ref{lemma: independent Ws front}.

\begin{lemma}
    \label{lemma: independent Ws}
    Given any determined $\widehat{\beta}(\lambda)$, $1\{W\upb_i<0\}$ are some independent Bernoulli random variables. Furthermore, for $i\in S_0\cap\{i: |W\upb_i|=Z_i>0\}$, there holds
    \begin{align*}
        \Prob[W\upb_i<0]\ge \frac{1}{2}.
    \end{align*}
\end{lemma}

For shorthand notations, we rearrange the index of $W\upb$, such that $|W\upb_{(1)}|\ge|W\upb_{(2)}|\ge\cdots\ge|W\upb_{(m^*)}|>0$, and $\{(1), (2), \cdots, (m^*)\}= S_0\cap\{i:|\Ws_i| = Z_i> 0\}$. Further denote $B_{(i)} = 1\{W\upb_{(i)}<0\}$,
then there holds
\begin{align}
    \frac{\sum_{i\in S_0}1\{W\upb_i\ge T_q\upb\}}{1+\sum_{i\in S_0}1\{W\upb_i\le -T_q\upb\}} & =  \frac{1+\sum_{i\in S_0}1\{|W\upb_i|\ge T_q\upb\}}{1+\sum_{i\in S_0}1\{|W\upb_i|\ge T_q\upb, W\upb_i<0\}}-1,\nonumber\\
    & = \frac{1+J}{1+B_{(1)}+B_{(2)}+\cdots+B_{(J)}}-1,\label{eq: transfera}
\end{align}
where $J\le m^*$ is defined to be the index satisfying 
\begin{align*}
    |W\upb_{(1)}|\ge|W\upb_{(2)}|\ge\cdots\ge|W\upb_{(J)}|\ge T_q\upb>|W\upb_{(J+1)}|\ge\cdots\ge|W\upb_{(m^*)}|.
\end{align*}
In other words, $J = \argmax_{k\le m^*}\{|W\upb_{(k)}|\ge T_q\upb\}$.

Define the filtration $\{\mathcal{F}_j\}_{j=1}^{m}$ and $\{\mathcal{G}_j\}_{j=1}^{m}$ in inverse time as
\begin{align*}
    \mathcal{F}_j &= \sigma\left(\left\{\sum_{i=1}^jB_{(i)}, \zeta_{(j+1)}, \cdots, \zeta_{(m)}\right\}\right),\\
    \mathcal{G}_j &= \sigma\left(\left\{\sum_{i=1}^jB_{(i)}, B_{(j+1)}, \cdots, B_{(m)}\right\}\right).
\end{align*}
Conditional on $\widehat{\beta}(\lambda)$, since $\zeta_i$ determines $\sign(W\upb_i)$ and $B_i$, the filtration $\{\mathcal{F}_j\}_{j=1}^{m}$ is a refined filtration of $\{\mathcal{G}_j\}_{j=1}^{m}$ in inverse time, i.e. $\mathcal{G}_j\subseteq \mathcal{F}_j$. 
By \cite{barber2015controlling, barber2019knockoff}, $J$ is a stopping time on $\{\mathcal{G}_j\}_{j=1}^{m}$ in inverse time, thus $J$ is also a stopping time on the refined filtration $\{\mathcal{F}_j\}_{j=1}^{m}$ in inverse time.
Then it will be proper to apply Lemma \ref{lemma: key lemma front} as a supermartingale inequality to give an upper bound on the expectation of Equation \eqref{eq: transfera}.

Applying Lemma \ref{lemma: key lemma front} to Equation \eqref{eq: transfera}, with the estimation that $\Prob[B_i=1]\ge \rho = \frac{1}{2}$ for $i\in S_0$ by Lemma \ref{lemma: independent Ws}, we will have the following inequality on $\M_{T_q\upb}(W\upb)$ that
\begin{align} 
    \Expect\left[\M_{T_q\upb}(W\upb)\right]=\Expect\left[\frac{\sum_{i\in S_0}1\{W\upb_i\ge T_q\upb\}}{1+\sum_{i\in S_0}1\{W\upb_i\le -T_q\upb\}}\right] \le \rho^{-1}-1=1.
\end{align}
Combining such results with Equation \eqref{eq: knockoff bound} and Equation \eqref{eq: knockoff+ bound}, and we will finish the proof.

\subsubsection{Proof of Theorem \ref{theorem: fdr}: Case II}

\label{sec: proof thm part 2}

As a reminder of notations, $(\A, \B)$ denotes an arbitrary element from $\{(\bc, \s), (\bct, \stau)\}$. 
In the case of $\A$, due to the failure of exchangeability, $\M_{T}(W\upa)$ is now no longer a supermartingale with stopping time $T_q\upa$. To address this challenge, we are going show that $\M_{T}(W\upb)$ gives an upper bound of $\M_{T}(W\upa)$ and is associated with a supermartingale structure at the same time. Specifically, we will show the following properties in addition to Proposition \ref{prop: s to bc}:
\begin{itemize}
    \item[(a)] $\M_{T}(W\upb)$ provides an upper bound for $\M_{T}(W\upa)$, that $\M_{T}(W\upa)\leq \M_{T}(W\upb)$;
    \item[(b)] $T_q\upa$ induces a stopping time for the inverse martingale associated with $\M_{T}(W\upb)$ which enables the application of upper bounds in Case I to Case II.
\end{itemize}

\subsubsection*{a) Upper Bound Property}

Comparing the definition of $W\upa$ and the definition of $W\upb$, we have
\begin{enumerate}
    \item for $W\upb_i>0$, $W\upa_i= Z_i = W\upb_i$;
    \item for $W\upb_i<0$, if $(\A, \B) = (\bc, \s)$, $W\upa_i= - \tilde{Z}_i\le -Z_i = W\upb_i$;
    \item for $W\upb_i<0$, if $(\A, \B) = (\bct, \stau)$, $W\upa_i= - \tau(\tilde{Z})_i\le -Z_i = W\upb_i$;
\end{enumerate}

Therefore, $\{i\in S_0: W\upa_i\ge T\} = \{i\in S_0: W\upb_i\ge T\}$, while $\{i\in S_0: W\upb_i\le -T\}\subseteq \{i\in S_0: W\upa_i\le- T\}$ for $T>0$, which further indicates
\begin{align}
    \M_{T}(W\upa)=\frac{\sum_{i\in S_0}1\{W\upa_i\ge T\}}{1+\sum_{i\in S_0}1\{W\upa_i\le -T\}}
    \le & \frac{\sum_{i\in S_0}1\{W\upb_i\ge T\}}{1+\sum_{i\in S_0}1\{W\upb_i\le -T\}}=:\M_{T}(W\upb).\label{transfer_W_S}
\end{align}
Thus $\M_{T}(W\upb)$ offers an upper bound for $\M_{T}(W\upa)$.

\subsubsection*{b) Stopping Time and Supermartingale Inequalities}

To apply the supermartingale inequality in Lemma \ref{lemma: key lemma front}, we will need to check that 
$T_q\upa$ defined by $W\upa$ can be induced to a stopping time associated with $\M_{T}(W\upb)$ with respect to a proper filtration. 

For shorthand notations, we rearrange the index on of $W\upb$, such that $|W\upb_{(1)}|\ge|W\upb_{(2)}|\ge\cdots\ge|W\upb_{(m^*)}|>0$, and $\{(1), (2), \cdots, (m^*)\}= S_0\cap\{i:|\Ws_i| = Z_i> 0\}$. Further denote $B_{(i)} = 1\{W\upb_{(i)}<0\}$, then there holds
\begin{align}
    \frac{\sum_{i\in S_0}1\{W\upb_i\ge T_q\upa\}}{1+\sum_{i\in S_0}1\{W\upb_i\le -T_q\upa\}} & =  \frac{1+\sum_{i\in S_0}1\{|W\upb_i|\ge T_q\upa\}}{1+\sum_{i\in S_0}1\{|W\upb_i|\ge T_q\upa, W\upb_i<0\}}-1,\nonumber\\
    & = \frac{1+J}{1+B_{(1)}+B_{(2)}+\cdots+B_{(J)}}-1,\label{eq: transferred bound bctos}
\end{align}
where $J\le m^*$ is defined to be the index satisfying 
\begin{align*}
    |W\upb_{(1)}|\ge|W\upb_{(2)}|\ge\cdots\ge|W\upb_{(J)}|\ge T_q\upa>|W\upb_{(J+1)}|\ge\cdots\ge|W\upb_{(m^*)}|,
\end{align*}
in other words, $J = \argmax_{k\le m^*}\{|W\upb_{(k)}|\ge T_q\upa\}$. It can be shown  $\Fs\upb(T)$ is a refined filtration of $\Fs\upa(T)$, i.e. $\Fs\upa(T)\subseteq \Fs\upa(T)$ in a similar way as discussed in Proposition \ref{prop: s to bc}. Rigorously speaking, we will have the following lemma showing that $J$ is also a stopping time in inverse time with respect to the filtration $\mathcal{F}$ associated with $W\upb$.
\begin{lemma}
    \label{lemma: stopping time}
    For any determined $\widehat{\beta}(\lambda)$, $J = \max_{i\le m^*}\{|W\upb_{(i)}|\ge T_q\upa\}$ is a stopping time with respect to the filtration
    $\{\mathcal{F}_j\}_{j=1}^m$ in inverse time defined as
    \begin{align*}
        \mathcal{F}_j = \sigma\left(\left\{\sum_{i=1}^jB_{(i)}, \zeta_{(j+1)}, \cdots, \zeta_{(m)}\right\}\right).
    \end{align*}
\end{lemma}

With this stopping time property, we can apply Lemma \ref{lemma: key lemma front} again to get the desired FDR bound. Specifically we will have the following inequality on $\M_{T_q\upa}(W\upa)$ and $\M_{T_q\upa}(W\upb)$:
\begin{align} \label{eq:alpha}
    \Expect\left[\M_{T_q\upa}(W\upa)\right]&\le\Expect\left[\M_{T_q\upa}(W\upb)\right]=\Expect\left[\frac{\sum_{i\in S_0}1\{W\upb_i\ge T_q\upa\}}{1+\sum_{i\in S_0}1\{W\upb_i\le -T_q\upa\}}\right]\le 1,
\end{align}
which ends the proof.

\subsection{Proof of Supporting Lemmas}

\label{sec: proof lemma}

In this section, we will give the proof of Lemma \ref{lemma: key lemma front}, Lemma \ref{lemma: independent Ws} (an extension of Lemma \ref{lemma: independent Ws front}), and 
Lemma \ref{lemma: stopping time} respectively. As a reminder, $(\A, \B)$ refers to any element from $\{(\bc, \s), (\bct, \stau)\}$, where $\Wbct$ is defined in Section \ref{sec: proof thm} for completeness.

In the proof of Lemma \ref{lemma: key lemma front}, we explicitly construct a decomposition in $B_i$ similar with that in proof of Lemma 1 in \cite{barber2019knockoff}, based on the property that $B_i = 1\{\zeta_i\in G_i\}$ for some Borel set $G_i$. Such a decomposition enables us the desired result.

In the proof of Lemma \ref{lemma: independent Ws}, we give detailed analysis on the KKT conditions \eqref{eq: kkts}, which provides a specific mapping from the value of $\zeta$ to $\sign(W\upb)$. Then $\Prob[W\upb<0]$ can be estimated based on the probability measure on $\zeta$.

In the proof of Lemma \ref{lemma: stopping time}, we apply the ideas of $\Fs^{\bc}(T)\subseteq \Fs^{\s}(T)$ introduced in Proposition \ref{prop: s to bc} in Section \ref{sec:mainresults} in a discrete form. We will show that such a result holds for any choice of $(\A, \B)$ in the proof.

\subsubsection{Proof of Lemma \ref{lemma: key lemma front} }

\label{sec: proof key lemma}

\begin{proof}
We start from the following constructions. For each $i$, we divide the space $\mathbb{R}$ into 4 disjoint Borel sets, $\mathbb{R} = A_1^i\cup A_2^i\cup A_3^i\cup A_4^i$, with $G_i := A_2^i\cup A_3^i\cup A_4^i$, and
\begin{enumerate}
    \item $\Prob[\zeta_i\in A_1^i] = 1-\rho_i$;
    \item $\Prob[\zeta_i\in A_2^i] = \rho\frac{1-\rho_i}{1-\rho}$;
    \item $\Prob[\zeta_i\in A_3^i] = \rho\frac{\rho_i-\rho}{1-\rho}$;
    \item $\Prob[\zeta_i\in A_4^i] = \rho_i-\rho$.
\end{enumerate}
The existence of such division is ensured by $1-\rho_i+\rho\frac{1-\rho_i}{1-\rho}+\rho\frac{\rho_i-\rho}{1-\rho}+\rho_i-\rho=1$, and $\rho\frac{1-\rho_i}{1-\rho}+\rho\frac{\rho_i-\rho}{1-\rho}+\rho_i-\rho=\rho_i$.

Define $U_i=A_1^i\cup A_2^i$, and $V_i = A_2^i\cup A_3^i$ for each $i$. Further define $Q_i = 1\{\zeta_i\in V_i\}$ for each $i$ and a random set $A:=\{i:\zeta_i\in U_i\}$. There holds from the definition that
\begin{align*}
    Q_i\cdot 1\{i\in A\}+1\{i\notin A\} = & 1\{\{\zeta_i\in V_i\cap U_i\}\cup\{\zeta_i\in U_i^C\}\},\\
    = & 1\{\{\zeta_i\in A_2^i\}\cup\{\zeta_i\in A_3^i\cup A_4^i\}\},\\
    = & 1\{\zeta_i\in G_i\} = B_i.
\end{align*}

Therefore
\begin{align}
    \frac{1+J}{1+B_{(1)}+B_{(2)}+\cdots+B_{(J)}} & = \frac{1+|\{i\le J: (i)\in A\}|+|\{i\le J: (i)\notin A\}|}{1+\sum_{i\le J, (i)\in A}Q_{(i)}+|\{i\le J:(i)\notin A\}|},\nonumber\\
    & \le \frac{1+|\{i\le J: (i)\in A\}|}{1+\sum_{i\le J, (i)\in A}Q_{(i)}},
\end{align}
where the last step is by the inequality $\frac{a+c}{b+c}\le \frac{a}{b}$ whenever $0<b\le a$ and $c\ge 0$.

Let $\tilde{Q}_i = Q_i\cdot 1\{i\in A\}$, and define
\begin{align*}
    \F_j' = \sigma\left(\left\{\sum_{i=1}^j \tilde{Q}_{(i)}, \zeta_{(j+1)}, \cdots, \zeta_{(m)}, A\right\}\right),
\end{align*}
then clearly $\{\F_j'\}_{j=1}^m$ is a filtration in inverse time, satisfying $\F_j\subseteq \F_j'$. Therefore, $J$ is also a stopping time on $\F_j'$. Moreover, note that by definition
\begin{align*}
    \Prob[Q_i=1|i\in A] = \Prob[\zeta_i\in V_i|\zeta_i\in U_i] = \frac{\Prob[A_2^i]}{\Prob[A_1^i\cup A_2^i]}=\rho=\Prob[Q_i=1],\\
    \Prob[Q_i=1|i\notin A] = \Prob[\zeta_i\in V_i|\zeta_i\notin U_i] = \frac{\Prob[A_3^i]}{\Prob[A_3^i\cup A_4^i]}=\rho=\Prob[Q_i=1].
\end{align*}
Such an observation combining with the the independence on $\zeta_i$ means that, conditional on any fixed $A$, $Q_i$ are i.i.d. Bernoulli random variables with parameter $\rho$. This further suggests that $Q_i$ as i.i.d. Bernoulli random variables with parameter $\rho$, and are independent from $A$. This means that $Q_i$ are exchangeable \citep{barber2019knockoff} conditional on any fixed $A$, then by Lemma 2 in the supplementary material of \cite{barber2019knockoff}, there holds
\begin{align*}
    \Expect\left[\left.\frac{1+|\{i\le J: (i)\in A\}|}{1+\sum_{i\le J, (i)\in A}Q_{(i)}}\right|A\right]\le \rho^{-1}.
\end{align*}
Taking expectation over $A$, and we will get our desired result.
\end{proof}

\subsubsection{Proof of Lemma \ref{lemma: independent Ws}}

\label{sec: proof independence}

Note that by Equation \eqref{eq: copy}, there holds
\begin{align}
    &\tilde{A}_\gamma^T A_\beta = \tilde{A}^T_{\gamma,1}\frac{X_2}{\sqrt{n_2}}+\tilde{A}^T_{\gamma,2}\frac{D}{\sqrt{\nu}}= A_\gamma^TA_\beta = -\frac{D}{\nu},\nonumber\\
    &\tilde{A}_\gamma^TA_\gamma = -\frac{\tilde{A}^T_{\gamma,2}}{\sqrt{\nu}}= A_\gamma^TA_\gamma - \diag(\vecs)= \frac{I_m}{\nu}-\diag(\vecs),\nonumber
\end{align}
where as a reminder, $\tilde{A}_{\gamma,1}\in\R^{n_2\times m}$ and $\tilde{A}_{\gamma,2}\in\R^{m\times m}$ is defined in the end of Section \ref{sec: construction}. Therefore, there holds
\begin{align}
    \tilde{A}_\gamma^T\tilde{y} & = \tilde{A}^T_{\gamma,1}\frac{X_2\beta^*+\varepsilon_2}{\sqrt{n_2}}=-\diag(\vecs)\gamma^*+\frac{\tilde{A}^T_{\gamma,1}}{\sqrt{n_2}} \varepsilon_2, \label{eq: calculate agammay}
\end{align}
where $\varepsilon_2$ is the last $n_2$ rows of $\varepsilon$ as defined in Section \ref{sec:s analysis}.

From Equation \eqref{eq: feature kkt}, $\gamma(\lambda)$ is determined by $\widehat{\beta}(\lambda)$, suggesting that $\widehat{\beta}(\lambda)$ is the sufficient statistics for $Z$ and $r$. From Equation \eqref{eq: knockoff kkt}, $\widehat{\beta}(\lambda)$ and $\zeta_i$ are the sufficient statistics for $\tilde{\gamma}_i(\lambda)$ and consequently $\tilde{Z}_i$ and $\tilde{r}_i$ for all $i$. Therefore, $\widehat{\beta}(\lambda)$ and $\zeta_i$ are the sufficient statistics for $W\upb_i$ for all $i$. 

Now, we further present by the following lemma that $B_i:=1\{W\upb_i<0\}$ are some independent random variables through the independence among $\zeta_i$.

\begin{lemma}
    For any determined $\widehat{\beta}(\lambda)$, let $B_i:=1\{W\upb_i<0\}$, then $B_i$ are some independent random variables.
\end{lemma}

\begin{proof}
Note that by Equation \eqref{eq: copy}, there holds
\begin{align}
    &\tilde{A}_\gamma^T\tilde{A}_\gamma = \tilde{A}^T_{\gamma,1}\tilde{A}_{\gamma,1}+\tilde{A}^T_{\gamma,2}\tilde{A}_{\gamma,2} =  A_\gamma^TA_\gamma = \frac{I_m}{\nu},\nonumber\\
    &\tilde{A}_\gamma^TA_\gamma = -\frac{\tilde{A}^T_{\gamma,2}}{\sqrt{\nu}}= A_\gamma^TA_\gamma - \diag(\vecs)= \frac{I_m}{\nu}-\diag(\vecs),\nonumber
\end{align}
where $\tilde{A}_{\gamma,2}\in\R^{m\times m}$ is defined to be the matrix that takes last $m$ rows of $\tilde{A}_{\gamma}$. Therefore, there holds
\begin{align}
    \tilde{A}^T_{\gamma,1}\tilde{A}_{\gamma,1}=\diag(\vecs)(2I_m-\diag(\vecs)\nu), \label{eq: ortho of agamma1}
\end{align}
and 
$\zeta+\diag(\vecs)\gamma^*$ follows the distribution presented in Equation \eqref{eq: zeta dis}. 
Therefore, $\zeta$ consists of some \emph{independent} random variables. Combining with the fact that $\widehat{\beta}(\lambda)$ and $\zeta_i$ are the sufficient statistics for $W\upb_i$ for all $i$, we will get our desired result.
\end{proof}

We show by the following proposition that $Z_i < +\infty$ for all $i$, and the event $\{\exists i: Z_i = \tilde{Z}_i>0\}$ is a zero probability event, conditional on any determined $\widehat{\beta}(\lambda)$ satisfying $\lim_{\lambda\to\infty} \frac{\widehat{\beta}(\lambda)}{\lambda} = 0$. To avoid redundant notations, we assume that  $Z_i\neq \tilde{Z}_i$ for all $i$ satisfying $Z_i>0$ throughout this paper.

\begin{proposition}
\label{prop: zero probability event}
    For any determined $\widehat{\beta}(\lambda)$, there holds for any $i$ that
    \begin{itemize}
        \item $Z_i < +\infty$;
        \item $\Prob[Z_i = \tilde{Z}_i>0]= 0$.
    \end{itemize}
\end{proposition}

\begin{proof}
    From Equation \eqref{eq: kkts}, there holds for all $\lambda>0$ and $i\in\{1, 2, \cdots, m\}$ that
    \begin{align*}
        \rho_i(\lambda) + \frac{\gamma_i(\lambda)}{\lambda\nu}&= \frac{[D\widehat{\beta}(\lambda)]_i}{\lambda\nu},\\
        \tilde{\rho}_i(\lambda) + \frac{\tilde{\gamma}_i(\lambda)}{\lambda\nu} &= \frac{[D\widehat{\beta}(\lambda)]_i}{\lambda\nu} +\frac{\zeta_i}{\lambda}.
    \end{align*}
    Therefore, there holds for $Z_i>0$ and $\tilde{Z}_i>0$ that
    \begin{align}
        Z_i = \sup\left\{\lambda: \left|\frac{[D\widehat{\beta}(\lambda)]_i}{\lambda\nu}\right|>1\right\}, \mbox{ and }
        \tilde{Z}_i = \sup\left\{\lambda: \left|\frac{[D\widehat{\beta}(\lambda)]_i}{\lambda\nu}+\frac{\zeta_i}{\lambda}\right|>1\right\}.\label{eq: Z by beta(lambda)}
    \end{align}
    By the condition $\lim_{\lambda\to\infty} \frac{\widehat{\beta}(\lambda)}{\lambda} = 0$ in Section \ref{sec: intercept est}, there holds
    \begin{align*}
        \lim_{\lambda\to\infty} \left|\frac{[D\widehat{\beta}(\lambda)]_i}{\lambda\nu}\right| = 0. 
    \end{align*}
    Therefore, by Equation \eqref{eq: Z by beta(lambda)}, there holds $Z_i< +\infty$.
    
    Moreover, by the continuity of $\widehat{\beta}(\lambda)$, by Equation \eqref{eq: Z by beta(lambda)}, there holds for $Z_i>0$ and $\tilde{Z_i}>0$ that
    \begin{align*}
        \left|\frac{[D\widehat{\beta}(Z_i)]_i}{Z_i\nu}\right| = \lim_{\lambda\to Z_i^+}\left|\frac{[D\widehat{\beta}(\lambda)]_i}{\lambda\nu}\right| =1, \mbox{ and }
        \left|\frac{[D\widehat{\beta}(\tilde{Z}_i)]_i}{\tilde{Z}_i\nu}+\frac{\zeta_i}{\tilde{Z}_i}\right| = \lim_{\lambda\to \tilde{Z}_i^+}\left|\frac{[D\widehat{\beta}(\lambda)]_i}{\lambda\nu}+\frac{\zeta_i}{\lambda}\right| =1.
    \end{align*}
    Therefore, if there further holds $Z_i = \tilde{Z}_i>0$, there holds
    \begin{align*}
        \left|\frac{[D\widehat{\beta}(Z_i)]_i}{Z_i\nu}\right| = \left|\frac{[D\widehat{\beta}(Z_i)]_i}{Z_i\nu}+\frac{\zeta_i}{Z_i}\right|,
    \end{align*}
    which only have finite solutions of $\zeta_i$ for any $\widehat{\beta}(\lambda)$ (which determines $Z_i$). Therefore $\Prob[Z_i = \tilde{Z}_i>0] = 0$.
\end{proof}

We describe by the following lemma on how $\zeta$ determines $\sign(W\upb)$ conditional on a determined $\widehat{\beta}(\lambda)$. 

\begin{lemma}
    \label{lemma: new interval}
    For any determined $\widehat{\beta}(\lambda)$, there holds for $i$ satisfying $Z_i>0$,
    \begin{align*}
        \{r_i\zeta_i>0\}\subseteq\{\Ws_i<0\},\ \{r_i\zeta_i>0\}\subseteq \{\Wst_i<0\},
    \end{align*}
    where $r_i$ is defined in Equation \eqref{def: equivalent def r} and is shown to be equal to $\sign(\rho_i(Z_i))$ for $Z_i>0$ in Proposition \ref{prop: well definedness}.
\end{lemma}

\begin{proof}
    We prove the two properties one by one.

    1) $\{r_i\zeta_i>0\}\subseteq\{\Ws_i<0\}$. Suppose that $r_i\zeta_i>0$. We are going to show below that $\tilde{Z}_i>Z_i$ and $\Ws_i<0$.
    
    By the continuity of $\widehat{\beta}(\lambda)$, $\gamma(\lambda)$ and $\rho(\lambda)$ solved from Equation \eqref{eq: feature kkt} are continuous. Therefore, by the definition of $Z_i$ in Equation \eqref{def: Z}, for $i$ satisfying $Z_i>0$, there holds $\gamma_i(Z_i) = \lim_{\lambda\to Z_i^+}\gamma_i(\lambda) = 0$ and $|\rho_i(Z_i)|=\limsup_{\lambda\to Z_i^-}|\rho_i(\lambda)|=1$. From Equation \eqref{eq: feature kkt}, there holds
    \begin{align*}
        Z_i\rho_i(Z_i) = Z_i\rho_i(Z_i)+\frac{\gamma_i(Z_i)}{\nu} = \frac{[D\widehat{\beta}(Z_i)]_i}{\nu}.
    \end{align*}
    Then from Equation \eqref{eq: knockoff kkt}, there holds
    \begin{align*}
        Z_i\tilde{\rho}_i(Z_i)+\frac{\tilde{\gamma}_i(Z_i)}{\nu} = \frac{[D\widehat{\beta}(Z_i)]_i}{\nu}+\zeta_i = Z_i\rho_i(Z_i) +\zeta_i.
    \end{align*}
    Multiple $r_i = \sign(\rho_i(Z_i))$ on both sides and combining with the fact that $|\rho_i(Z_i)|=1$, there holds
    \begin{align*}
        r_iZ_i\tilde{\rho}_i(Z_i)+\frac{r_i\tilde{\gamma}_i(Z_i)}{\nu} = Z_i + r_i\zeta_i.
    \end{align*}
    Therefore $r_i\tilde{\gamma}_i(Z_i) = \nu r_i\zeta_i>0$. By the continuity of $\widehat{\beta}(\lambda)$, $\tilde\gamma(\lambda)$ and $\tilde\rho(\lambda)$ solved from Equation \eqref{eq: knockoff kkt} are continuous. Therefore, there exists $\delta>0$, such that $r_i\tilde{\gamma}_i(Z_i+\delta) > \frac{\nu r_i\zeta_i}{2}>0$. This further suggests that $\tilde{Z}_i>Z_i+\delta>Z_i$ and $\Ws_i<0$.

    2) $\{r_i\zeta_i>0\}\subseteq\{\Wst_i<0\}$. Suppose that $r_i\zeta_i>0$. It is already shown above that $\tilde{Z}_i>Z_i>0$, and it remains to show that $r_i = \tilde{r}_i = \sign(\tilde\rho_i(\tilde{Z}_i))$ for showing $\Wst_i<0$.

    By Equation \eqref{eq: feature kkt} and the definition of $Z_i$ in Equation \eqref{def: Z}, there holds for $\lambda> Z_i$ that
    \begin{align*}
        \lambda\rho_i(\lambda) = \lambda\rho_i(\lambda)+\frac{\gamma_i(\lambda)}{\nu} = \frac{[D\widehat{\beta}(\lambda)]_i}{\nu}.
    \end{align*}
    Then from $r_i\times$Equation \eqref{eq: knockoff kkt}, there holds for $\lambda> Z_i$ and $r_i\zeta_i>0$ that
    \begin{align*}
        r_i\lambda\tilde{\rho}_i(\lambda)+r_i\frac{\tilde{\gamma}_i(\lambda)}{\nu} = \frac{r_i[D\widehat{\beta}(\lambda)]_i}{\nu}+r_i\zeta_i = r_i\lambda\rho_i(\lambda) +r_i\zeta_i>r_i\lambda\rho_i(\lambda)\ge -\lambda.
    \end{align*}
    Therefore 
    there holds $r_i\tilde{\rho}_i(\lambda)>-1$ for $\lambda> Z_i$ and consequetly $r_i\tilde{\rho}_i(\tilde{Z}_i)>-1$ as $\tilde{Z}_i>Z_i$. 

    However, since $\tilde{Z}_i>Z_i>0$, there holds $|\tilde{\rho}_i(\tilde{Z}_i)| = \limsup_{\lambda\to \tilde Z_i^-}|\tilde\rho_i(\lambda)|=1$. Combining with $r_i\tilde{\rho}_i(\tilde{Z}_i)>-1$, there holds $r_i\tilde{\rho}_i(\tilde{Z}_i) = 1$, 
    and consequently $\tilde{r}_i = r_i$ for $\tilde{r}_i= \sign(\tilde\rho_i(\tilde{Z}_i))$. Therefore, $\Wst_i<0$.
\end{proof}

Then it remains to show that for $i\in S_0\cap\{i: |W\upb_i|=Z_i>0\}$, there holds $\Prob[B_i = 1]\ge \rho(\nu)\ge \frac{1}{2}$ to prove  Lemma \ref{lemma: independent Ws}. 

\begin{proof}
    First, from Lemma \ref{lemma: new interval}, conditional on any $\widehat{\beta}(\lambda)$, there holds $\{r_i \zeta_i>0\}\subseteq\{W\upb_i<0\}$ for $i\in S_0\cap\{i: |W\upb_i|=Z_i>0\}$. While for $i\in S_0$, $\zeta_i$ is symmetrically distributed around 0 by Equation \eqref{eq: zeta dis}. 
    Thus $\Prob[W\upb_i<0]\ge \Prob[r_i \zeta_i>0]=\frac{1}{2}$ conditional on $\widehat{\beta}(\lambda)$, for $i\in S_0\cap\{i: |W\upb_i|=Z_i>0\}$.
    This ends the proof.
\end{proof}

\subsubsection{Proof of Lemma \ref{lemma: stopping time}}

\label{sec: proof stopping time}

\begin{proof}
As a reminder, $(\A, \B)$ represents an arbitrary element from $\{(\bc, \s), (\bct, \stau)\}$, where $\Wbct$ is defined in Section \ref{sec: proof thm} for completeness.
    By definition of stopping time in reverse time, we need to show that $\{J < k\}\in \Fs_k$ for $k\le m^*$.
    We first validate that $\{\mathcal{F}_j\}_{j=1}^m$ is indeed a filtration in inverse time. By Equation \eqref{eq: kkts}, for any determined $\widehat{\beta}(\lambda)$, $\zeta_{(i)}$ will determine $W\upa_{(i)}$ as well as $W\upb_{(i)}$ for all $i$, thus determine $B_{(i)}$. Such a fact validates our claim that $\{\mathcal{F}_j\}_{j=1}^m$ is indeed a filtration. Furthermore, by the fact that $W\upa$ and $W\upb$ share the same sign, there holds: 
    \begin{itemize}
        \item $\Fs_{k}$ includes $\{W\upa_{(i)}: i>k\}$, $V^-_s(k):=\#\{W\upa_{(i)}<0, i\le k\}$, and $V^+_s(k):=\#\{W\upa_{(i)}>0, i\le k\}=k-V^-_s(k)$.
        \item the event $\{J < k\}$ is defined by, for all $l$ such that $|W\upa_{(l)}|\le|W\upb_{(k)}|$:
    \begin{align*}
        &\frac{|\{i:W_{(i)}\upa\le-|W\upa_{(l)}|\}|}{1\vee|\{i:W_{(i)}\upa\ge |W\upa_{(l)}|\}|}> q,\mbox{ for Split Knockoff};
        &\frac{1+|\{i:W_{(i)}\upa\le-|W\upa_{(l)}|\}|}{1\vee|\{i:W_{(i)}\upa\ge |W\upa_{(l)}|\}|}> q,\mbox{ for Split Knockoff+},
    \end{align*}
    i.e. is decided by $\#\{i:W\upa_{(i)}\le-|W\upa_{(l)}|\}$ and $\#\{i:W\upa_{(i)}\ge |W\upa_{(l)}|\}$. 
    \end{itemize}
    Hence it suffices to show that, for all $l$ such that $|W\upa_{(l)}|\le |W\upb_{(k)}|$, $\Fs_{k}$ includes $$\#\{i:W\upa_{(i)}\le-|W\upa_{(l)}|\} \mbox{    and     } \#\{i:W\upa_{(i)}\ge |W\upa_{(l)}|\}. $$ 
    To see this, the first set is decomposed by
    \begin{align*}
        &\#\{W\upa_{(i)}\le-|W\upa_{(l)}|\}\\
        = & \#\{W\upa_{(i)}\le-|W\upa_{(l)}|, i\le k\} + \#\{W\upa_{(i)}\le-|W\upa_{(l)}|, i>k\},\\
        = & \#\{W\upa_{(i)}<0, i\le k\}+ \#\{W\upa_{(i)}\le-|W\upa_{(l)}|, i>k\}.
    \end{align*}
    The last step is by for $i\le k$, $|W\upa_{(i)}|\ge |W\upb_{(i)}|\ge|W\upb_{(k)}|\ge |W\upa_{(l)}|$. Here the first part is equal to $V^-_s(k):=\#\{W\upa_{(i)}<0, i\le k\}$, the second part is determined by $\{W\upa_{(i)}, i>k\}$, 
    both of which are in the $\sigma$-field of $\Fs_{k}$. This shows that $\Fs_{k}$ includes $\#\{i:W\upa_{(i)}\le-|W\upa_{(l)}|\}$.
   Similarly, the second set is decomposed by
    \begin{align*}
        &\ \ \ \ \#\{i:W\upa_{(i)}\ge |W\upa_{(l)}|\}\\
        &= \#\{i:W\upa_{(i)}\ge |W\upa_{(l)}|, i\le k\}+ \#\{i:W\upa_{(i)}\ge |W\upa_{(l)}|, i>k\},\\
        &= \#\{i:W\upa_{(i)}>0, i\le k\}+ \#\{i:W\upa_{(i)}\ge |W\upa_{(l)}|, i>k\},
    \end{align*}
    where last step is by for $i\le k$, $|W\upa_{(i)}|\ge |W\upb_{(i)}|\ge|W\upb_{(k)}|\ge |W\upa_{(l)}|$. The first part is equal to $\Vs^+(k)=\#\{W\upa_{(i)}>0, i\le k\}$ and the second part is determined by $\{W\upa_{(i)}: i>k\}$, both of which are in the $\sigma$-field of $\Fs_{k}$. Therefore $\Fs_{k}$ includes $\#\{i:W\upa_{(i)}\ge |W\upa_{(l)}|\}$. This finishes the proof.
\end{proof}

\subsection{Proof of Theorem \ref{thm: fdr hd}}
\label{sec: proof hd thm}

\begin{proof}

    When the event $\Upsilon$ occurs, by definition, $\tilde{y}$, $A_\beta$, $A_\gamma$ and $\tilde{\varepsilon}$ defined in Equation \eqref{eq: features hd} satisfy
    \begin{equation}
        \tilde{y}=A_\beta\beta^*_{\hat S_\beta}+A_\gamma\gamma^*_{\hat S_\gamma}+\tilde\varepsilon,
    \end{equation}
    where $\beta^*_{\hat S_\beta}$ is a subvector of $\beta^*$, consisting of the rows in $\hat S_\beta$, and $\gamma^*_{\hat S_\gamma}$ is a subvector of $\gamma^*$, consisting of the rows in $\hat S_\gamma$. 

    From Equation \eqref{eq:gamma} and Equation \eqref{eq:t_gamma}, in the high dimensional setting, the KKT conditions that the solution $\gamma(\lambda),\ \tilde{\gamma}(\lambda)$ --- both are $\R_+\to\R^{|\hat S_\gamma|}$ functions --- should satisfy are
\begin{subequations}
    \label{eq: kkts hd}
    \begin{align}
        \lambda\rho(\lambda) + \frac{\gamma(\lambda)}{\nu}&= \frac{D\widehat{\beta}(\lambda)}{\nu},\label{eq: feature kkt hd}\\
        \lambda\tilde{\rho}(\lambda) + \frac{\tilde{\gamma}(\lambda)}{\nu} &= \frac{D\widehat{\beta}(\lambda)}{\nu} +\underbrace{\left\{  - \diag(\vecs)\gamma^*_{\hat S_\gamma} + \frac{\tilde{A}^T_{\gamma,1}}{\sqrt{n_2}} \varepsilon_2 \right\}}_{=:\zeta},\label{eq: knockoff kkt hd}
    \end{align}
\end{subequations}
where $\rho(\lambda) \in \partial \|\gamma(\lambda)\|_1$, $\tilde{\rho}(\lambda) \in \partial \|\tilde{\gamma}(\lambda)\|_1$, and $\tilde{A}_{\gamma,1}$ is defined to be the matrix that takes first $n_2$ rows of $\tilde{A}_{\gamma}$. 
For shorthand notations, define $\zeta : = - \diag(\vecs)\gamma^*_{\hat S_\gamma} + \frac{\tilde{A}^T_{\gamma,1}}{\sqrt{n_2}} \varepsilon_2$. Follow the same steps as those in achieving Equation \eqref{eq: ortho of agamma1}, it can be shown that the $\zeta$ should satisfy
\begin{align}
    \zeta\sim \mathcal{N}\left(-\diag(\vecs)\gamma^*_{\hat S_\gamma}, \frac{1}{n_2}\diag(\vecs)(2I_{|\hat S_\gamma|}-\diag(\vecs)\nu)\sigma^2\right).
\end{align}
Such a fact suggests that $\zeta+\diag(\vecs)\gamma^*_{\hat S_\gamma}$ consists of independent Gaussian random variables. With this property, follow the same steps as in Section \ref{sec: proof thm} (proof of Theorem \ref{theorem: fdr}), there holds
for Split Knockoff
    \begin{equation}
           \E\left[\left.\frac{\left|\left\{i:i\in \hat{S}\upc\cap S_0\right\}\right|}{\left|\hat{S}\upc\right|+q^{-1}}\right|\Upsilon\right]\le q,\nonumber
    \end{equation}
and for Split Knockoff+
    \begin{equation}
           \E\left[\left.\frac{\left|\left\{i:i\in \hat{S}\upc\cap S_0\right\}\right|}{\left|\hat{S}\upc\right|\vee 1}\right|\Upsilon\right]\le q.\nonumber
    \end{equation}
    This ends the proof.
\end{proof}

\subsection{Proof of Proposition \ref{thm: sign consistency front}}
\label{sec:proof_path_consistency}

In this section, we approach the model selection consistency by constructing Primal-Dual Witness (PDW) of Split LASSO regularization paths, following the same treatment in the traditional LASSO problem \citep{Wainwright09} and Split Linearized Bregman Iterations \citep{SplitLBI,huang2020boosting}. We first introduce the concept of the successful Primal-Dual Witness which has a unique solution of Split LASSO; then we introduce where the incoherence condition for Split LASSO comes from and establish the no-false-positive and sign consistency of Split LASSO regularization path, i.e. Proposition \ref{thm: sign consistency front}.

We first list the KKT conditions that an optimal solution of the Split LASSO problem \eqref{eq:split_lasso} satisfies, as they will be commonly used throughout this section. The KKT conditions are 
\begin{subequations}\label{eq:slasso-kkt}
    \begin{align}
      0 &= - (\Sigma_X + L_D) \beta(\lambda) + \frac{D^T}{\nu} \gamma(\lambda) + \left\{\Sigma_X \beta^* + \frac{X^T}{n} \varepsilon \right\}, \label{eq:slasso-kkta}  \\
      \lambda\rho(\lambda)&= \frac{D\beta(\lambda)}{\nu} -  \frac{\gamma(\lambda)}{\nu}, \label{eq:slasso-kktb}
    \end{align}
    \end{subequations}
    where $\rho(\lambda)\in \partial\|\gamma(\lambda)\|_1$.

\subsubsection{Primal-Dual Witness}

 In this section, we will introduce the lemma that ensures the uniqueness of the successful Primal-Dual Witness for Split LASSO problem. In the beginning, we will give the definition of the PDW. The set of witness $(\widehat{\beta}^\lambda, \hat{\gamma}^\lambda, \hat{\rho}^\lambda)\in \R^p \times \R^m \times \R^m$ is constructed in the following way:

\begin{enumerate}
    \item First, we set $\hat\gamma_{S_0}^\lambda=0$, and obtain $(\hat\beta^\lambda, \hat\gamma_{\S_1}^\lambda)\in \R^p \times \R^{|\S_1|}$ by solving
    \begin{align}
        (\hat\beta^\lambda, \hat\gamma_{\S_1}^\lambda) = \argmin_{(\beta, \gamma_{\S_1})}\left\{\frac{1}{2n}\|y - X\beta\|_2^2 + \frac{1}{2\nu}\|D\beta - \gamma\|_2^2 + \lambda \|\gamma_{\S_1}\|_1\right\}.\label{subproblem}
    \end{align}
    \item Second, we choose $\hat\rho_{\S_1}^\lambda=\partial \|\hat\gamma_{S_1}^\lambda\|_1$ as the subgradient of $\|\hat\gamma_{S_1}^\lambda\|_1$.
    \item Third, for no-false-positive, we solve for $\hat\rho_{\S_0}^\lambda \in \R^{\S_0}$ satisfying the KKT condition \eqref{eq:slasso-kkt}, and check whether or not the dual feasibility condition $|\hat\rho_j^\lambda|< 1$ for all $j\in \S_0$ is satisfied. 
    \item Fourth, for model selection (sign) consistency, we check whether $\hat\rho_{\S_1}^\lambda = \sign(\beta^*_{\S_1})$ is satisfied.
\end{enumerate}

Then we can give the following lemma.

\begin{lemma}
    When the PDW succeed, if the subproblem \eqref{subproblem} is strictly convex, the solution $(\hat\beta, \hat\gamma)$ is the unique optimal solution for split LASSO.
\end{lemma}

\begin{proof}
    When PDW succeed, we have $\|\hat\rho_{\S_0}\|_\infty<1$, and therefore $(\hat\beta, \hat\gamma)$ is a set of optimal solution, while $\hat\rho$ is in the subgradient of $\|\hat\gamma\|_1$. Let $(\tilde{\beta}, \tilde{\gamma})$ be any other optimal solution for Split LASSO. Denote
    \begin{align*}
        F(\beta, \gamma) = \frac{1}{2n}\|y - X\beta\|_2^2 + \frac{1}{2\nu}\|D\beta - \gamma\|_2^2.
    \end{align*}
    Then there holds
    \begin{align}
        F(\hat\beta, \hat\gamma) + \lambda\langle\hat\rho, \hat\gamma\rangle = F(\tilde\beta, \tilde\gamma)+\lambda\|\tilde\gamma\|_1,\nonumber
    \end{align}
    which is 
    \begin{align}
        F(\hat\beta, \hat\gamma) - \lambda\langle\hat\rho, \tilde\gamma - \hat\gamma\rangle - F(\tilde\beta, \tilde\gamma) = \lambda(\|\tilde\gamma\|_1-\langle\hat\rho, \tilde\gamma \rangle).\nonumber
    \end{align}
    Also, by Equation \eqref{eq:slasso-kkt}, there holds $\frac{\partial F(\hat\beta, \hat\gamma)}{\partial \hat\beta} = 0$, and $\frac{\partial F(\hat\beta, \hat\gamma)}{\partial \hat\gamma} = -\lambda \hat\rho$. Therefore
    \begin{align}
        F(\hat\beta, \hat\gamma) + \langle\frac{\partial F(\hat\beta, \hat\gamma)}{\partial \hat\gamma}, \tilde\gamma - \hat\gamma\rangle + \langle\frac{\partial F(\hat\beta, \hat\gamma)}{\partial \hat\beta}, \tilde\beta-\hat\beta\rangle - F(\tilde\beta, \tilde\gamma) = \lambda(\|\tilde\gamma\|_1-\langle\hat\rho, \tilde\gamma \rangle).\label{lhs<0}
    \end{align}
    Since $F$ is convex, the left hand side of Equation \eqref{lhs<0} is non-positive, then there holds
    \begin{align}
        \|\tilde\gamma\|_1 \le \langle\hat\rho, \tilde\gamma \rangle.\nonumber
    \end{align}
    Since $|\hat\rho_{S_0}|<1$, there holds $\tilde\gamma_{S_0}=0$. Therefore $(\tilde\beta, \tilde\gamma)$ is also an optimal solution for the subproblem \eqref{subproblem}. Thus if the subproblem \eqref{subproblem} is strictly convex, then $(\hat\beta, \hat\gamma)$ is the only solution for split LASSO.
\end{proof}

\subsubsection{Incoherence Condition and Path Consistency}

In this section, We first show on how the incoherence condition in Proposition \ref{thm: sign consistency front} is formalized and then we give the proof of Proposition \ref{thm: sign consistency front}.

 From $D[\Sigma_X+L_D]^{-1}\times$ Equation \eqref{eq:slasso-kkta} + $\nu\times$ Equation \eqref{eq:slasso-kktb}, and the fact that $\gamma^* = D\beta^*$, there holds for the solution $(\hat\beta, \hat\gamma)$ to Equation \eqref{eq:split_lasso} that
\begin{align*}
    \lambda\nu \hat\rho = & -\hat\gamma+\frac{D[\Sigma_X+L_D]^{-1}D^T \hat\gamma}{\nu} + D[\Sigma_X+L_D]^{-1}(\Sigma_X+L_D-L_D)\beta^*+\ldots\nonumber\\
    &\ldots+D[\Sigma_X+L_D]^{-1}\frac{X^T}{n}\varepsilon,\\
    = & -H_\nu (\hat\gamma - \gamma^*)+\omega,
\end{align*}
where $\omega = D[\Sigma_X+L_D]^{-1}\frac{X^T}{n}\varepsilon$. From the definition of $H_\nu^{11}$, $H_\nu^{00}$, $H_\nu^{10}$, and $H_\nu^{01}$, there further holds
\begin{align}
    \lambda \nu
    \begin{bmatrix}
        \hat\rho_{\S_1}\\
        \hat\rho_{\S_0}
    \end{bmatrix}
    =-
    \begin{bmatrix}
        H_\nu^{11} & H_\nu^{10}\\
        H_\nu^{01} & H_\nu^{00}
    \end{bmatrix}
    \begin{bmatrix}
        \hat\gamma_{\S_1} - \gamma^*_{\S_1}\\
        0_{\S_0}
    \end{bmatrix}
    +
    \begin{bmatrix}
        \omega_{\S_1}\\
        \omega_{\S_0}
    \end{bmatrix},\nonumber
\end{align}
which means
\begin{subequations}
    \begin{align}
        \lambda\nu \hat\rho_{\S_1} = & - H_\nu^{11}(\hat\gamma_{\S_1} - \gamma^*_{\S_1}) + \omega_{\S_1}\label{eq1: irr},\\
        \lambda\nu \hat\rho_{\S_0} = & -  H_\nu^{01}(\hat\gamma_{\S_1} - \gamma^*_{\S_1}) + \omega_{\S_0}\label{eq2: irr}.
    \end{align}
\end{subequations}
Since $H_\nu^{11}$ is reversible, $\hat\gamma_{\S_1} - \gamma^*_{\S_1}$ can be solved from Equation \eqref{eq1: irr}, and there holds
\begin{align}
    \hat\gamma_{\S_1} - \gamma^*_{\S_1} = -\lambda\nu [H_\nu^{11}]^{-1} \hat\rho_{\S_1} + [H_\nu^{11}]^{-1}\omega_{\S_1}.\label{solve_gamma}
\end{align}
Then, plug Equation \eqref{solve_gamma} into Equation \eqref{eq2: irr}, there further holds
\begin{align}
    \hat\rho_{\S_0} =  H_\nu^{01} [H_\nu^{11}]^{-1} \hat\rho_{\S_1} +\frac{1}{\lambda\nu}\{ \omega_{\S_0} - H_\nu^{01}[H_\nu^{11}]^{-1}\omega_{\S_1}\}.\label{estimation_rho}
\end{align}
The $\nu$-Incoherence Condition can be now formalized from the right hand side of Equation \eqref{estimation_rho}. Below, we are going to give the proof of Proposition \ref{thm: sign consistency front}.

\begin{proof}[Proof of Proposition \ref{thm: sign consistency front}]
    ~\\
    
    From Equation \eqref{estimation_rho} and $\nu$-Incoherence Condition, there holds
    \begin{align}
        \|\hat\rho_{\S_0}\|_\infty \le & (1-\chi_\nu) + \frac{1}{\lambda\nu}[\|\omega_{\S_0}\|_\infty + (1-\chi_\nu)\|\omega_{\S_1}\|_\infty],\nonumber\\
        \le & (1-\chi_\nu) + \frac{2}{\lambda}\left\|\frac{\omega}{\nu}\right\|_\infty.\nonumber
    \end{align}
    By definition $\frac{\omega}{\nu} = D\left[\nu \frac{X^TX}{n} +D^TD\right]^{-1}\frac{X^T}{n}\varepsilon$, therefore 
    \begin{align}
        \Prob\left(\left\|\frac{2\omega}{\lambda\nu}\right\|_\infty\ge \frac{\chi_\nu}{2}\right)\le 2m \exp\left(-\frac{n}{2c\sigma^2}\frac{\lambda^2\chi_\nu^2}{16}\right),\nonumber
    \end{align}
    for some constant $c>0$ related with $D$ and $X$. Take $\lambda = \lambda_n> \frac{8}{\chi_\nu}\sqrt{\frac{c\sigma^2\log m}{ n}}$, there holds
    \begin{align}
        \Prob\left(\|\hat\rho_{\S_0}\|_\infty>1-\frac{\chi_\nu}{2}\right) \le 2\exp\left(-c_1 n\lambda_n^2\right).\label{eq: no false positive}
    \end{align}
    for some constant $c_1>0$.
    
    Take $\lambda = \lambda_n$ in Equation \eqref{solve_gamma} and consider the infinity norm on both sides, there holds
    \begin{align}
        \|\hat\gamma_{\S_1} - \gamma^*_{\S_1}\|_\infty \le \lambda_n\nu\| [H_\nu^{11}]^{-1}\|_\infty + \|[H_\nu^{11}]^{-1}\omega_{\S_1}\|_\infty.\nonumber
    \end{align}
    Note that the first term in the right hand side is a deterministic term, therefore we can only estimate the second term. 
    By definition $\frac{\omega}{\nu} = D\left[\nu \frac{X^TX}{n} +D^TD\right]^{-1}\frac{X^T}{n}\varepsilon$, similarly like above, there holds
    \begin{align}
        \Prob(\|[H_\nu^{11}]^{-1}\omega_{\S_1}\|_\infty>\nu t) \le 2|S_1|\exp\left(-\frac{n}{2c\sigma^2}t^2C^2_\mathrm{min}\right).\nonumber
    \end{align}
    for some constant $c>0$. Take $t = \frac{\sigma\lambda_n}{2C_{\min}}$, there holds
    \begin{align}
        \Prob(\|[H_\nu^{11}]^{-1}\omega_{\S_1}\|_\infty>\lambda_n\nu \frac{\sigma}{2C_\mathrm{min}}) \le 2\exp\left(-c_2 n\lambda_n^2\right),\nonumber
    \end{align}
    for some constant $c_2>0$. After all, there holds
    \begin{align}
        \|\hat\gamma_{\S_1}-\gamma^*_{\S_1}\|_\infty\le \lambda_n \nu \left[\frac{\sigma}{2C_\mathrm{min}}+\| [H_\nu^{11}]^{-1}\|_\infty\right],\nonumber
    \end{align}
    with probability greater than $1-2\exp(-c_2\lambda_n^2 n)$. 
    
    Take $C = \max\{8\sqrt{c}, c_1, c_2\}$. Then from Equation \eqref{eq: no false positive}, for $\lambda = \lambda_n>\frac{C}{\chi_\nu}\sqrt{\frac{\sigma^2\log m}{n}}$, $\hat\gamma$ will not have false discoveries with probability greater than $1-4\exp(-Cn \lambda_n^2)$. 
    If there further holds $\min_{i\in \S_1}{\gamma^*_i}> \lambda_n \nu \left[\frac{\sigma}{2C_\mathrm{min}}+\| [H_\nu^{11}]^{-1}\|_\infty\right]$, $\hat\gamma$ recovers the support set of $\gamma^*$.
\end{proof}

\begin{remark}
The influence of $\nu$ on the power can be understood from this theorem as follows: (a) for the early stage of the Split LASSO path characterized by $\lambda_n > \frac{C}{\chi_\nu}\sqrt{\frac{\sigma^2\log m}{n}}.$, there is no false positive and only nonnull features are selected here; (b) all the strong nonnull features whose magnitudes are larger than $O(\nu \sigma \chi_\nu^{-1}\sqrt{\log m/n} )$ could be selected on the path with sign consistency. Hence a sufficiently large $\nu$ will ensure the incoherence condition for path consistency such that strong nonnull features will be selected earlier on the Split LASSO path than the nulls, at the cost of possibly losing weak nonnull features below $O(\nu \sigma \chi_\nu^{-1}\sqrt{\log m/n} )$. Therefore, a good power must rely on a proper choice of $\nu$ for the trade-off.  
\end{remark}

\begin{remark}
There is a close relationship between Proposition \ref{thm: sign consistency front} in this paper and the model selection consistency in \cite{huang2020boosting}. First of all, both of them aim to address the regression with transformational sparsity problem, in which the most popular method is the generalized LASSO \eqref{eq:gen_lasso}. However, the model selection consistency of generalized LASSO suffers from the incoherence condition which often fails in applications. To alleviate this issue, the variable splitting idea is adopted in both approaches, by relaxing the linear constraint to an Euclidean neighborhood controlled with a proximity parameter $\nu>0$. In this paper, the Split LASSO in \eqref{eq:split_lasso} adopts a new objective function to optimize; while \cite{huang2020boosting} adopts a dynamical system called the inverse scale space which does not optimize an objective function but renders an iterative regularization path in its discrete form. In terms of model selection consistency, Proposition \ref{thm: sign consistency front} shows that Split LASSO achieves sign consistency under strong signal assumption and a family of $\nu$-incoherence conditions which is weaker than that for generalized LASSO as $\nu$ increases, easier to meet in applications. Such a theoretical result is similar to the model selection consistency in \cite{huang2020boosting} when a proper early stopping regularization is chosen instead. It is shown in \cite{huang2020boosting} that on such iterative regularization paths, there exists an estimator with model selection consistency, provided that the signal is strong enough and the same family of $\nu$-incoherence conditions (Equation (2.5), Page 12 in \cite{huang2020boosting}) are satisfied. 

However, the incoherence conditions above are unknown in practice since the true support set is unknown, hence they can not be used for data adaptive model selection. The departure of our current paper aims to develop a data adaptive model selection method towards FDR control. Here, Theorem \ref{theorem: fdr} and its high dimensional extension Theorem \ref{thm: fdr hd}, show that the FDR of Split Knockoffs can be uniformly controlled for all parameter $\nu>0$. 

On the other hand, the selection power of Split Knockoffs, is $\nu$-dependent, which can be observed from the simulation experiments in Section \ref{Sec: simulation results}, where the selection power of Split Knockoffs undergoes a first increase then decrease as $\nu$ grows. Such a phenomenon is possibly explained by Proposition \ref{thm: sign consistency front}. In Proposition \ref{thm: sign consistency front}, there are two folds of influence of $\nu$ on model selection consistency of Split LASSO. On the one hand, the $\nu$-incoherence condition of Split LASSO --- which is critical for discovering strong signals --- will be easier to be satisfied with the increase of $\nu$. On the other hand, the requirement on the signal noise ratio \eqref{eq: min snr} in achieving the sign consistency becomes harder to satisfy and weak signals might get lost with the increase of $\nu$. As a consequence, the selection power of Split Knockoffs shows a first increasing then decreasing trend, in the simulation experiments in Section \ref{Sec: simulation results}. Therefore, a good choice of $\nu$ in terms of power will depend on such a trade-off. This suggests us to use cross-validation to select a good $\hat{\nu}$ to reach a good selection power with a desired FDR control.  
\end{remark}

\section{Supplementary Material on Simulation Experiments}

In this section, we will provide various supplementary material on simulation experiments for Split Knockoffs. In particular, we will present the simulation experiments of Split Knockoffs in the cases where $m$ is close to $n$ and in high dimensional settings. Then, we will discuss the choice of the data splitting fraction in Split Knockoffs, as well as the robustness of the random data splitting. After that, simulation experiments are implemented to compare the performance of Split Knockoffs and Knockoffs when the signal strength varies. In the end, we will discuss the computational cost of Split Knockoffs.

\subsection{Simulation Experiments where $m$ is close to $n$}

\label{sec: simulation m close to n}

In this section, we conduct simulation experiments in a setting where $m$ is closer to $n$ compared with that of Section \ref{sec: simulation_settings}. In particular, we consider the following setting, whose choices of $m$, $n$ and $p$ are close to those in the connection selection of Alzheimer’s Disease in Section \ref{sec: connection selection} ($m_{\mathrm{AD}} = 463$, $n_{\mathrm{AD}} = 752$, $p_{\mathrm{AD}}=90$).

In model \eqref{eq: model}, we generate $X\in \mathbb R^{n\times p}$ ($n=750$ and $p=100$) i.i.d. from $\Nm(0_p, \Sigma)$, where $\Sigma_{i,i}=1$ and $\Sigma_{i,j}=c^{|i-j|}$ for $i\neq j$, with feature correlation $c=0.5$. Define $\beta^*\in\mathbb R^p$ by
\begin{equation*}
    \beta_i^*:=\left\{
    \begin{array}{ccl}
        1   &   & i \le 20,\ i \equiv 0, -1 (\mathrm{mod}\ 3),\\
        0   &   & \mathrm{otherwise}.
    \end{array} \right.
\end{equation*}
Then $n$ linear measurements are generated by
$$y = X \beta^* + \varepsilon,$$
where $\varepsilon\in \mathbb R^n$ is generated i.i.d. from $\Nm(0, 1)$.

The linear transformation $D\in\R^{m\times p}$ is specified in the following way. Consider the graph $G = (V, E)$, where $V$ denotes the vertex set $V = \{1, 2, \cdots, p\}$, and $E$ denotes the edge set. There is an edge connecting two vertices $i\neq j$ if and only if 
\begin{align*}
    |i-j|\le 5 (\mathrm{mod}\ p).
\end{align*}
Through this construction, each vertex in $V$ is connected to 10 neighbouring vertices. Therefore, $|E| = 10*p/2 = 500$. Take $D$ to be the graph difference operator on $G$, then $m = 500>p = 100$. Then we generate $\gamma^* = D\beta^*$.

For Split Knockoffs, we take $\widehat{\beta}(\lambda)$ as a fixed cross validation optimal estimator $\widehat{\beta}_{\hat\nu, \hat{\lambda}}$ 
in the Split LASSO path with dataset $\D_1=(X_1, y_1)$. The dataset $\D = (X, y)$ is randomly split into two parts $\D_1 = (X_1, y_1)$ and $\D_2 = (X_2, y_2)$ with $n_1$ and $n_2$ samples respectively, where $n_1 = 150$ and $n_2 = n-n_1 = 600 = m+p$. The performance of Split Knockoffs is presented in Figure \ref{fig: simulation where m close to n}.

\begin{figure}[!ht]
\centering
\subfigure[$\Ws$]{
\begin{minipage}[t]{0.3\textwidth}
\centering
\includegraphics[width=\textwidth]{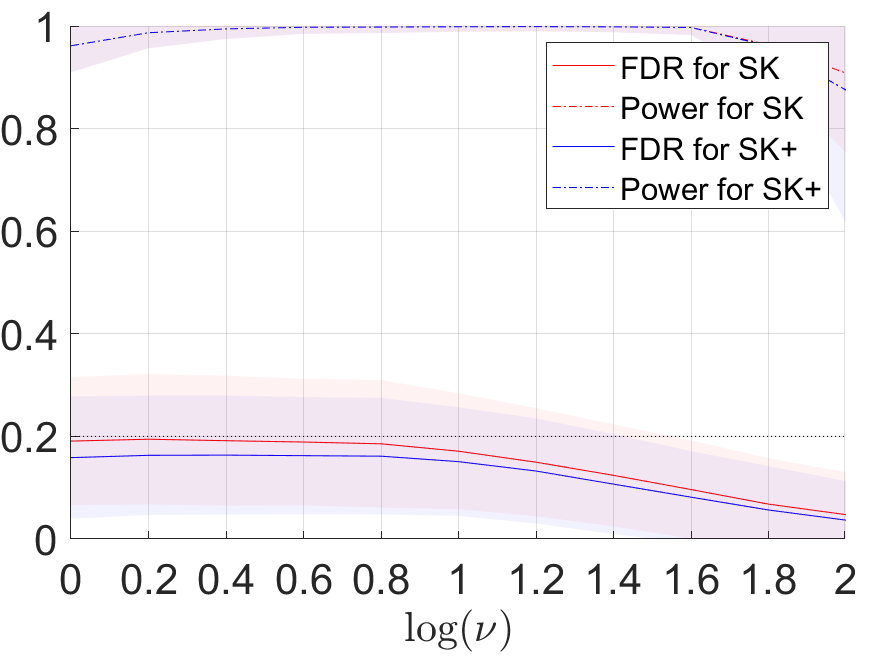}
\end{minipage}%
}%
\subfigure[$\Wst$]{
\begin{minipage}[t]{0.3\textwidth}
\centering
\includegraphics[width=\textwidth]{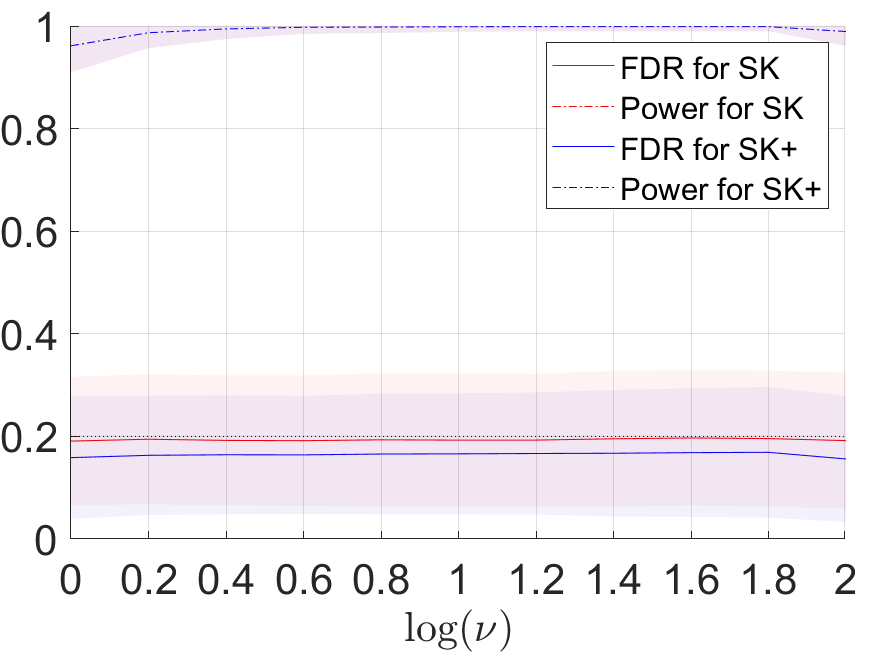}
\end{minipage}%
}%
\subfigure[$\Wbc$]{
\begin{minipage}[t]{0.3\textwidth}
\centering
\includegraphics[width=\textwidth]{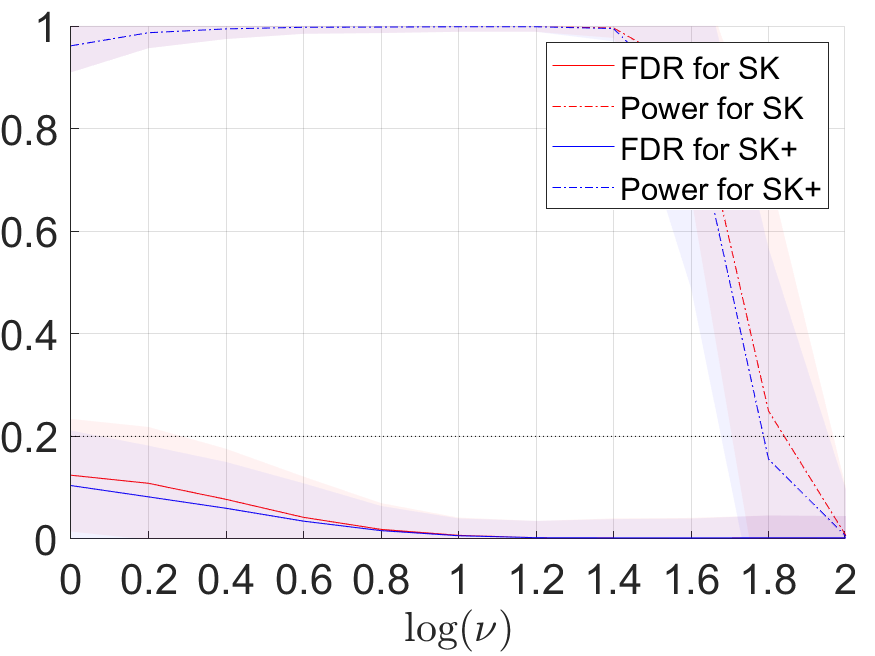}
\end{minipage}%
}%

\caption{The performance of Split Knockoffs: FDR and Power for $q=0.2$. $\widehat{\beta}(\lambda)$ is taken as a fixed cross validation optimal estimator $\widehat{\beta}_{\hat\nu, \hat{\lambda}}$. The curves in the figures represent the average performance of Split Knockoffs in FDR and Power in 200 simulation instances, while the shaded areas represent the 80\% confidence intervals truncated to the range $[0, 1]$.}
\label{fig: simulation where m close to n}
\end{figure}

As presented in Figure \ref{fig: simulation where m close to n}, the performance of Split Knockoffs follows the same trend as that of the other simulation experiments in Section \ref{Sec: simulation results}. The FDR is controlled universally for $\Ws$, $\Wst$ and $\Wbc$, where $\Wst$ exhibits the least conservative FDR control as suggested by Proposition \ref{prop: inequality of w statistics}. Meanwhile, the selection power of Split Knockoffs presents a first increase then decrease trend as suggested by Proposition \ref{thm: sign consistency front}, where $\Wst$ exhibits the best selection power, as predicted by Proposition \ref{prop: relations}.

\subsection{Simulation Experiments for Split Knockoffs in High Dimensional Settings}

\label{sec: hd_simulation}

In this section, we will show the results of the simulation experiments on Split Knockoffs  in the high dimensional setting in a similar setting as in Section \ref{sec: simulation_settings}.
In model \eqref{eq: model}, we generate $X\in \mathbb R^{n\times p}$ ($n=400$ and $p=1000$) i.i.d. from $\Nm(0_p, \Sigma)$, where $\Sigma_{i,i}=1$ and $\Sigma_{i,j}=c^{|i-j|}$ for $i\neq j$, with feature correlation $c=0.5$. Define $\beta^*\in\mathbb R^p$ by
\begin{equation*}
    \beta_i^*:=\left\{
    \begin{array}{ccl}
        1   &   & i \le 20,\ i \equiv 0, -1 (\mathrm{mod}\ 3),\\
        0   &   & \mathrm{otherwise},
    \end{array} \right.
\end{equation*}
in the same way as in Section \ref{sec: simulation_settings}.
Then we generate $n$ linear measurements,
$y = X \beta^* + \varepsilon,$
where $\varepsilon\in \mathbb R^n$ is generated i.i.d. from $\Nm(0, 1)$. 
For transformational sparsity, we should specify the linear transformation $D$ such that $\gamma^*=D\beta^*$, where $\gamma^*$ is sparse. We choose three types of transformation $D_1$, $D_2$, $D_3$ in the same way as in Section \ref{sec: simulation_settings}.

\begin{figure}[!ht]
\centering
\subfigure[$\Ws$ in $D_1$]{
\begin{minipage}[t]{0.3\textwidth}
\centering
\includegraphics[width=\textwidth]{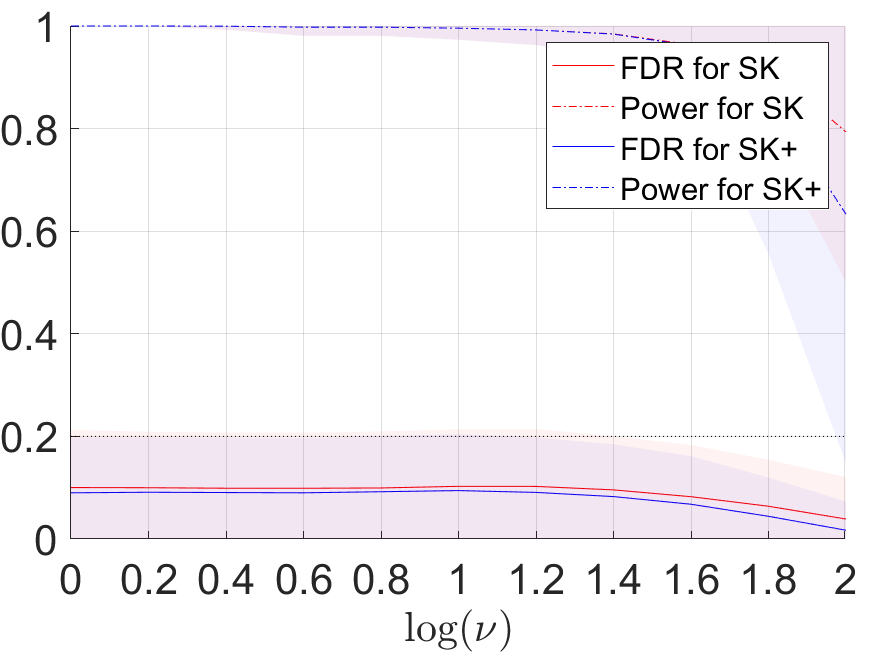}
\end{minipage}%
}%
\subfigure[$\Ws$ in $D_2$]{
\begin{minipage}[t]{0.3\textwidth}
\centering
\includegraphics[width=\textwidth]{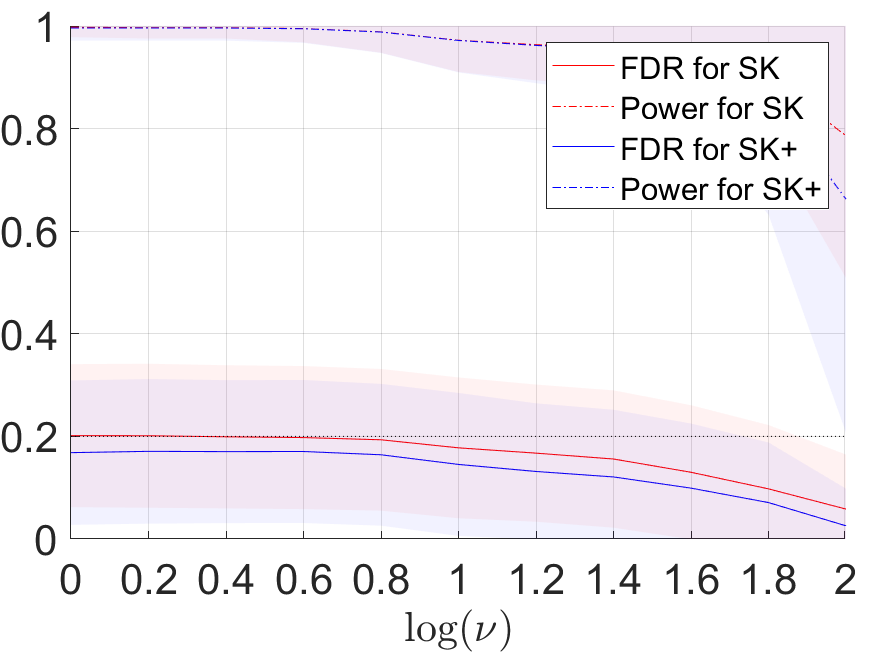}
\end{minipage}%
}%
\subfigure[$\Ws$ in $D_3$]{
\begin{minipage}[t]{0.3\textwidth}
\centering
\includegraphics[width=\textwidth]{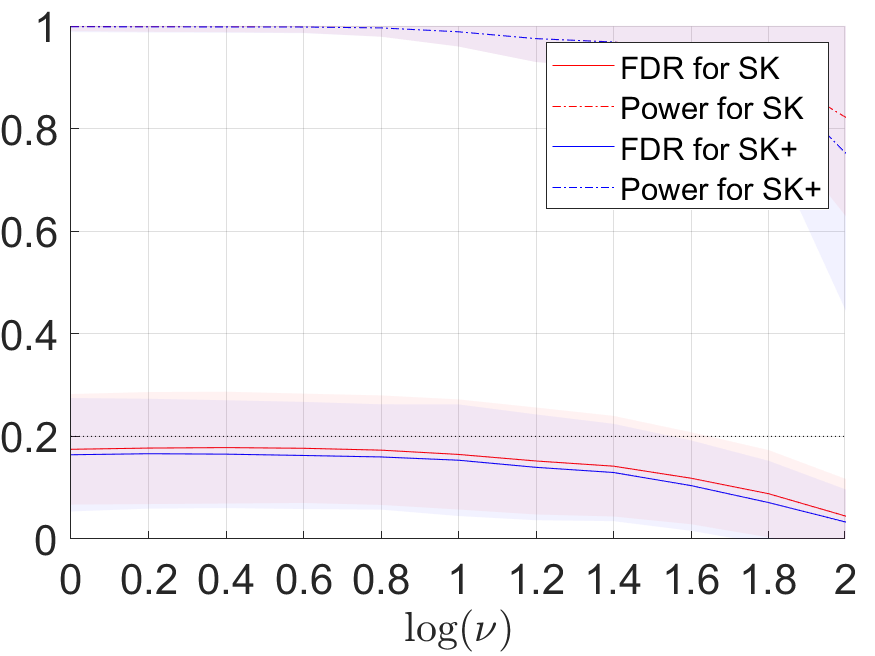}
\end{minipage}%
}%

\centering
\subfigure[$\Wst$ in $D_1$]{
\begin{minipage}[t]{0.3\textwidth}
\centering
\includegraphics[width=\textwidth]{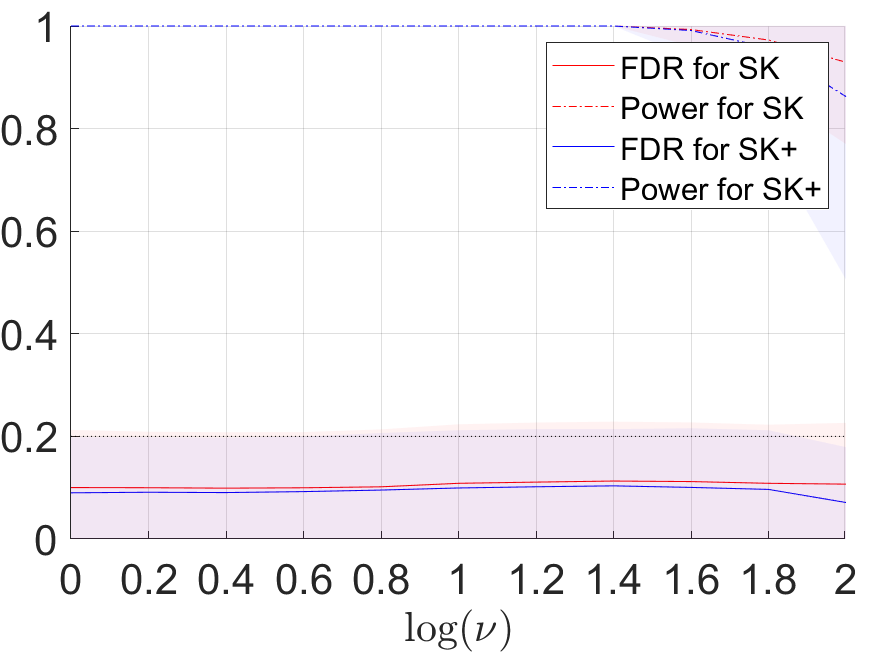}
\end{minipage}%
}%
\subfigure[$\Wst$ in $D_2$]{
\begin{minipage}[t]{0.3\textwidth}
\centering
\includegraphics[width=\textwidth]{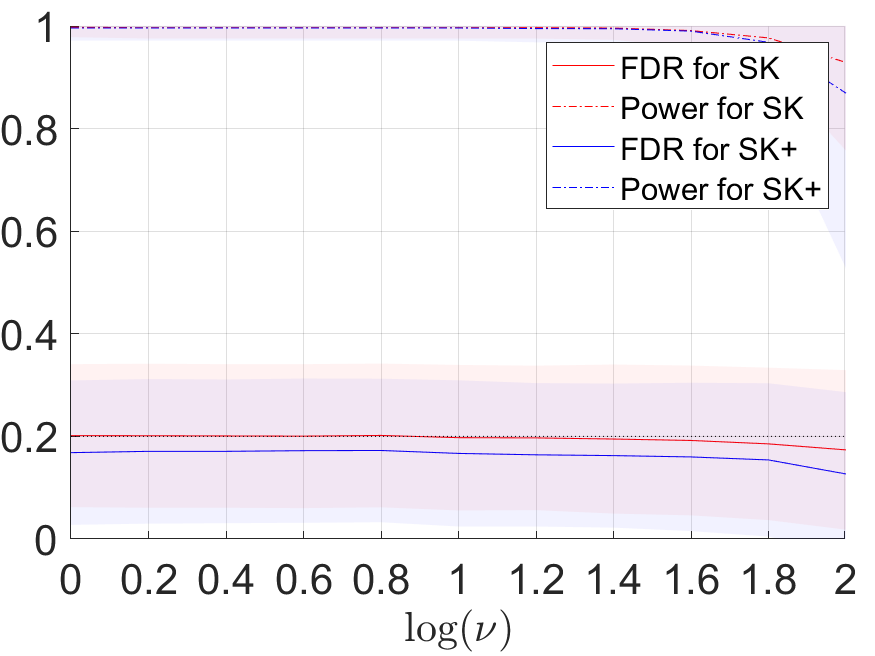}
\end{minipage}%
}%
\subfigure[$\Wst$ in $D_3$]{
\begin{minipage}[t]{0.3\textwidth}
\centering
\includegraphics[width=\textwidth]{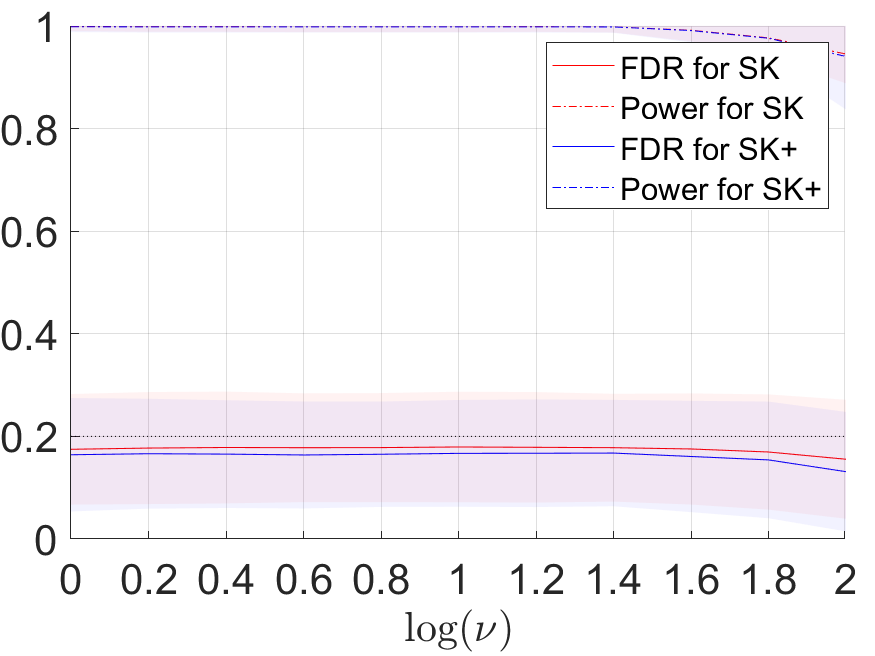}
\end{minipage}%
}%

\centering
\subfigure[$\Wbc$ in $D_1$]{
\begin{minipage}[t]{0.3\textwidth}
\centering
\includegraphics[width=\textwidth]{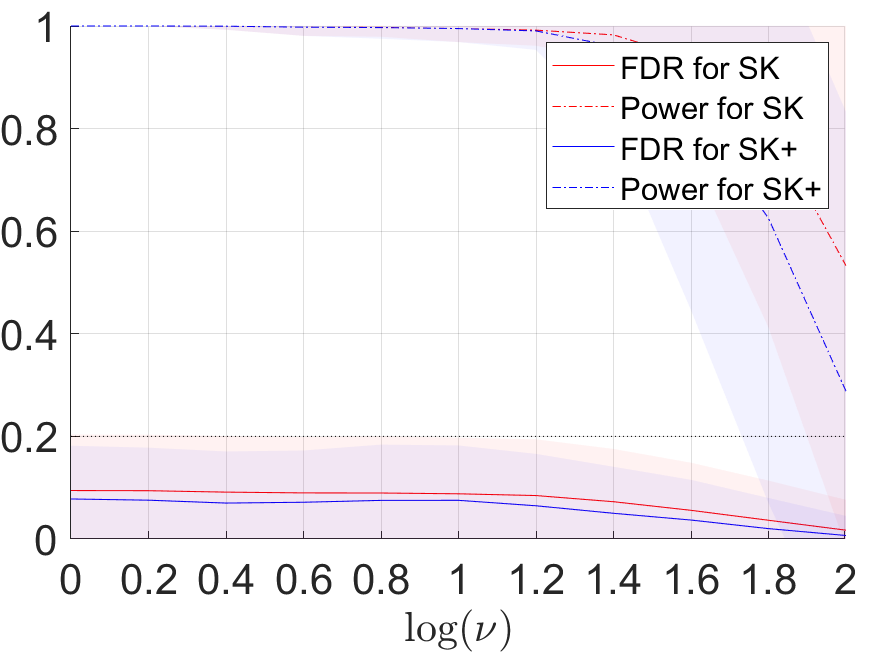}
\end{minipage}%
}%
\subfigure[$\Wbc$ in $D_2$]{
\begin{minipage}[t]{0.3\textwidth}
\centering
\includegraphics[width=\textwidth]{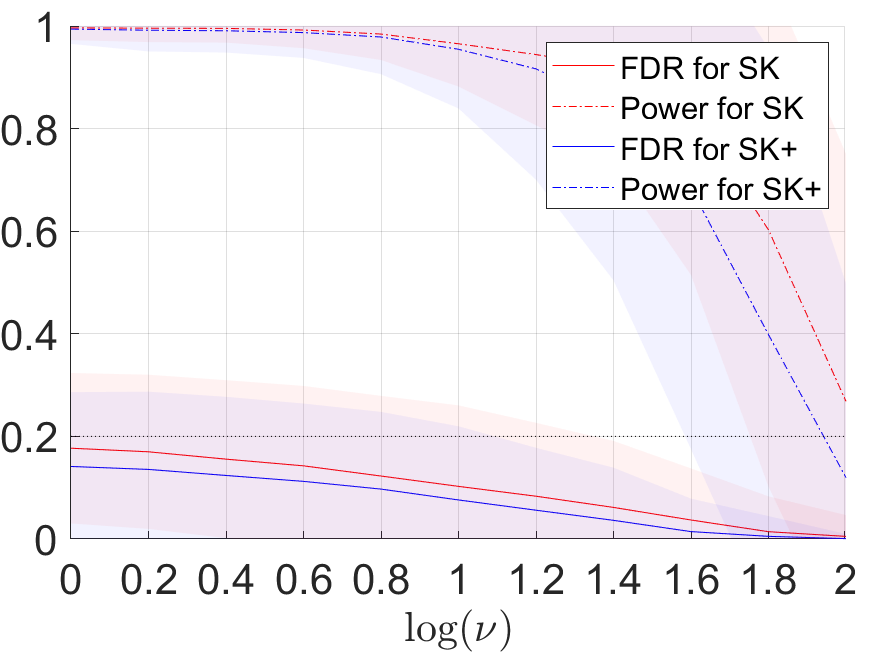}
\end{minipage}%
}%
\subfigure[$\Wbc$ in $D_3$]{
\begin{minipage}[t]{0.3\textwidth}
\centering
\includegraphics[width=\textwidth]{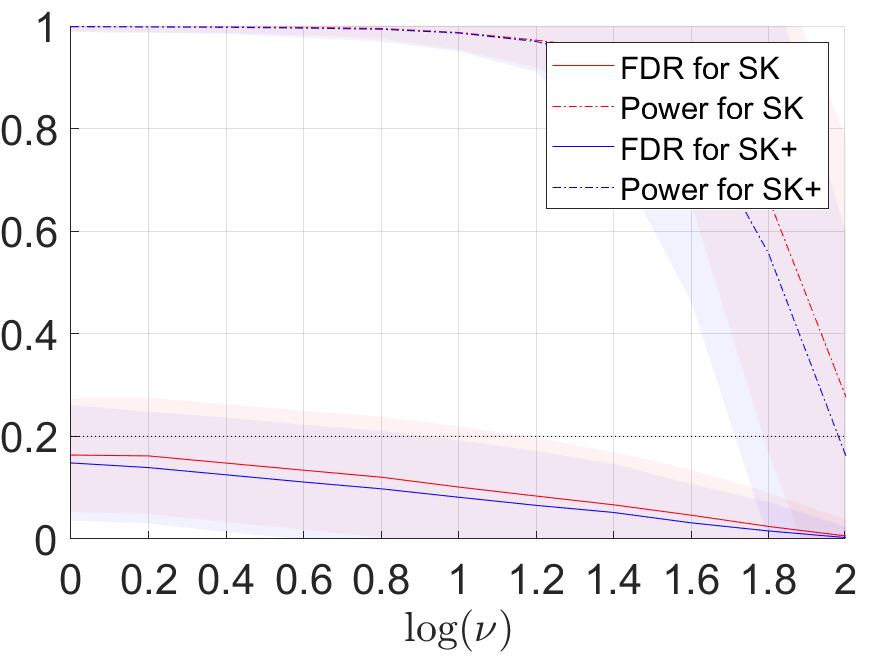}
\end{minipage}%
}%

\caption{The performance of Split Knockoffs in high dimensional settings: FDR and Power for $q=0.2$. $\widehat{\beta}(\lambda)$ is taken as a fixed cross validation optimal estimator $\widehat{\beta}_{\hat\nu, \hat{\lambda}}$. The curves in the figures represent the average performance of Split Knockoffs in FDR and Power in 200 simulation instances, while the shaded areas represent the 80\% confidence intervals truncated to the range $[0, 1]$.}
\label{fig: simulation with beta hat hd}
\end{figure}

In simulation experiments, we use \url{glmnet} package \citep{friedman2010regularization, simon2011regularization} to compute regularization paths for Split LASSO, etc.
For the data splitting, we randomly split the dataset $\D = (X, y)$ into two parts $\D_1 = (X_1, y_1)$ and $\D_2 = (X_2, y_2)$ with $n_1$ and $n_2$ samples respectively, where $n_1 = 100$ and $n_2 = 300$.

The estimated support sets $\hat{S}_\beta$, $\hat{S}_\gamma$ in this experiment are generated from $\D_1$ in the following two steps. First, to estimate a super support set $\hat{S}_\beta$, we perform the standard LASSO on a proper $\lambda>0$ (chosen from cross validation in the experiments),
\begin{align*}
    \min_{\beta} \frac{1}{2n}\|y_1 - X_1\beta\|_2^2 + \lambda\|\beta\|_1.
\end{align*}
Second, to estimate support set $\hat{S}_\gamma$, we perform the Split LASSO on the reduced model with respect to a proper $\lambda>0$ (chosen from cross validation in the experiments),
\begin{align*}
    \min_{\beta,\gamma} \frac{1}{2n}\left\|y_1 - X_{\hat{S}_\beta}\beta\right\|_2^2 + \frac{1}{2\nu}\|D\beta - \gamma\|_2^2 + \lambda\|\gamma\|_1,
\end{align*}
where $X_{\hat{S}_\beta}$ is the submatrix of $X_1$ consisting of the columns in $\hat{S}_\beta$. 

With the estimated support set $\hat{S}_\beta$ and $\hat{S}_\gamma$, we generate the intercept $\widehat{\beta}(\lambda) = \widehat\beta_{\hat\lambda, \hat\nu}$ as an optimal estimator with minimal cross validation loss (with respect to $\lambda$ and $\nu$) on the following Split LASSO regularization path
\begin{align*}
    \min_{\beta,\gamma} \frac{1}{2n}\left\|y_1 - X_{\hat{S}_\beta}\beta\right\|_2^2 + \frac{1}{2\nu}\|D_{\hat S_\beta, \hat S_\gamma}\beta - \gamma\|_2^2 + \lambda\|\gamma\|_1,\ \ \ \ \mbox{for $\lambda>0$,}
\end{align*}
where $D_{\hat S_\beta, \hat S_\gamma}$ is the submatrix of $D$, consisting of the columns in $\hat S_\beta$ and rows in $\hat S_\gamma$. With these choices of $\hat{S}_\beta$, $\hat{S}_\gamma$ and $\widehat{\beta}(\lambda)$, we proceed with the rest steps of Split Knockoffs following the instructions in Section \ref{sec: hd}.

As shown in Figure \ref{fig: simulation with beta hat hd}, all three versions of the Split Knockoffs achieve desired performance in FDR control in the high dimensional settings. It is worth to mention that similar with the performance of Split Knockoffs in the case $n\ge m+p$ in Section \ref{Sec: simulation results}, Split Knockoffs with $\Wbc$ statistics exhibits stricter FDR control, at the cost of the selection power --- especially when $\nu$ is large --- compared with $\Ws$ and $\Wst$, while $\Wst$ exhibits the highest selection power with the most adaptive FDR control with respect to the target. Such an observation 
is explained by Proposition \ref{prop: relations} that $\Wbc$ is the most conservative statistics, while $\Wst$ is the most aggressive statistics in terms of feature selection.

Moreover, we show in Table \ref{table: hd} that, for the calculation of the feature and knockoff significance, the cross validation optimal choice of $\nu$ --- $\hat\nu$ --- can still achieve high power and desired FDR in Split Knockoffs in high dimensional settings. Specifically, we present the performance of Split Knockoffs with all three versions of $W$ statistics, under the above choice of $\widehat{\beta}(\lambda)$ being taken as a fixed cross validation optimal estimator $\widehat{\beta}_{\hat\nu, \hat{\lambda}}$ on the Split LASSO paths, while $\nu$ is chosen to be $\hat\nu$. For such a choice, $\Wbc$ operates slightly stricter FDR control at the cost of slightly lower selection power compared with $\Ws$ and $\Wst$.

\begin{table}[!ht]
    \caption{FDR and Power for Split Knockoffs in high dimensional settings ($q = 0.2$). The intercept $\widehat{\beta}(\lambda)$ for Split Knockoffs is taken as a fixed cross validation optimal estimator $\widehat{\beta}_{\hat\nu, \hat{\lambda}}$, with the $\nu$ for calculating the feature and knockoff significance taken as $\hat\nu$. In this table, we present the average performance of Split Knockoffs in FDR and Power, together with the standard deviations in 200 simulation instances. For shorthand notations, we use "SK(+)" to refer to "Split Knockoff(+)".}
    \centering
    \resizebox{\textwidth}{!}{
    \begin{tabular}{c|ccc|ccc}
    \hline
        Performance & SK with $\Ws$ & SK with $\Wst$ & SK with $\Wbc$ & SK+ with $\Ws$ & SK+ with $\Wst$ & SK+ with $\Wbc$ \\
        \hline
        FDR in $D_1$  & 0.1003 &  0.1003 &  0.0945 &   0.0900 &  0.0900  & 0.0780  \\
        ~ & $\pm$0.0874 & $\pm$0.0874 & $\pm$0.0849 & $\pm$0.0840 & $\pm$0.0840 & $\pm$0.0807 \\
        Power in $D_1$  & 1.0000 & 1.0000 & 1.0000 & 1.0000 & 1.0000 & 1.0000 \\
        ~ & $\pm$0.0000 & $\pm$0.0000 & $\pm$0.0000 & $\pm$0.0000 & $\pm$0.0000 & $\pm$0.0000\\
        \hline
        FDR in $D_2$  & 0.2015 &  0.2015 & 0.1770  & 0.1683 &  0.1683 & 0.1414 \\
        ~ & $\pm$0.1084 & $\pm$0.1084 & $\pm$0.1141 & $\pm$0.1095  & $\pm$0.1095  & $\pm$0.1127 \\
        Power in $D_2$  & 0.9975 &  0.9975  & 0.9964& 0.9961 &  0.9961  & 0.9939 \\
        ~ & $\pm$0.0150 & $\pm$0.0150 & $\pm$0.0186 & $\pm$0.0192 & $\pm$0.0192 & $\pm$0.0224 \\
        \hline
        FDR in $D_3$  & 0.1749 & 0.1749 &  0.1637 & 0.1642  & 0.1642  & 0.1482  \\
        ~ & $\pm$0.0838 & $\pm$0.0838 & $\pm$0.0866 & $\pm$0.0860 & $\pm$0.0860 & $\pm$0.0876 \\
        Power in $D_3$  & 0.9991  & 0.9991 &  0.9991 & 0.9985 & 0.9985 &  0.9985 \\
        ~ & $\pm$0.0069 & $\pm$0.0069 & $\pm$0.0069 & $\pm$0.0082 & $\pm$0.0082 & $\pm$0.0082 \\
    \hline
      \end{tabular}
      }
      \label{table: hd}
\end{table}

\subsection{Simulation Experiments on Data Splitting Fraction}

\label{sec: split fraction}

In this section, we will use simulation experiments to study how the data splitting fraction can affect the performance of Split Knockoffs, and give an practical guideline on how to choose the data splitting fraction. We succeed all the simulation settings on the dataset $(X, y)$ and transformational sparsity $D_1$, $D_2$, $D_3$ from Section \ref{sec: simulation_settings}, and take the regression parameter $\beta^*\in\mathbb R^p$ as
\begin{equation}
    \beta_i^*:=\left\{
    \begin{array}{ccl}
        0.5   &   & i \le 20,\ i \equiv 0, -1 (\mathrm{mod}\ 3),\\
        0   &   & \mathrm{otherwise}.
    \end{array} \right.\label{eq: snr = 0.5}
\end{equation}
The signal strength used in this section was lowered compared with that of Section \ref{sec: simulation_settings} in order to increase the variation and amplify the effects of data splitting fractions on the performance of Split Knockoffs. Then we conduct simulation experiments with respect to different data splitting fractions in the range from 0.1 to 0.8 with a step size 0.1. The performance of Split Knockoffs in the FDR and selection power with respect to different data splitting fractions is shown in Figure \ref{fig: split size 0.5}.

\begin{figure}[!ht]
\centering
\subfigure[$\Ws$ in $D_1$]{
\begin{minipage}[t]{0.3\textwidth}
\centering
\includegraphics[width=\textwidth]{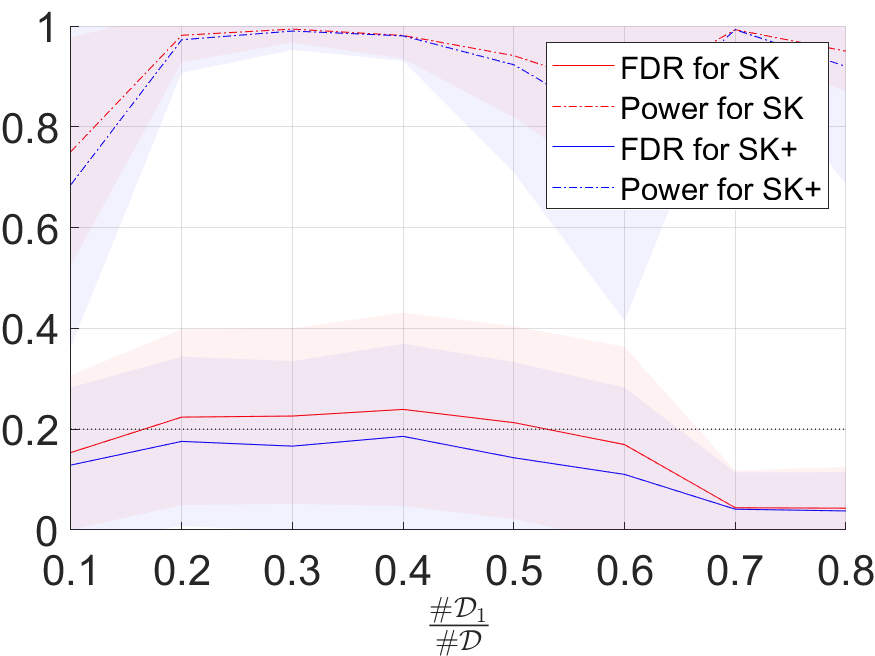}
\end{minipage}%
}%
\subfigure[$\Ws$ in $D_2$]{
\begin{minipage}[t]{0.3\textwidth}
\centering
\includegraphics[width=\textwidth]{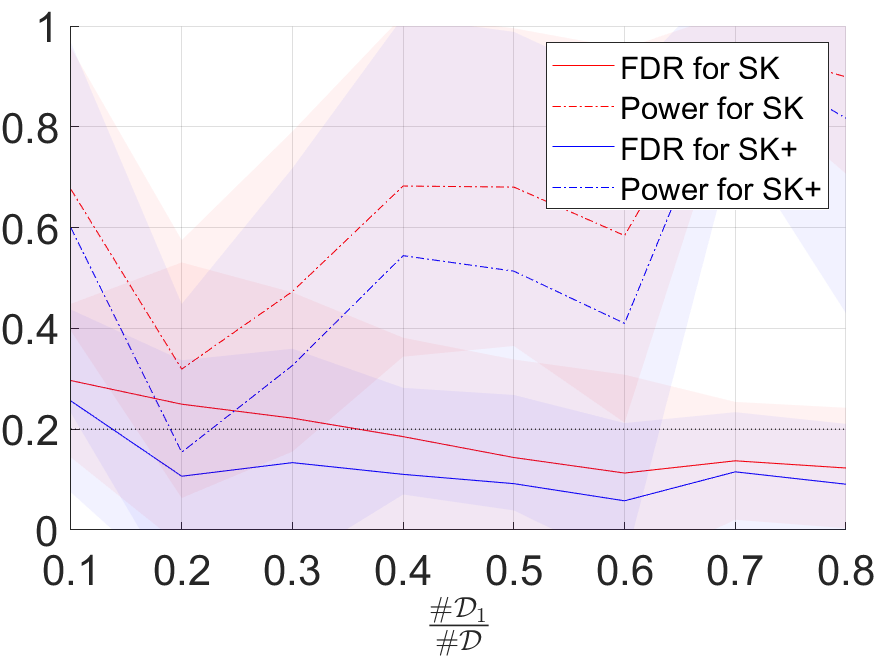}
\end{minipage}%
}%
\subfigure[$\Ws$ in $D_3$]{
\begin{minipage}[t]{0.3\textwidth}
\centering
\includegraphics[width=\textwidth]{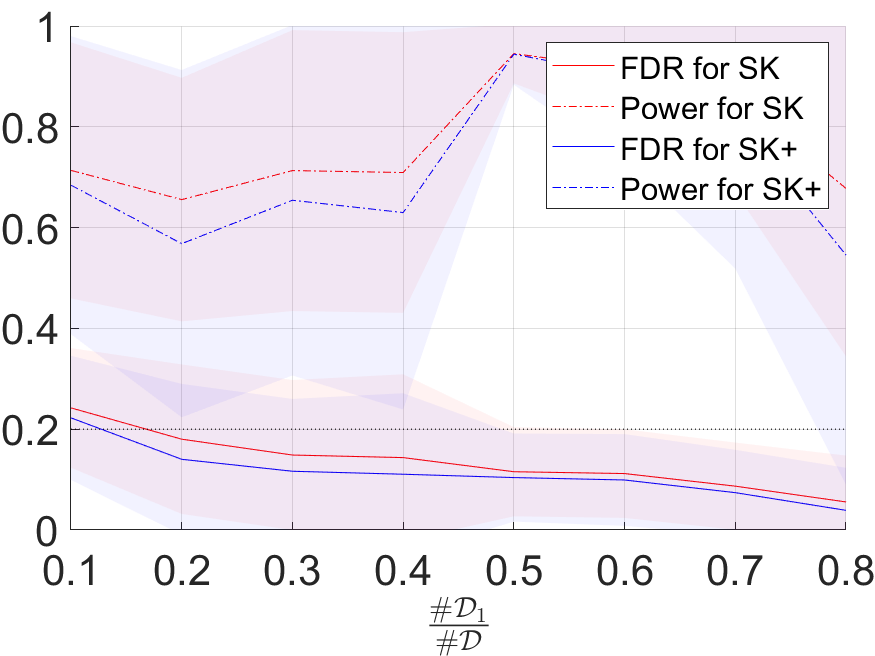}
\end{minipage}%
}%

\centering
\subfigure[$\Wst$ in $D_1$]{
\begin{minipage}[t]{0.3\textwidth}
\centering
\includegraphics[width=\textwidth]{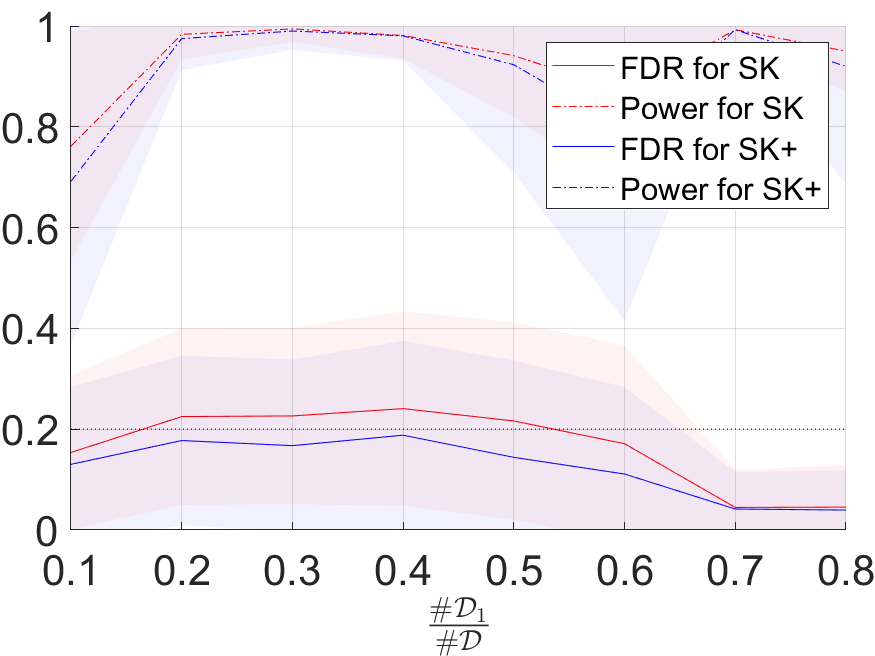}
\end{minipage}%
}%
\subfigure[$\Wst$ in $D_2$]{
\begin{minipage}[t]{0.3\textwidth}
\centering
\includegraphics[width=\textwidth]{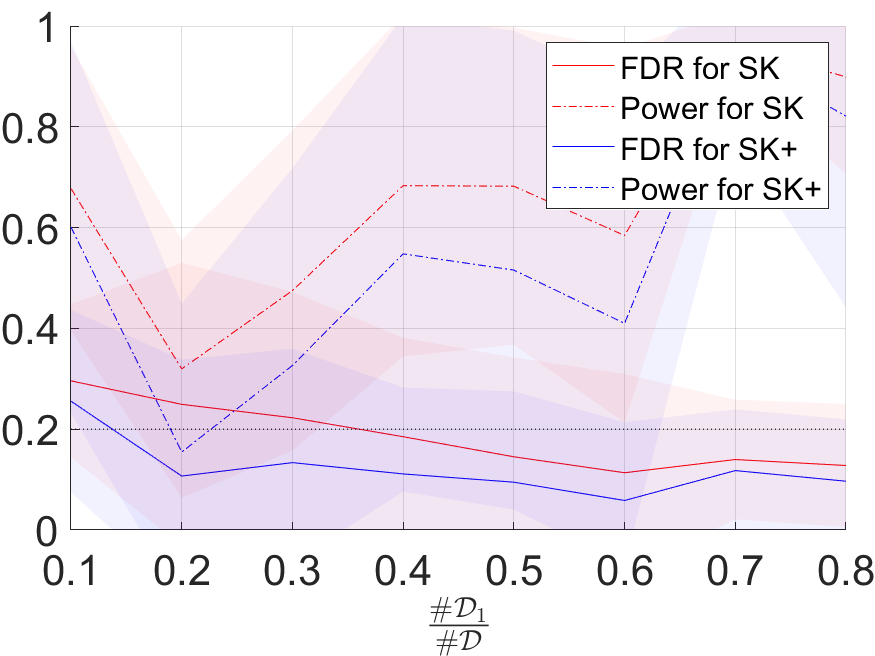}
\end{minipage}%
}%
\subfigure[$\Wst$ in $D_3$]{
\begin{minipage}[t]{0.3\textwidth}
\centering
\includegraphics[width=\textwidth]{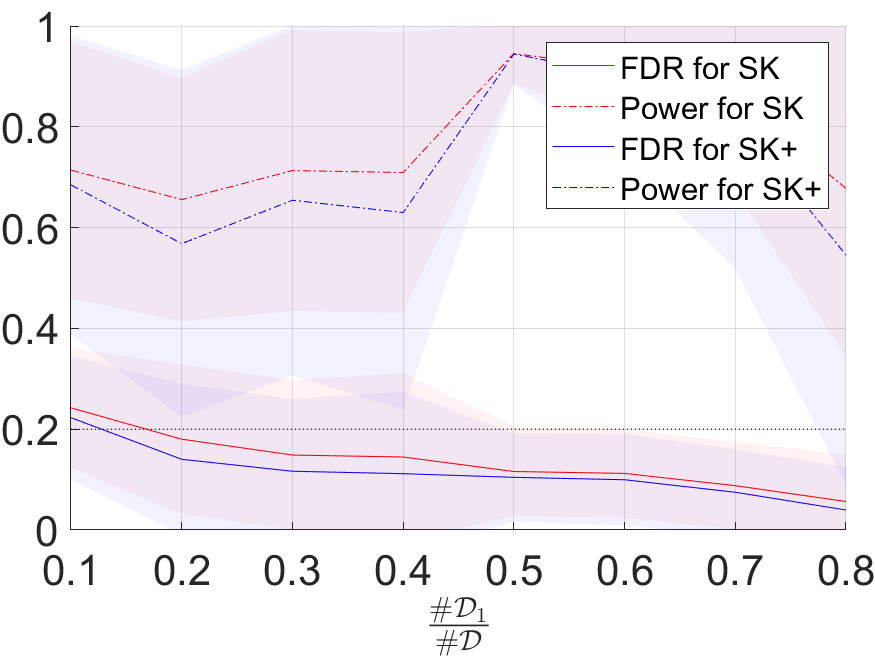}
\end{minipage}%
}%

\centering
\subfigure[$\Wbc$ in $D_1$]{
\begin{minipage}[t]{0.3\textwidth}
\centering
\includegraphics[width=\textwidth]{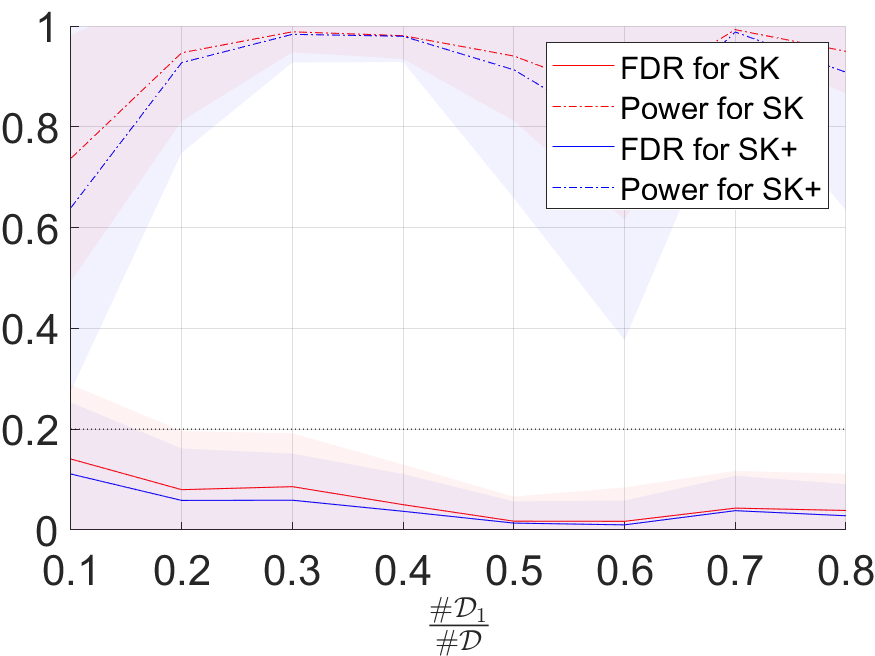}
\end{minipage}%
}%
\subfigure[$\Wbc$ in $D_2$]{
\begin{minipage}[t]{0.3\textwidth}
\centering
\includegraphics[width=\textwidth]{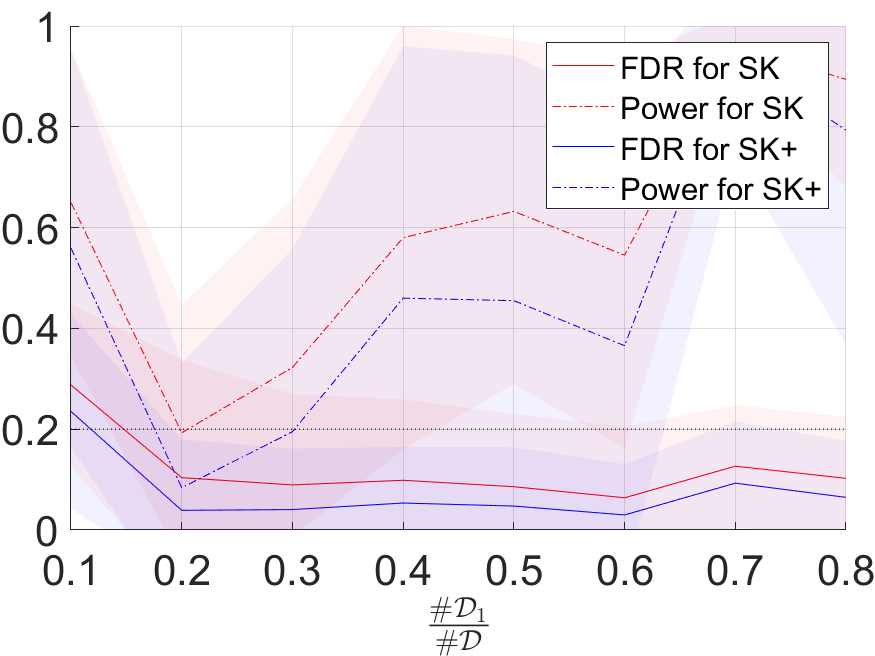}
\end{minipage}%
}%
\subfigure[$\Wbc$ in $D_3$]{
\begin{minipage}[t]{0.3\textwidth}
\centering
\includegraphics[width=\textwidth]{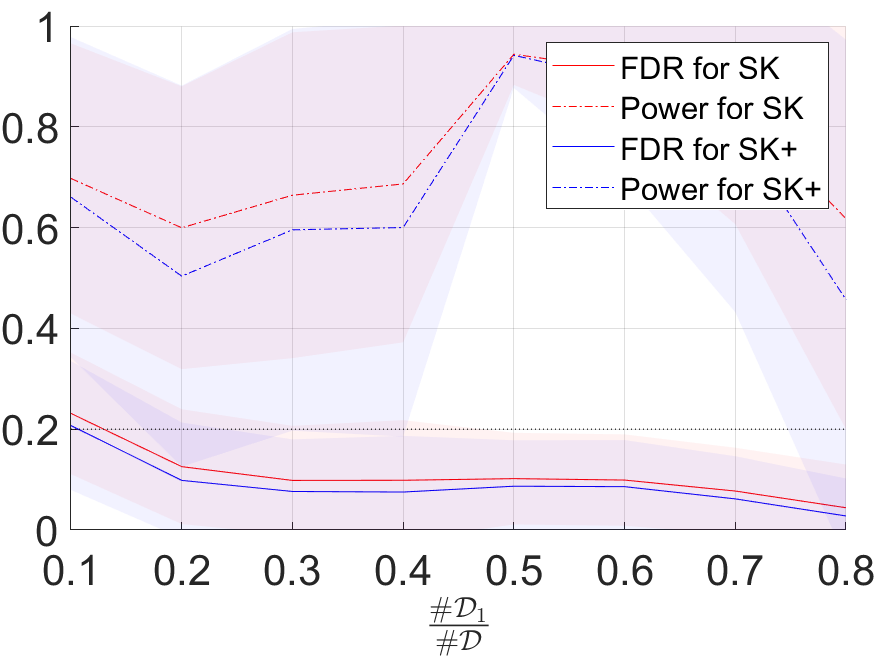}
\end{minipage}%
}%

\caption{The performance of Split Knockoffs under different data splitting fractions: FDR and Power for $q=0.2$. \textbf{The signal strength is taken to be 0.5 as in Equation \eqref{eq: snr = 0.5} for this experiment instead of being 1 in Section \ref{sec: simulation_settings}.} The $x$-axis represents the percentage of samples in $\D_1$ among $\D$. The curves in the figures represent the average performance of Split Knockoffs in FDR and Power in 200 simulation instances, while the shaded areas represent the 80\% confidence intervals truncated to the range $[0, 1]$.  The intercept $\widehat{\beta}(\lambda)$ for Split Knockoffs is taken as a fixed cross validation optimal estimator $\widehat{\beta}_{\hat\nu, \hat{\lambda}}$, with the $\nu$ for calculating the feature and knockoff significance taken as $\hat\nu$.}
\label{fig: split size 0.5}
\end{figure}

In particular, since we extend Split Knockoffs to high dimensional settings in Section \ref{sec: hd}, our method is still applicable for the case where the data splitting fraction of $\D_1$ is large and $n_2\ge m+ p$ is dissatisfied. As one can observe in Figure \ref{fig: split size 0.5}, the FDR of Split Knockoffs is not sensitive to the data splitting fraction --- in all cases, Split Knockoffs achieves desired FDR despite its potential theoretical inflation when the data splitting fraction of $\D_1$ is large. The main effect of the data splitting fraction lies on the power of Split Knockoffs. In Figure \ref{fig: split size 0.5}, the optimal selection power is commonly achieved when the data splitting fraction is balanced for $\D_1$ and $\D_2$. Therefore, a practical guideline for selecting the data splitting fraction should be: make the data splitting fraction balanced when the sample size permits, and favor the sample size of $\D_1$ a little bit for better selection power when the sample size is limited.

\subsection{Simulation Experiments on Stability of Data Splitting}

\begin{figure}[!ht]
\centering
\subfigure[$\Ws$ in $D_1$]{
\begin{minipage}[t]{0.3\textwidth}
\centering
\includegraphics[width=\textwidth]{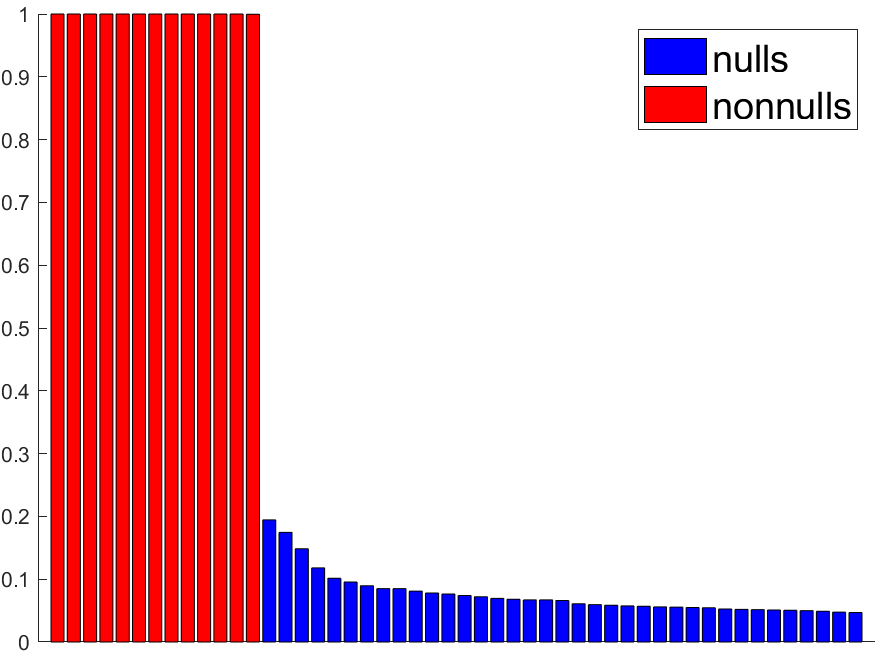}
\end{minipage}%
}%
\subfigure[$\Ws$ in $D_2$]{
\begin{minipage}[t]{0.3\textwidth}
\centering
\includegraphics[width=\textwidth]{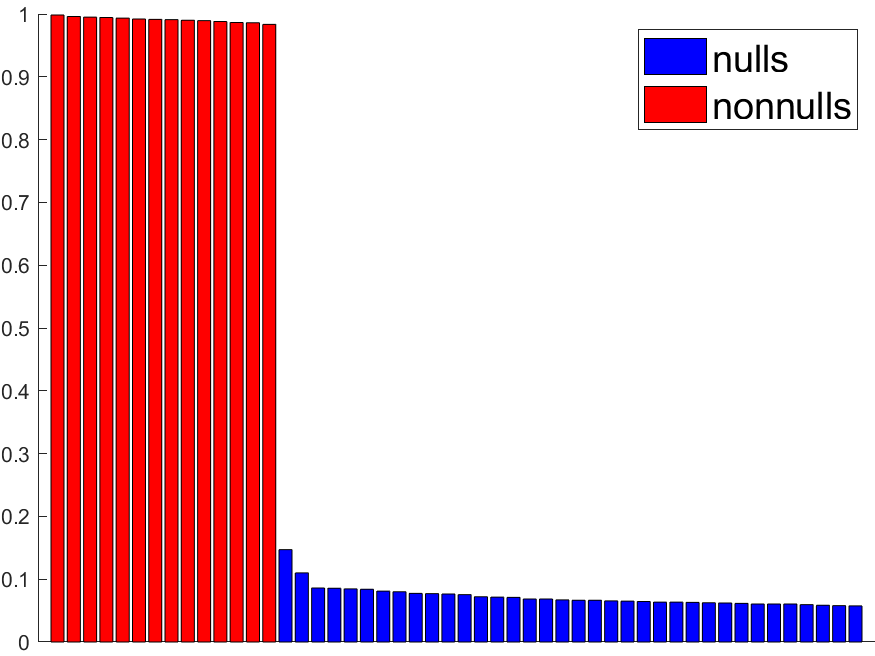}
\end{minipage}%
}%
\subfigure[$\Ws$ in $D_3$]{
\begin{minipage}[t]{0.3\textwidth}
\centering
\includegraphics[width=\textwidth]{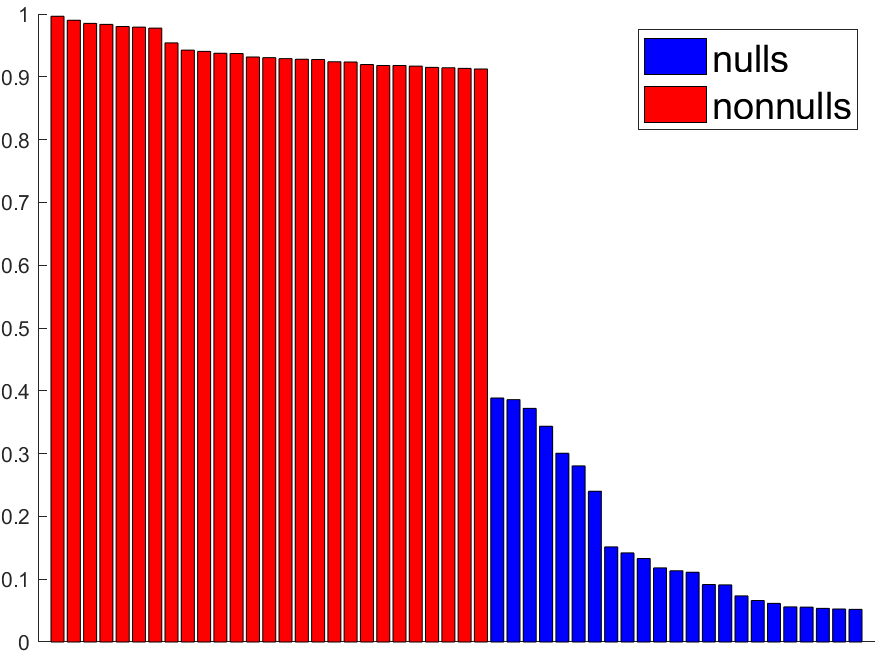}
\end{minipage}%
}%

\centering
\subfigure[$\Wst$ in $D_1$]{
\begin{minipage}[t]{0.3\textwidth}
\centering
\includegraphics[width=\textwidth]{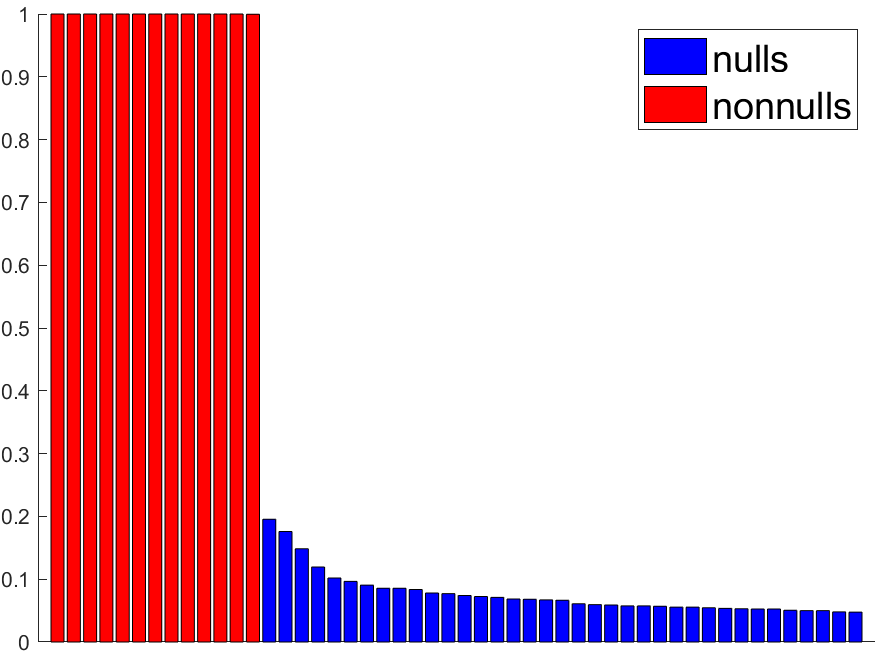}
\end{minipage}%
}%
\subfigure[$\Wst$ in $D_2$]{
\begin{minipage}[t]{0.3\textwidth}
\centering
\includegraphics[width=\textwidth]{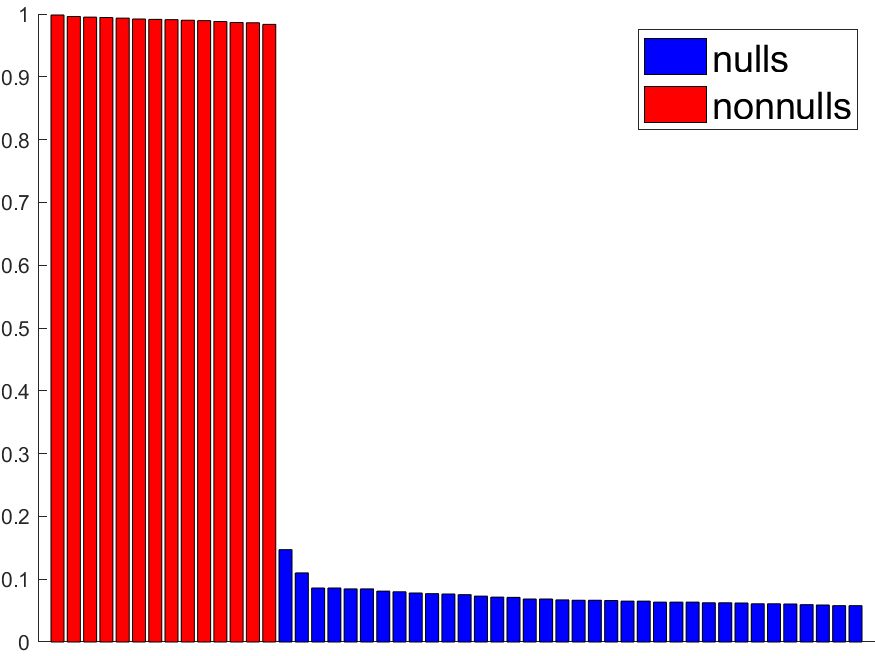}
\end{minipage}%
}%
\subfigure[$\Wst$ in $D_3$]{
\begin{minipage}[t]{0.3\textwidth}
\centering
\includegraphics[width=\textwidth]{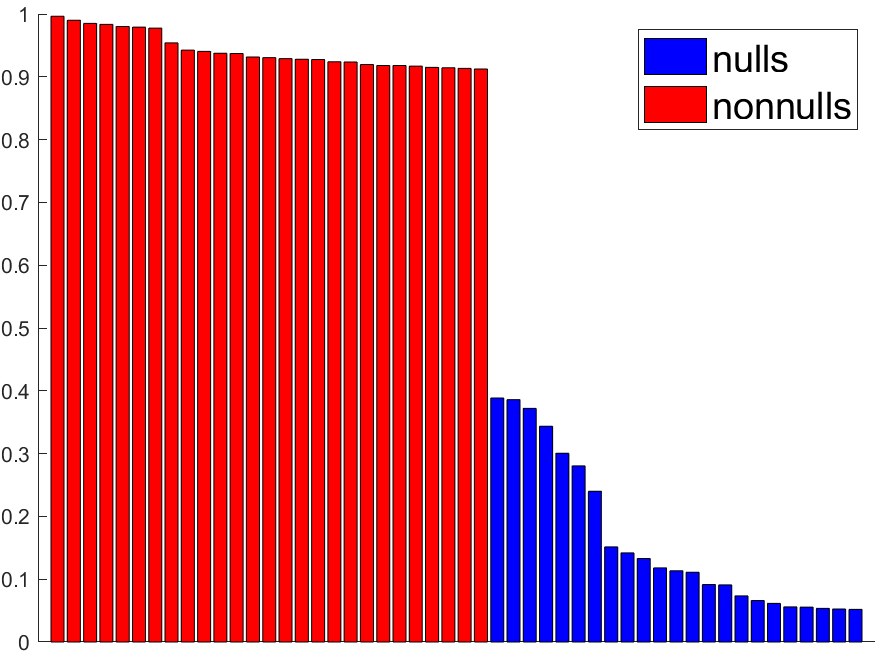}
\end{minipage}%
}%

\centering
\subfigure[$\Wbc$ in $D_1$]{
\begin{minipage}[t]{0.3\textwidth}
\centering
\includegraphics[width=\textwidth]{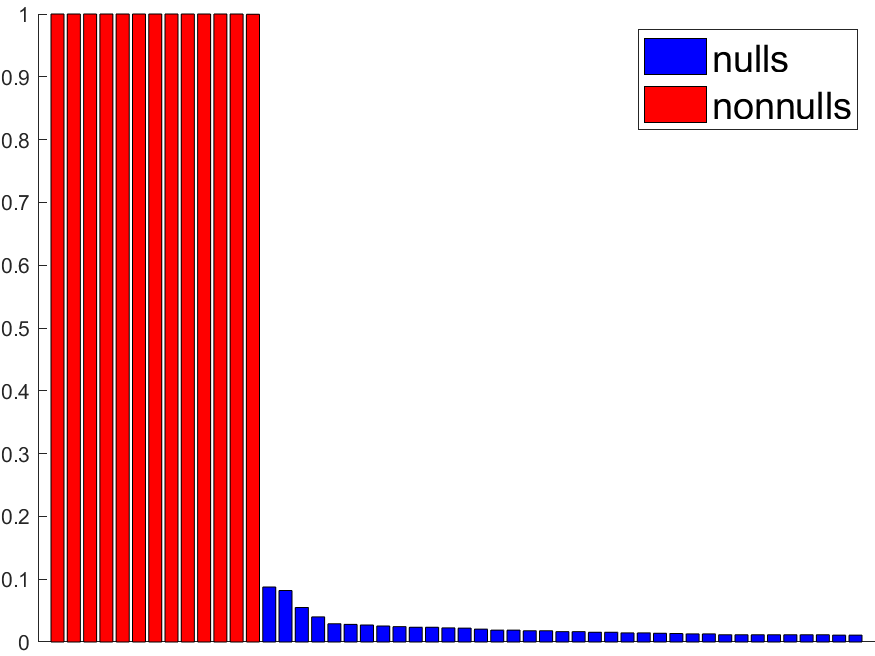}
\end{minipage}%
}%
\subfigure[$\Wbc$ in $D_2$]{
\begin{minipage}[t]{0.3\textwidth}
\centering
\includegraphics[width=\textwidth]{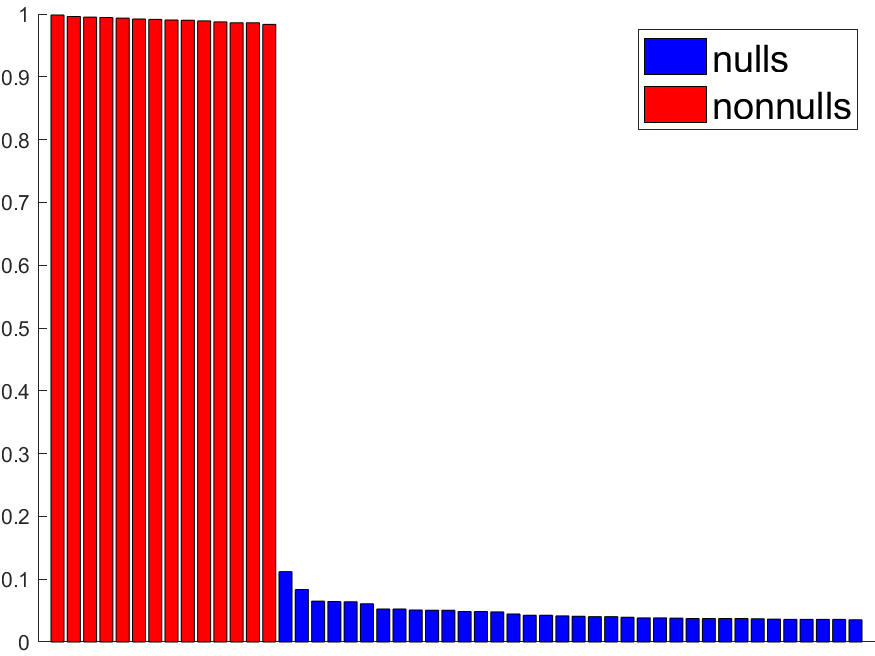}
\end{minipage}%
}%
\subfigure[$\Wbc$ in $D_3$]{
\begin{minipage}[t]{0.3\textwidth}
\centering
\includegraphics[width=\textwidth]{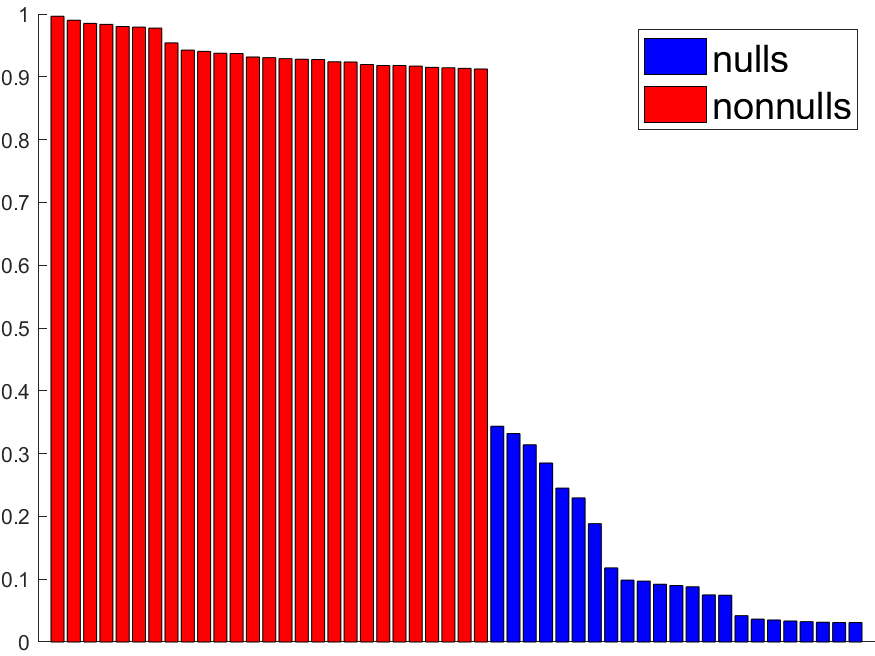}
\end{minipage}%
}%

\caption{Top 50 most frequently selected features of Split Knockoffs  in simulation experiments among 100 different data splits. Each bar in the figure represents the selection frequency of a particular feature. The intercept $\widehat{\beta}(\lambda)$ for Split Knockoffs is taken as a fixed cross validation optimal estimator $\widehat{\beta}_{\hat\nu, \hat{\lambda}}$, with the $\nu$ for calculating the feature and knockoff significance taken as $\hat\nu$. }
\label{fig: robust}
\end{figure}

In order to handle the transformational sparsity, Split Knockoffs introduce a random data splitting scheme in Section \ref{sec:splitknockoff}. Consequently, the random data splitting will lead to the random estimated support sets. 
In this section, we will use simulation experiments to study the robustness on features in the estimated support sets selected by Split Knockoffs in random data splitting. 
We succeed all the simulation settings on the dataset $(X, y)$ and transformational sparsity $D_1$, $D_2$, $D_3$ from Section \ref{sec: simulation_settings}.  
In the experiment, we randomly generate the dataset $(X, y)$ for 20 times. For each generation, we randomly perform 100 different data splits and conduct Split Knockoffs with each data split. The estimated support sets of Split Knockoffs are recorded for the $20\times 100$ tests.

In Figure \ref{fig: robust}, we present the selection frequencies of the features selected by Split Knockoffs in the $20\times 100$ tests.
For Split Knockoffs conducted in this section, we take $\widehat{\beta}(\lambda)$ as a fixed cross validation optimal estimator $\widehat{\beta}_{\hat\nu, \hat{\lambda}}$, and take the parameter $\nu$ for calculation of the feature and knockoff significance as $\hat\nu$.

As shown in Figure \ref{fig: robust}, in all cases, the nonnull features will have relatively stable and high frequencies to be selected, while the null features occur randomly and have much lower selection frequencies. Such a result suggests that the nonnull features can be robustly selected in the selection set of Split Knockoffs  with respect to random data splitting.

\subsection{Simulation Experiments on the Signal Strength}

\label{sec: snr change}
In this section, we implement simulation experiments to compare the performance of Split Knockoffs and Knockoffs when the signal strength varies. We succeed all simulation settings in Section \ref{sec: simulation_settings}, except that we take $\beta^*\in\mathbb R^p$ as
\begin{equation*}
    \beta_i^*:=\left\{
    \begin{array}{ccl}
        A   &   & i \le 20,\ i \equiv 0, -1 (\mathrm{mod}\ 3),\\
        0   &   & \mathrm{otherwise},
    \end{array} \right.
\end{equation*}
where $\log(A)$ varies from -0.5 to 0.5. The performance comparisons between Split Knockoffs and Knockoffs when the signal strength varies are presented in Figure \ref{fig: snr change main}. 

\begin{figure}[!ht]
\centering
\subfigure[(Split) Knockoff in $D_1$]{
\begin{minipage}[t]{0.3\textwidth}
\centering
\includegraphics[width=\textwidth]{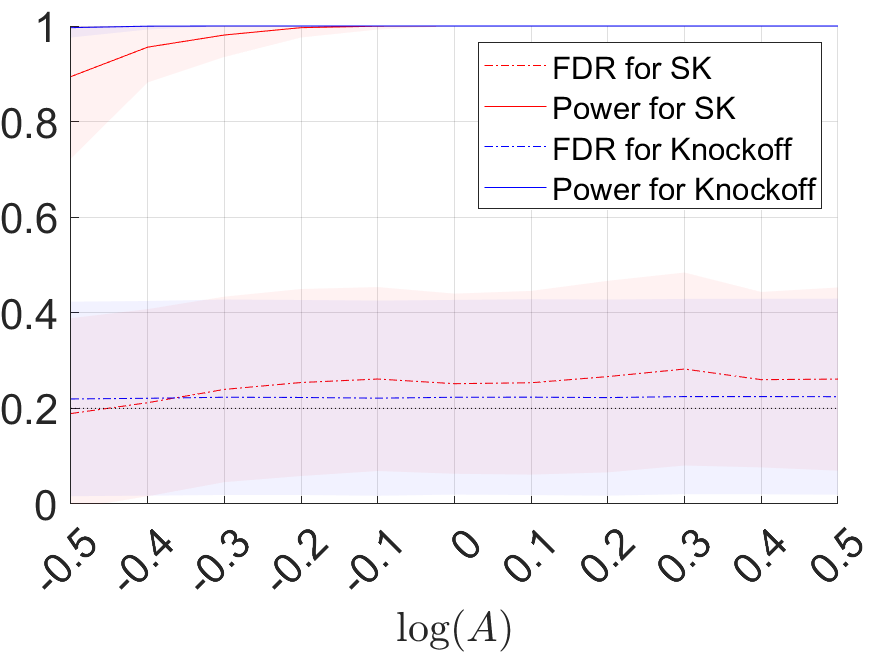}
\end{minipage}%
}%
\subfigure[(Split) Knockoff in $D_2$]{
\begin{minipage}[t]{0.3\textwidth}
\centering
\includegraphics[width=\textwidth]{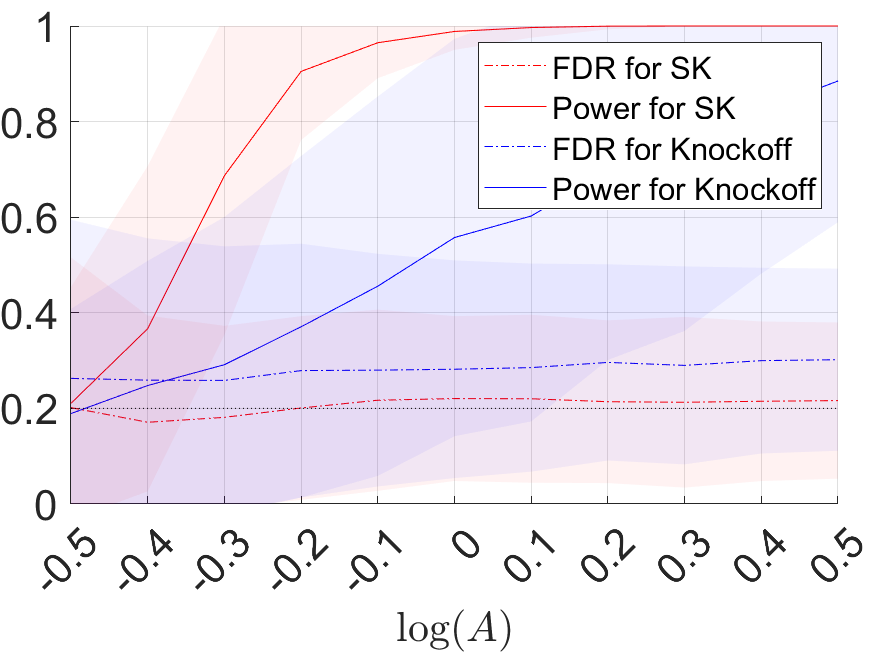}
\end{minipage}%
}%
\subfigure[Split Knockoff in $D_3$]{
\begin{minipage}[t]{0.3\textwidth}
\centering
\includegraphics[width=\textwidth]{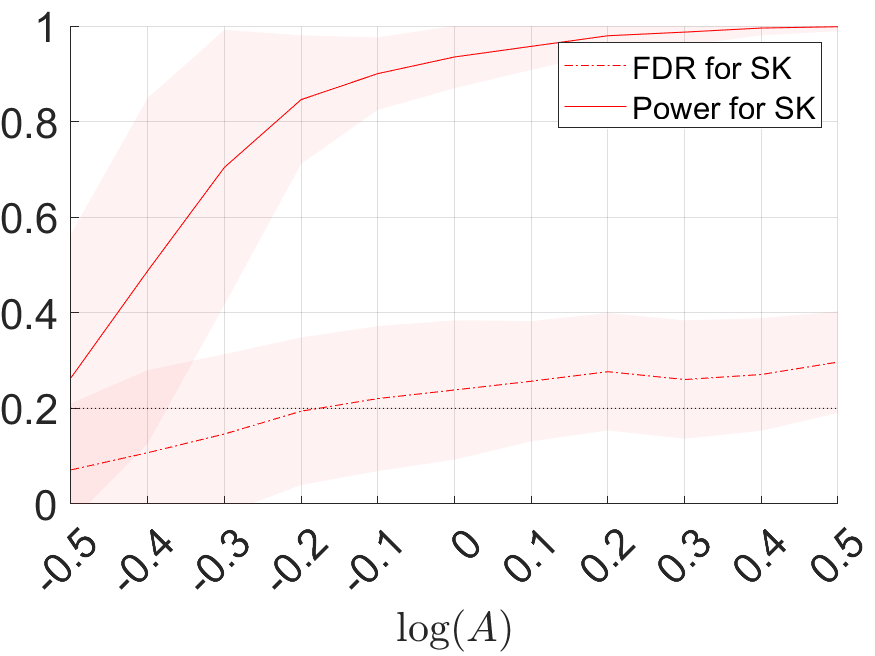}
\end{minipage}%
}%

\centering
\subfigure[(Split) Knockoff+ in $D_1$]{
\begin{minipage}[t]{0.3\textwidth}
\centering
\includegraphics[width=\textwidth]{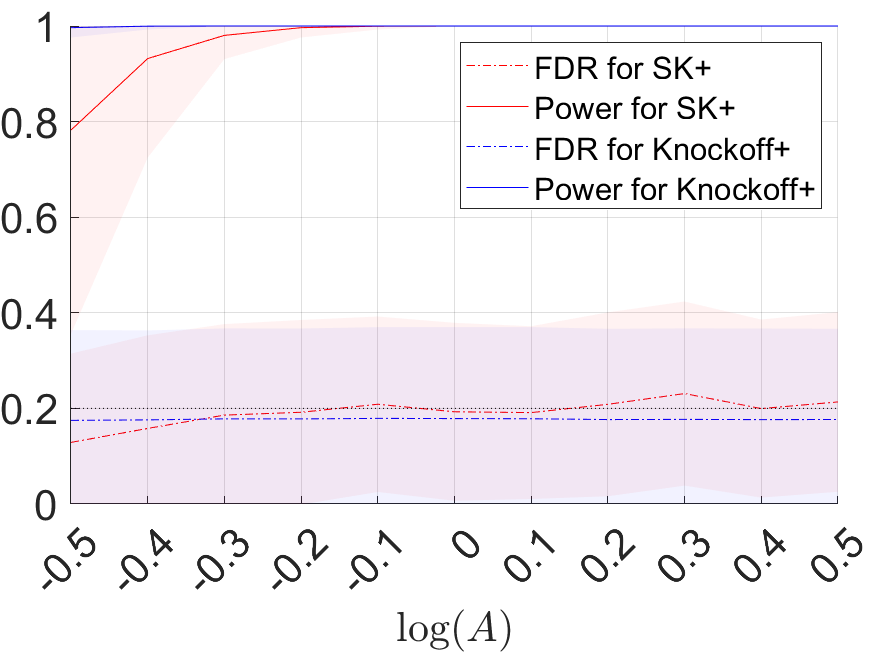}
\end{minipage}%
}%
\subfigure[(Split) Knockoff+ in $D_2$]{
\begin{minipage}[t]{0.3\textwidth}
\centering
\includegraphics[width=\textwidth]{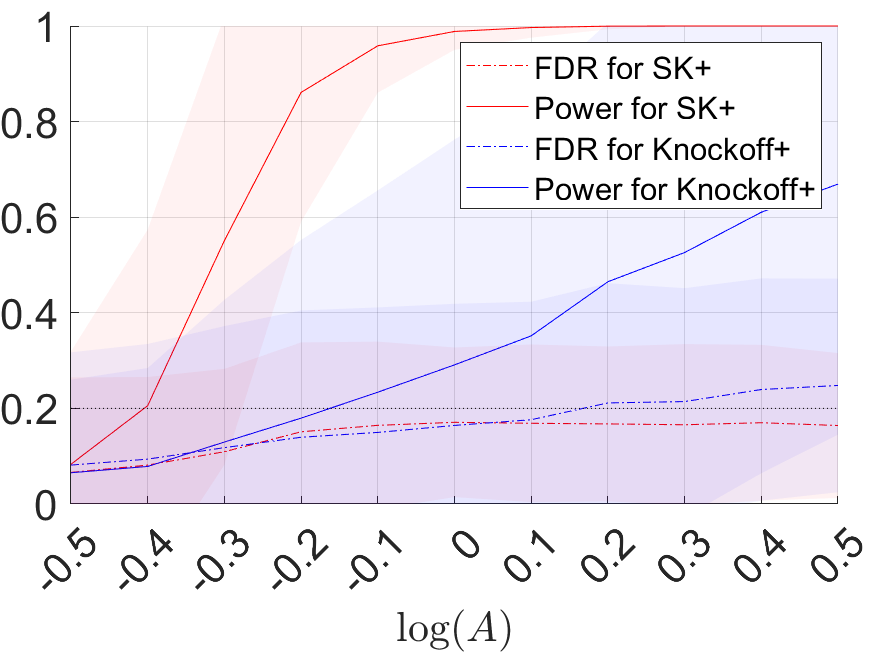}
\end{minipage}%
}%
\subfigure[Split Knockoff+ in $D_3$]{
\begin{minipage}[t]{0.3\textwidth}
\centering
\includegraphics[width=\textwidth]{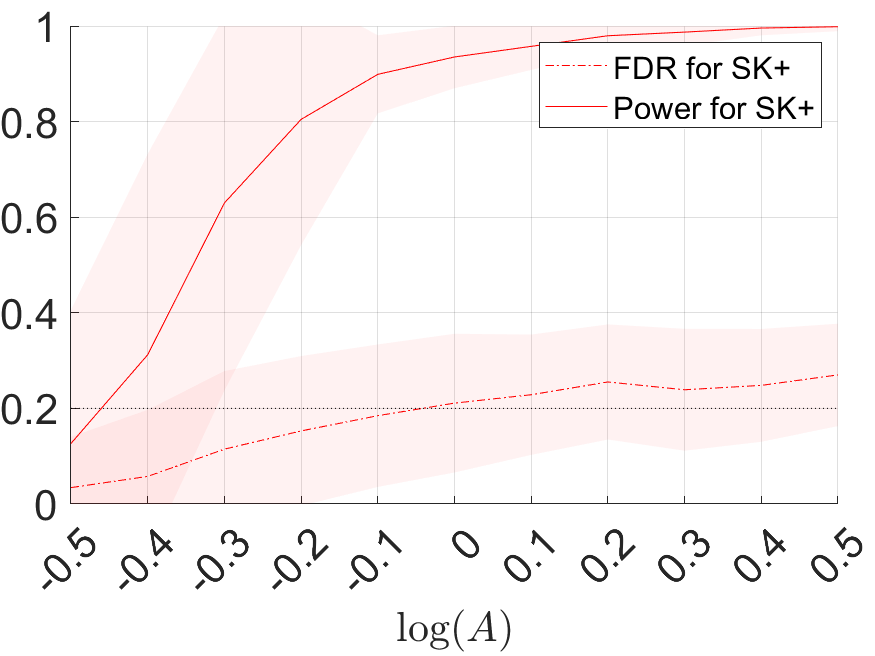}
\end{minipage}%
}%

\caption{The performance of Split Knockoffs and Knockoffs under different signal noise ratios: FDR and Power for $q=0.2$. The curves in the figures represent the average performance in FDR and Power in 200 simulation instances, while the shaded areas represent the 80\% confidence intervals truncated to the range $[0, 1]$.  The intercept $\widehat{\beta}(\lambda)$ for Split Knockoffs is taken as a fixed cross validation optimal estimator $\widehat{\beta}_{\hat\nu, \hat{\lambda}}$, with the $\nu$ for calculating the feature and knockoff significance taken as $\hat\nu$. We take the $W$ statistics of Split Knockoffs as $\Wst$ in this figure.}
\label{fig: snr change main}
\end{figure}

As presented in Figure \ref{fig: snr change main}, both Split Knockoffs and Knockoffs achieve desired FDR control when applicable. 

As for the selection power, in the case that the transformation is the identity matrix ($D_1$), Split Knockoffs exhibit lower selection power compared with Knockoffs when the signal strength is weak. In this case, the potential power improvement from the weaker $\nu$-incoherence condition in Proposition \ref{thm: sign consistency front} is overwhelmed by the dissatisfaction of the minimal signal strength condition \eqref{eq: min snr}. This leads to a loss in the selection power of Split Knockoffs.

On the other hand, when the linear transformation is nontrivial such as $D_2$, the power improvement brought by the $\nu$-incoherence condition in Proposition \ref{thm: sign consistency front} dominants the minimal signal strength condition \eqref{eq: min snr}. This enables Split Knockoffs to achieve higher selection power compared with Knockoffs.

\subsection{Analysis on Computational Cost}

In this section, we will give an analysis on the computational cost of Split Knockoffs. Among the procedures of Split Knockoffs, the computation of the two regularization paths in Equation \eqref{eq:gamma} and \eqref{eq:t_gamma} takes the most computational time. In particular, the computational cost on regularization paths varies for different choices of $\widehat{\beta}(\lambda)$. In the case that $\widehat{\beta}(\lambda)$ is a fixed point, the computational cost will be twice the cost of the generalized LASSO in the \url{glmnet} package \citep{friedman2010regularization, simon2011regularization}. However, in the case that $\widehat{\beta}(\lambda)$ is a solution path that changes with respect to $\lambda$, an additional cost of the generalized LASSO is spent; in this regard, the computational cost will be three times the cost of the generalized LASSO.

\section{Supplementary Material for Alzheimer's Disease Experiments}
\label{sec: freq}

In this section, we provide more details regarding the experiments on Alzheimer's Disease. In particular, we provide supplementary material for experiments in Alzheimer's Disease in Section \ref{Sec: applications}, where we plot the frequencies of the most frequently selected regions by Split Knockoffs in multiple data splits. 
To be specific, we show the selection results of Split Knockoff under 100 different random sample splits of the dataset in Alzheimer's Disease. We take the $W$ statistics for Split Knockoff as $\Wst$, with $\widehat{\beta}(\lambda)$ being taken as a cross validation optimal estimator $\widehat{\beta}_{\hat\nu, \hat{\lambda}}$, and $\nu$ being taken as $\hat\nu$. Throughout the random splitting, important regions and connections should be selected with high frequencies, while the false discoveries in each selection should have relatively low frequencies to be selected. We plot the top 10 most frequently selected regions/connections by Split Knockoff in 100 random data splits in Figure \ref{fig: ad regions} and Figure \ref{fig: ad connections}, with abbreviations of each region marked in the figures. A comparison table between the full region names and their abbreviations can be found in Table \ref{tab:name of region}. 

\begin{figure}[ht]
\centering
\includegraphics[width=0.75\textwidth]{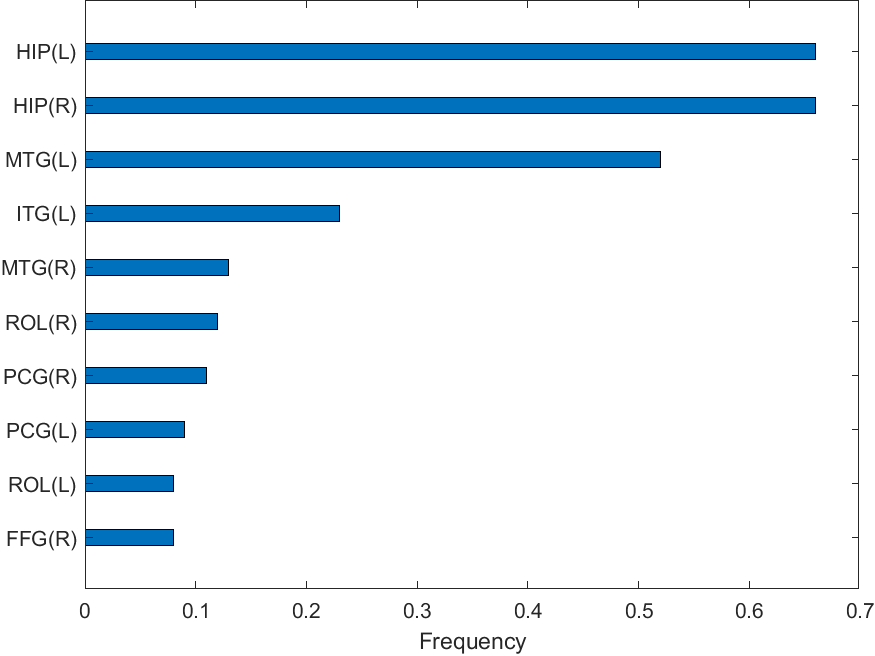}
\caption{Top 10 most frequently selected lesion regions by Split Knockoff in 100 random data splits. We take the $W$ statistics in Split Knockoff as $\Wst$, with $\widehat{\beta}(\lambda)$ being taken as a fixed cross validation optimal estimator $\widehat{\beta}_{\hat\nu, \hat{\lambda}}$, and $\nu$ being taken as $\hat\nu$.}
\label{fig: ad regions}
\end{figure}

\begin{figure}[ht]
\centering
\includegraphics[width=0.75\textwidth]{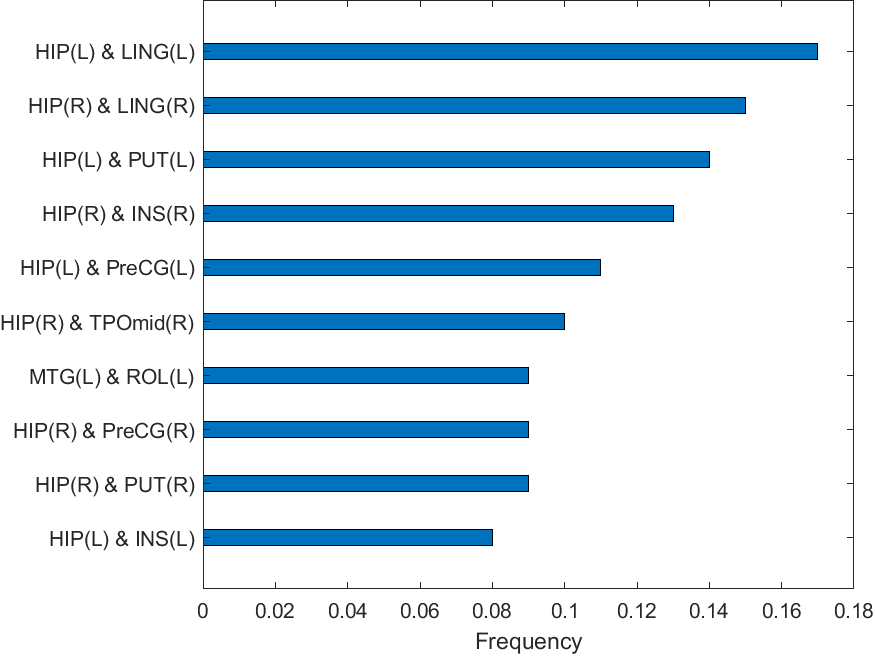}
\caption{Top 10 most frequently selected connections of adjacent regions by Split Knockoff in 100 random data splits.  We choose the $W$ statistics for Split Knockoff as $\Wst$, with $\widehat{\beta}(\lambda)$ being taken as a fixed cross validation optimal estimator $\widehat{\beta}_{\hat\nu, \hat{\lambda}}$, and $\nu$ being taken as $\hat\nu$.}
\label{fig: ad connections}
\end{figure}

\begin{table}[!ht]
    \caption{Names and Abbreviations for Cerebrum Brain Anatomical Regions}
    \label{tab:name of region}
    \centering
    \footnotesize{
    \begin{tabular}{|c|c|}
    \hline
Region Name & Abbreviation\\
\hline
Precental gyrus & PreCG\\
 Superior frontal gyrus,  dorsolateral  & SFGdor\\
 Superior frontal gyrus,  orbital part  & ORBsup\\
Middle frontal gyrus & MFG\\
 Middle frontal gyrus,  orbital part  & ORBmid\\
 Inferior frontal gyrus,  opercular part  & IFGoperc\\
 Inferior frontal gyrus,  triangular part  & IFGtriang\\
 Inferior frontal gyrus,  orbital part  & ORBinf\\
Rolandic operculum & ROL\\
Supplementary motor area & SMA\\
Olfactory cortex & OLF\\
 Superior frontal gyrus,  medial  & SFGmed\\
 Superior frontal gyrus,  medial orbital  & ORBsupmed\\
Gyrus rectus & REC\\
Insula & INS\\
Anterior cingulate and paracingulate gyri & ACG\\
Median cingulate and paracingulate gyri & MCG\\
Posterior cingulate gyrus & PCG\\
Hippocampus & HIP\\
Parahippocampal gyrus & PHG\\
Amygdala & AMYG\\
Calcarine fissure and surrounding cortex & CAL\\
Cuneus & CUN\\
Lingual gyrus & LING\\
Superior occipital gyrus & SOG\\
Middle occipital gyrus & MOG\\
Inferior occipital gyrus & IOG\\
Fusiform gyrus & FFG\\
Postcentral gyrus & PoCG\\
Superior parietal gyrus & SPG\\
 Inferior parietal,  but supramarginal and angular gyri  & IPL\\
Supramarginal gyrus & SMG\\
Angular gyrus & ANG\\
Precuneus & PCUN\\
Paracentral lobule & PCL\\
Caudate nucleus & CAU\\
Lenticular nucleus putamen & PUT\\
 Lenticular nucleus,  pallidum  & PAL\\
Thalamus & THA\\
Heschl gyrus & HES\\
Superior temporal gyrus & STG\\
Temporal pole: superior temporal gyrus & TPOsup\\
Middle temporal gyrus & MTG\\
Temporal pole: middle temporal gyrus & TPOmid\\
Inferior temporal gyrus & ITG\\
\hline
    \end{tabular}
    }
\end{table}

\end{document}